%% file: main.tex
\documentclass{lmcs} 
\usepackage{hyperref}
\usepackage{cleveref}
\usepackage[all]{xy}
\usepackage{amssymb}	
\usepackage{cmll}    	
\usepackage{stmaryrd}   
\usepackage{Guyboxes}
\usepackage{mathpartir}
\usepackage{mathrsfs}   
\usepackage{thmtools}
\usepackage{thm-restate}

\usepackage{tikz-cd}

\input{macros}

\begin{document}

\title[The Quantitative Collapse of Concurrent Games with Symmetry]{The Quantitative
Collapse\\of Concurrent Games with Symmetry}

\author[P. Clairambault]{Pierre Clairambault}	
\address{Univ Lyon, EnsL, UCBL, CNRS,  LIP, F-69342, LYON Cedex 07, France}	
\email{pierre.clairambault@ens-lyon.fr}  

\author[H. Paquet]{Hugo Paquet}
\address{University of Oxford}
\email{hugo.paquet@cs.ox.ac.uk}

\newcommand{\changed}[1]{{#1}}

\begin{abstract}
We explore links between the thin concurrent games of
Castellan, Clairambault and Winskel, and the weighted relational models of linear logic studied by 
Laird, Manzonetto, McCusker and Pagani.
More precisely, we show that there is an interpretation-preserving ``collapse'' functor from the former to the latter. On objects, the functor defines for each game a set of possible \emph{execution states}. Defining the action on morphisms is more subtle, and this is the main contribution of the paper. 

Given a strategy and an execution state, our functor needs to \emph{count} the witnesses for this state within the strategy. Strategies in thin concurrent games describe non-linear behaviour explicitly, so in general each witness exists in countably many symmetric \emph{copies}. The challenge is to define the right notion of witnesses, factoring out this infinity while matching the weighted relational model. Understanding how witnesses compose is particularly subtle and requires a delve into the combinatorics of witnesses and their symmetries. 

In its basic form, this functor connects thin concurrent games and a relational model weighted by $\mathbb{N} \cup \{+\infty\}$. We will additionally consider a generalised setting where both models are weighted by elements of an arbitrary continuous semiring; this covers the probabilistic case, among others. Witnesses now additionally carry a value from the semiring, and our interpretation-preserving collapse functor extends to this setting. 
\end{abstract}

\maketitle

\section{Introduction}

The \emph{relational model} is one of the simplest model of linear logic. It naturally gives rise to a model of higher-order programming often described as \emph{quantitative}, because aspects of computation such as the multiplicity of function calls are represented explicitly. The model assigns to every type a set known as its
\emph{web}, whose elements are thought of as (desequentialized)
execution states, and to any term, a relation. 


Relations are equivalently boolean-valued matrices, and a natural
extension of the model consists in considering more general
coefficients. This idea was first explored by Lamarche
\cite{DBLP:journals/tcs/Lamarche92} and developed in detail by Laird,
Manzonetto, McCusker
and Pagani \cite{DBLP:conf/lics/LairdMMP13,DBLP:journals/iandc/Laird20}. Their construction gives a family of \emph{weighted relational models}, in which the interpretation of a term is a matrix assigning to each point of the web a \emph{weight}, coming from a continuous semiring $\R$, the \emph{resource semiring}. This has an operational interpretation for a program $M$: while the relational model only indicates whether a given state $\alpha$ \emph{can} be realized by an execution of $M$ (\emph{i.e.} do we have $\alpha \in \intr{M}$?), the weighted
relational model aggregates information about all executions that realize this state. In the simplest case, $\R = \mathbb{N}\cup \{+\infty\}$, the model simply \emph{counts} these executions. Using various
semirings one can adequately represent probabilistic
evaluation, best and worst case time analysis, etc. 



Another well-established quantitative model -- or rather, family of
models -- is game semantics \cite{ho,ajm}. In game semantics, 
an execution is regarded as a play in a
two-player game between Player (playing for the program), and Opponent
(playing for the environment). Types are presented as games, whose
rules specify the possible executions, and terms as strategies
describing the interactive behaviour of the program under any evaluation
context. The connections between game semantics and
relational semantics have been thoroughly studied
\cite{DBLP:conf/csl/BaillotDER97,DBLP:conf/lics/Mellies05,DBLP:conf/tlca/Boudes09,DBLP:conf/lics/Ong17}.
In particular, the family of \emph{concurrent games}
\cite{DBLP:conf/lics/RideauW11,cg2} inherit from Melli\`es'
\emph{asynchronous games} \cite{DBLP:conf/lics/Mellies05} a particularly
neat relationship with relational semantics. In this framework, both games
and strategies are \emph{event structures}
\cite{DBLP:conf/ac/Winskel86}, and as such admit a canonical notion of
state/position: the \emph{configurations}. As we will see, the web can be recovered as a subset
of the configurations of the game. Then, we can ``collapse'' a strategy into a relation, by recording which of these configurations are reached, and forgetting the chronological history (see
\cite{mall} for a recent account). In an affine setting, i.e. without replication, this collapse operation can immediately be generalised to weighted relations with $\R = \mathbb{N} \cup \{+\infty\}$: if $\sigma$ is a strategy on a game $A$, and $x$ is (a point of the web
corresponding to) a configuration of $A$, the collapse simply \emph{counts} the distinct configurations of $\sigma$ realizing $x$. 

The difficulty arises in the non-affine setting, necessary for
languages with duplication of resources. For this, the mature extension of concurrent
games is \emph{thin concurrent games with symmetry} \cite{cg2}. In thin
concurrent games, infinite games arise from
the construction $\oc A$, which creates countably many copies of $A$ labelled with natural numbers, called \emph{copy indices}. (This is similar to the situation in AJM games \cite{ajm}.) Games are
equipped with sets of bijections called \emph{symmetries} (so they are
\emph{event structures with symmetry}
\cite{DBLP:journals/entcs/Winskel07}) which specify authorized
reindexings. Additionally strategies must act uniformly with respect to
these symmetries. The collapse to the relational model is relatively
undisturbed by symmetry: points of the web now correspond to a subset
of \emph{symmetry classes} of configurations and, as before, we
collapse a strategy to the set of symmetry classes it reaches (see \emph{e.g.} \cite{cg3}). 
However the extension to the \emph{weighted} relational model is no longer obvious. We cannot simply count all concrete configurations of
$\sigma$ witnessing some symmetry class: there are infinitely many. This prompts the central question of this paper: how can we count configurations up to symmetry, in order to match coefficients of the weighted relational model? In other words: what does the weighted relational model count?

An answer to this question is our main contribution. For a symmetry class of configurations of the game, an apparent ``obvious'' solution is to consider the set of corresponding \emph{symmetry classes} of configurations of the strategy. Suprisingly, the induced coefficient is wrong! Instead we are led to introduce a notion of
\emph{positive witnesses} for a given symmetry class of configurations of the
game. We will show that counting positive witnesses yields an interpretation-preserving collapse to the
relational model weighted by $\mathbb{N} \cup \{+\infty\}$. In proving 
this, the main
challenge is functoriality of the collapse: whereas in the affine
case, witnesses in a composite $\tau \odot \sigma$ cleanly correspond
to pairs of witnesses in $\sigma$ and $\tau$, this
fails with symmetry, and is only salvaged via a proper account of symmetries on both sides. 

Finally, we also 
extend our results to any continuous semiring $\R$ with a condition
called \emph{integer
division}. We consider an extension of thin concurrent games where strategies carry valuations
in $\R$, and show that this collapses into the $\R$-weighted relational model.

\paragraph{Related work} To our knowledge, the first quantitative
collapse from games to relations is from
\emph{probabilistic} concurrent games to the relational model weighted by
$\overline{\mathbb{R}}_+ = \mathbb{R} \cup \{+\infty\}$ in
\cite{DBLP:conf/lics/CastellanCPW18}. However, when working on an
extension to a quantum language
\cite{DBLP:journals/pacmpl/ClairambaultV20}, the first author discovered
an error in \cite{DBLP:conf/lics/CastellanCPW18}: the paper uses
symmetry classes as witnesses, which  \textemdash as we show here\textemdash{} is
inadequate. The correct notion of witness and its validity \emph{w.r.t.}
composition was established by the first author in an unpublished
report \cite{DBLP:journals/corr/abs-2006-05080}. Here we
complete this to a full interpretation-preserving functor and to the
$\R$-weighted case. In particular, Theorem \ref{thm:main4} for $\R =
\overline{\mathbb{R}}_+$ corrects the collapse theorem of
\cite{DBLP:conf/lics/CastellanCPW18}.

\paragraph{Outline} In Section \ref{sec:wrel_npcf} we recall the
$\R$-weighted relational model and the language of concern for most of
the paper, a non-deterministic $\PCF$. In Section \ref{sec:cg_pcf}, we
recall \emph{thin concurrent games} and the corresponding interpretation of
non-deterministic $\PCF$. In Section \ref{sec:wit_comp} we address the
main challenge of the paper, the definition of positive witnesses
and their compatibility with composition. In Section \ref{sec:pres_intr}, we fix $\R = \mathbb{N} \cup \{+ \infty\}$, show a number of properties ensuring that the interpretation is preserved, and prove our main result (Theorem
\ref{thm:main2}). Finally, in Section \ref{sec:coll_rw} we generalize the result to an arbitrary $\R$ (Theorem
\ref{thm:main4}).

\section{The Weighted Relational Model and $\nPCF$}
\label{sec:wrel_npcf}

\paragraph{Notations.} If $X$ is a set, we write $\mathcal{P}_f(X)$ for
the finite subsets, and $\M_f(X)$ for finite
multisets. For $x_1, \dots, x_n \in X$, $[x_1, \dots, x_n] \in
\M_f(X)$ is the corresponding multiset. We use $\mu, \nu \in \M_f(X)$
to range over multisets; and write $\mu + \nu \in \M_f(X)$ for the sum
of multisets, where $x \in X$ has multiplicity the sum of its
multiplicities in $X$ and $Y$.
If $\R = (\ev{\R}, +, \cdot, 0,
1)$ is a semiring and $x, y \in \ev{R}$, the \emph{Kronecker
symbol} $\delta_{x, y}$ means $1$ if $x = y$, and $0$ otherwise. If
$X$ is a set, we write $\sharp X \in \mathbb{N} \cup \{+\infty\}$ for
its cardinality if $X$ is finite, $+\infty$ otherwise.

We assume some familiarity with
categorical logic, in particular 
Seely categories \cite{panorama}.
\subsection{Continuous semirings}
We first recall the construction of the \emph{$\R$-weighted
relational model}, where $\R$ is a continuous semiring of
\emph{resources}. Our presentation follows
\cite{DBLP:conf/lics/LairdMMP13}.

A \textbf{complete partial order (cpo)} is a poset
$(X, \leq)$ with a bottom and such that any directed subset $D
\subseteq X$ has a sup $\vee D \in X$. 
\changed{For a cpo $X$, $F : X \to X$} is
\textbf{continuous} if it is monotone and preserves all suprema of
directed sets, \emph{i.e.} $F(\bigvee D) = \bigvee (F(D))$. 
\changed{An $n$-ary function $X^n \to X$} is continuous if it is continuous in each of its parameters. 
\begin{defi}
A \textbf{continuous semiring} $\R$ is a semiring $(\ev{\R}, +, \cdot,
0, 1)$ equipped with a partial order $\leq$ such that $(\ev{\R}, \leq)$
is a cpo with $0$ as bottom, and $+$ and $\cdot$ are continuous.
\end{defi}

We often denote the carrier set $\ev{\R}$ just by $\R$. The point
of considering the ordered structure on $\R$, is that for any $\R$ and
possibly infinite subset $S \subseteq \R$, the indexed sum 
\begin{eqnarray}
\sum_{x \in S} x &=& \bigvee_{F \subseteq_f S} \left( \sum_{x\in F} x
\right)\label{eq:inf_sum}
\end{eqnarray}
is always defined as the supremum of all the partial sums.

We impose two further conditions on continuous semirings. As in
\cite{DBLP:conf/lics/LairdMMP13}, they should be \textbf{commutative}:
$r \cdot r' = r' \cdot r$ for all $r, r' \in \R$.
Additionally, they should \emph{have integer
division}. If $x \in \R$ and $n \in \mathbb{N}$ is an integer, then one
may define $n*x = x + \dots + x$ (with $n$ occurrences of $x$). We say
that $\R$ has \textbf{integer division} if for all $n \geq 1$, for all
$x, y \in \R$, if $n*x = n*y$ then $x = y$. Unlike commutativity, this
condition is not required in \cite{DBLP:conf/lics/LairdMMP13};
nevertheless, all examples considered in
\cite{DBLP:conf/lics/LairdMMP13} do have integer division. From now on,
we assume all continuous semirings satisfy these two conditions.

Our core example of a continuous semiring is the following: 
\begin{defi}\label{def:semiring_N}
We write $\N$ for the continuous semiring $(\ev{\N}, +, \cdot, 0, 1)$ \changed{
equipped with the standard order on $\mathbb{N}$ extended with $x \leq +\infty$ for all $x \in \N$.
To ensure continuity we take $+$} to be the usual sum extended with $(+\infty) + x =
x + (+\infty) = +\infty$, and $\cdot$ to be  
multiplication extended with $+\infty \cdot 0 = 0 \cdot +\infty = 0$,
and $+\infty \cdot x = x \cdot +\infty = +\infty$ for any $x>0$. 
\end{defi}
As described in \cite{DBLP:conf/lics/LairdMMP13}, $\N$ may be used to
\emph{count} operational reduction sequences in a non-deterministic
language. There are other examples
\cite{DBLP:conf/lics/LairdMMP13}, including the completed non-negative
reals $\overline{\mathbb{R}}_+$, which provide an adequate model for
$\PCF$ with probabilistic choice. \changed{($\N$ has a canonical place
among those examples, because it is an initial object in the category
of continuous semirings and structure-preserving continuous maps.) 
 }


\subsection{Weighted relations} \changed{We fix a continuous semiring $\R$ and define the category $\Rel{\R}$ of $\R$-weighted relations. }
An \textbf{$\R$-relation} from a set $X$ to a set $Y$ is simply a function
\[
\alpha : X \times Y \to \R\,,
\]
also written $\alpha : X \profto Y$, regarded as a matrix with
coefficients in $\R$. We \changed{usually} 
 write $\alpha_{x, y}$ for
the coefficient $\alpha(x, y) \in \R$. For $\alpha : X \relto Y$
and $\beta : Y \relto Z$ and $x \in X, z \in Z$, we set
\begin{eqnarray}
(\beta \circ \alpha)_{x, z} &=& \sum_{y \in Y} \alpha_{x, y} \cdot
\beta_{y, z}\label{eq:relsum}
\end{eqnarray}
for the coefficients of the \textbf{composition} $\beta \circ \alpha : X
\relto Z$. For $X$ a set, the \textbf{identity} on $X$ has
$(\id_X)_{x, x'} = \delta_{x, x'}$; 
\emph{i.e.} the diagonal matrix on $X$ with only $1$'s as diagonal
coefficients. 

\begin{prop}
For any continuous semiring $\R$, there is a category $\Rel{\R}$ with sets as objects, and $\R$-relations from $X$ to $Y$ as morphisms. 
\end{prop}

\subsection{Categorical structure} \changed{
$\Rel{\R}$ is a Seely category: it is symmetric monoidal closed with
finite products and a linear exponential comonad. We review this
structure now. In fact, although this is not true in general for Seely
categories, $\Rel{\R}$ is \emph{compact closed}.}

\subsubsection{Compact closed structure}
The \emph{tensor} $X \tensor Y$ of
two sets $X, Y$, is simply their cartesian product $X \times Y$.
The \textbf{tensor} of $\alpha_1 : X_1 \relto Y_1$ and $\alpha_2 : X_2
\relto Y_2$ has coefficients
\[
(\alpha_1 \tensor \alpha_2)_{(x_1, x_2), (y_1, y_2)} = (\alpha_1)_{x_1,
y_1} \cdot (\alpha_2)_{x_2, y_2}
\]
for $(x_1, x_2) \in X_1 \tensor X_2$ and $(y_1, y_2) \in Y_1 \tensor
Y_2$, yielding $\alpha_1 \tensor \alpha_2 : X_1 \tensor X_2 \relto Y_1
\tensor Y_2$.

This operation yields a bifunctor $- \tensor - : \Rel{\R} \times
\Rel{\R} \to \Rel{\R}$ completed with
\[
\begin{array}{rcrcl}
\alpha_{X, Y, Z} &:& (X \tensor Y) \tensor Z &\relto& X \tensor (Y
\tensor Z)\\
\lambda_X &:& 1 \tensor X &\relto & X\\
\rho_X &:& X \tensor 1 &\relto& X\\
s_{X, Y} &:& X \tensor Y &\relto& Y \tensor X
\end{array}
\]
defined as the obvious variants of the identity matrix, where $1 =
\{\bullet\}$ is a singleton set. Those satisfy the \changed{necessary}
 naturality and
coherence properties, making $\Rel{\R}$ a symmetric monoidal category.
Furthermore, any set $X$ has a \emph{dual} $X^*$ defined simply as $X$
itself, and
\[
\eta_X : 1 \relto X \times X\,,
\qquad
\qquad
\epsilon_X : X\times X \relto 1
\]
turn $\Rel{\R}$ into a \emph{compact closed category}. In particular, it
follows that $\Rel{\R}$ is automatically \emph{symmetric monoidal
closed}. For $X$ and $Y$ any two sets, this gives us a notion of
\emph{linear arrow} $X \lin Y$, defined simply as $X \times Y$. We also
get a \textbf{currying} bijection for $X, Y, Z$:
\[
\Lambda : \Rel{\R}(X\tensor Y, Z) \bij \Rel{\R}(X, Y \lin Z)
\]
with, for any $\alpha : X \tensor Y \relto Z$, $\Lambda(\alpha)_{x, (y,
z)} = \alpha_{(x, y), z}$. We also get an \textbf{evaluation} morphism:
$\evm_{X, Y} : (X \lin Y) \tensor X \relto Y$
defined as having coefficients $(\evm_{X, Y})_{((x, y), x'), y')} =
\delta_{x, x'} \cdot \delta_{y,y'}$.

\subsubsection{Cartesian structure} Furthermore, $\Rel{\R}$ is 
cartesian. First, the empty set $\emptyset$ is a terminal
object, also written $\top$. If $X, Y$ are sets, we define $X \with Y$ as $X + Y$ their
\textbf{tagged disjoint union}, defined as $(\{1\} \times X) \uplus (\{2\}
\times Y)$. Note that here and from now on, we use $\uplus$ to denote
the standard set-theoretic union, when it is known to be disjoint. We
have
\[
\pi_1 : X \with Y \relto X\,,
\qquad
\qquad
\pi_2 : X \with Y \relto Y\,,
\]
the \textbf{projections} respectively defined as $(\pi_1)_{(i, x), y} = 1$
if $i = 1$ and $x = y$, and $0$ otherwise -- $\pi_2$ is defined
symmetrically. For $\alpha : X \relto Y$ and $\beta : X \relto Z$, their
\textbf{pairing} is
\[
\tuple{\alpha, \beta} : X \relto Y \with Z\,,
\]
defined with $\tuple{\alpha, \beta}_{(1, x), z} = \alpha_{x, z}$ and
$\tuple{\alpha, \beta}_{(2, y), z} = \beta_{y, z}$.
\changed{This makes $\Rel{\R}$ a cartesian 
category.} 
One must keep in mind that $\Rel{\R}$ is \emph{not} cartesian closed, as the closed structure is with respect to the tensor $\tensor$
and not the cartesian product $\with$.

\subsubsection{Linear exponential comonad.} \changed{We define a comonad $\oc$ on $\Rel{\R}$. On objects the operation $X \mapsto \oc X$ constructs the free commutative comonoid: this is defined as $\oc X = \M_f(X)$.
The action on morphisms is determined by the universal property of $\oc X$, but we give an explicit definition.}
For $\alpha : X \relto Y$, we set
\[
(\oc \alpha)_{\mu, [y_1, \dots, y_n]} = 
\sum_{\substack{(x_1, \dots, x_n)\,,\text{s.t.}\\ \mu = [x_1, \dots, x_n]}} \prod_{1\leq i \leq n}
\alpha_{x_i, y_i}\,.
\]

\changed{
Note that $\oc \alpha$ only has nonzero coefficients for pairs of multisets of the same size.
} 
Likewise, we define $\der_X : \oc X \relto X, \dig_X : \oc X \relto \oc
\oc X$,
$\mon_{X, Y} : \oc X \tensor \oc Y \relto \oc (X \with Y)$ with
\begin{eqnarray*}
(\der_X)_{\mu, x} &=& \delta_{\mu, [x]}\\
(\dig_X)_{\mu, [\nu_1, \dots, \nu_n]} &=& \delta_{\mu, \nu_1 + \dots +
\nu_n}\\
(\mon_{X, Y})_{([x_1, \dots, x_n], [y_1, \dots, y_p]), \mu} &=&
\delta_{\mu, [(1, x_1), \dots, (1, x_n), (2, y_1), \dots, (2, y_p)]}
\end{eqnarray*}
and $\mon^0
: 1 \relto \oc \top$ is defined as $1$ on its only point $(\bullet, [\,])$.
\changed{We have defined all the structure of a Seely category
\cite{panorama}, and the necessary axioms can be verified. We obtain that the Kleisli category $\Rel{\R}_\oc$
is cartesian closed.
}

\subsubsection{Recursion}\label{subsubsec:recursion}
\changed{For the interpretation of $\nPCF$ in
$\Rel{\R}$ we must give structure for recursion.} First, for any sets $X, Y$,
we order the homset $\Rel{\R}(X, Y)$ \changed{pointwise, i.e.}
\[
\alpha \leq \beta 
\quad
\Leftrightarrow
\quad
\forall x \in X, y \in Y,~\alpha_{x, y} \leq_\R \beta_{x, y}\,.
\]
It is straightforward that this defines a cpo, with bottom $\bot$ the \changed{zero matrix.}
All operations on weighted
relations involved in the Seely category structure (\emph{i.e.}
composition, tensor and pairing) are continuous with respect to this
order.

As all operations are continuous, we can define, for every set $X$, a
continuous operator
\[
\begin{array}{rcrcl}
F &:& \Rel{\R}_\oc(\top, \oc (\oc X \lin X) \lin X) &\to&
\Rel{\R}_\oc(\top, \oc (\oc X \lin X) \lin X)\\
&& \alpha & \mapsto& \lambda f.\,f\,(\alpha\,f)
\end{array}
\]
where \changed{the $\lambda$-calculus notation is well-defined since $\Rel{\R}_\oc$ is cartesian closed.} We then
define $\Y_X \in \Rel{\R}_\oc(\top, \oc (\oc X \lin X) \lin X)$ as usual with
\[
\Y_X = \bigvee_{n \in \mathbb{N}} F^n(\bot) \in \Rel{\R}_\oc(\top, \oc
(\oc X \lin X) \lin X)\,,
\] 
and with a context $Y$, we set
$\Y_{Y, X} = \Y_X \circ_\oc e_{Y}$ where $\circ_\oc$ denotes Kleisli
composition and $e_{Y} \in \Rel{\R}_\oc(Y,
\top)$ the terminal morphism, yielding 
$\Y_{Y, X} \in \Rel{\R}_\oc(Y, \oc(\oc X \lin X) \lin X)$.

\subsection{Interpretation of $\nPCF$} Now, we define $\nPCF$ and its
interpretation.

\subsubsection{Non-deterministic $\PCF$} 
The \textbf{types} of $\nPCF$ are given by the following grammar:
\[
\begin{array}{rcll}
A, B &::=& \tbool \mid \tnat \mid A \to B
\end{array}
\]
where $\tbool$ and $\tnat$ are respectively types for \emph{booleans}
and \emph{natural numbers}. We refer to $\tbool$ and $\tnat$ as
\emph{ground types}, and use $\tx, \ty$ to range over those.  
\begin{figure}
\boxit{
\begin{mathpar}
\inferrule
        { }
        { \Gamma \vdash \ttrue : \tbool }
\and
\inferrule
        { }
        { \Gamma \vdash \tfalse : \tbool }
\and
\inferrule
        { }
        { \Gamma \vdash n : \tnat }
\and
\inferrule
        { }
        { \Gamma, x : A \vdash x : A }
\and 
\inferrule
        { \Gamma, x : A \vdash M : B }
        { \Gamma \vdash \lambda x^A.\,M : A\to B }
\and 
\inferrule
        { \Gamma \vdash M : A \to B \\ 
          \Gamma \vdash N : A }
        { \Gamma \vdash M\,N : B }
\and
\inferrule 
        { \Gamma \vdash M : \tbool \\ 
          \Gamma \vdash N_1 : \tx \\
          \Gamma \vdash N_2 : \tx }
        { \Gamma \vdash \ite{M}{N_1}{N_2} : \tx }
\and
\inferrule
        { \Gamma \vdash M : \tnat }
        { \Gamma \vdash \tsucc\,M : \tnat }
\and
\inferrule
        { \Gamma \vdash M : \tnat }
        { \Gamma \vdash \tpred\,M : \tnat }
\and
\inferrule
        { \Gamma \vdash M : \tnat }
        { \Gamma \vdash \iszero\,M : \tbool }
\and
\inferrule
        { \Gamma \vdash M : A \to A } 
        { \Gamma \vdash \Y\,M : A }
\end{mathpar}
}
\caption{Typing rules for $\PCF$}
\label{fig:typ_pcf}
\end{figure}
\textbf{Typed terms} are defined via the typing rules of Figure
\ref{fig:typ_pcf} -- in this paper, all terms are
well-typed.
\textbf{Contexts} are lists of typed variables $x_1
: A_1, \dots, x_n : A_n$. \textbf{Typing judgments} have the form
$\Gamma \vdash M : A$, where $\Gamma$ is a context and $A$ is a type. In
addition to the rules listed in Figure \ref{fig:typ_pcf}, \changed{the language has}
 an explicit exchange rule \changed{for permuting variable declarations in contexts}.  Conditionals are restricted to the base
type, but \changed{as usual in call-by-name} general conditionals can be defined as syntactic sugar.
Finally, $\nPCF$ also has a
\textbf{non-deterministic choice} with typing rule:
\[
\inferrule
	{ }
	{ \Gamma \vdash \coin : \tbool }
\]

We omit the operational semantics \cite{DBLP:conf/lics/LairdMMP13}, \changed{and only recall that} we get for $\vdash M : \tx$ and value $v:\tx$ a
\emph{weight}, \changed{defined as} 
%
the number of distinct reduction
sequences evaluating $M$ to $v$. We write $M \eval^n v$ if there are exactly $n$ reduction sequences from $M$ to $v$.

\subsubsection{Interpretation} We may now define the interpretation of
$\nPCF$ in $\Rel{\N}_\oc$.

To every type $A$ we associate a set $\rintr{A}$, its \textbf{web}, 
defined by $\rintr{\tbool} = \{\ttrue, \tfalse\}$ and
$\rintr{\tnat} = \mathbb{N}$ the set of natural numbers, extended to all
types via $\rintr{A \to B} = \oc \rintr{A} \lin \rintr{B}$.
For contexts:
\[
\rintr{x_1 : A_1, \dots, x_n : A_n} = \with_{1\leq i \leq n}
\rintr{A_i}\,,
\]
and terms $\Gamma \vdash M : A$ are interpreted as morphisms
$\rintr{M} \in \Rel{\N}_\oc(\rintr{\Gamma}, \rintr{A})$. We omit the
standard definitions for the $\lambda$-calculus constructions. For
$\PCF$ combinators, we set: 
\begin{eqnarray*}
\rintr{\Gamma \vdash v : \tx}_{\gamma, v'} &=& \delta_{\gamma, []}
\cdot \delta_{v, v'}\\
\rintr{\Gamma \vdash \ite{M}{N_1}{N_2} : \tx} &=& \mathsf{if} \circ_\oc
\tuple{\rintr{M}, \rintr{N_1}, \rintr{N_2}}\\
\rintr{\Gamma \vdash \tpred\,M : \tnat} &=& \mathsf{pred} \circ_\oc
\rintr{M}\\
\rintr{\Gamma \vdash \tsucc\,M : \tnat} &=& \mathsf{succ} \circ_\oc
\rintr{M}\\
\rintr{\Gamma \vdash \iszero\,M : \tbool} &=& \mathsf{iszero} \circ_\oc
\rintr{M}\\
\rintr{\Gamma \vdash \Y\,M : A} &=& \Y_{\rintr{\Gamma},\rintr{A}}\,\rintr{M}
\end{eqnarray*}
\begin{figure}
\begin{minipage}{.49\linewidth}
\[
\scalebox{.8}{$
\begin{array}{rcrcll}
\mathsf{if} &:& \oc (\tbool \with \tx \with \tx) &\relto& \tx\\
&&\mathsf{if}_{\gamma, v} &=& 1 &\text{if $\gamma = [(1, \ttrue), (2, v)]$}\\
&&&&&\text{or $\gamma = [(1, \tfalse), (3, v)],$}\\
&&\mathsf{if}_{\gamma, v} &=& 0 &\text{otherwise.}\\\\
\mathsf{pred} &:& \oc \tnat &\relto & \tnat\\
&&\mathsf{pred}_{\gamma, n} &=& 1 &\text{if $\gamma = [0]$ and $n = 0$,}\\
&&&&&\text{or $\gamma = [k+1]$ and $n = k$,}\\
&&\mathsf{pred}_{\gamma, n} &=& 0 &\text{otherwise.}
\end{array}
$}
\]
\end{minipage}
\hfill
\begin{minipage}{.49\linewidth}
\[
\scalebox{.8}{$
\begin{array}{rcrcll}
\mathsf{succ} &:& \oc \tnat &\relto& \tnat\\
&&\mathsf{succ}_{\gamma, n} &=& 1 &\text{if $\gamma = [k]$ and $n =
k+1$,}\\
&&\mathsf{succ}_{\gamma, n} &=& 0 &\text{otherwise.}\\\\
\mathsf{iszero} &:& \oc \tnat &\relto& \tbool\\
&&\mathsf{iszero}_{\gamma, b} &=& 1 &\text{if $\gamma = [0]$ and $b =
\ttrue$,}\\
&&&&&\text{or $\gamma = [k+1]$ and $b = \tfalse$,}\\
&&\mathsf{iszero}_{\gamma, b} &=& 0 &\text{otherwise.}
\end{array}
$}
\]
\end{minipage}
\caption{Interpretation of basic $\PCF$ combinators}
\label{fig:pcf_comb}
\end{figure}
with the weighted relations in Figure \ref{fig:pcf_comb}. Finally, we set
$\rintr{\Gamma \vdash \coin : \tbool}_{\gamma, v} = \rintr{\ttrue}_{\gamma,
v} + \rintr{\tfalse}_{\gamma, v}$,
so that $\rintr{\Gamma \vdash \coin : \tbool} : \oc \rintr{\Gamma} \relto
\rintr{\tbool}$ as required, concluding the interpretation of $\nPCF$. 

The reader is referred to \cite{DBLP:conf/lics/LairdMMP13} for the proof
of the following adequacy property:

\begin{thm}\label{th:weight_adequacy}
For any $\vdash M : \tx$, for any value $v:\tx$ and $n \in
\mathbb{N}$, 
$M \eval^n v$ iff $\rintr{M}_{[], v} = n $.
\end{thm}

\subsubsection{What does the $\N$-weigthed relational model count?}
\label{subsubsec:question}
Theorem \ref{th:weight_adequacy} shows that at ground types, the
$\N$-weighted relational model counts the distinct reduction
sequences to a value. But for a general type $A$, a term $\vdash M : A$, and $x \in
\rintr{A}$, the
meaning of the coefficient $\rintr{M}_x \in \mathbb{N} \uplus
\{+\infty\}$ is more difficult to describe. Even at higher-order we expect it to be
related to the cardinality of some set of concrete witnesses: but which one?

In this paper we give one answer to this question, in 
terms of \emph{concurrent games}.


\section{Concurrent game semantics of $\nPCF$}
\label{sec:cg_pcf}

Our games model of $\nPCF$ is based on \emph{thin concurrent games}
\cite{cg2}, to which we add an \emph{exhaustivity} mechanism inspired by
Melliès \cite{DBLP:conf/lics/Mellies05}. \changed{Game semantics is naturally \emph{affine} and the purpose of exhaustivity is to ensure strict linearity, in order to  establish a tighter correspondence with $\Rel{\R}$. As we will see, the two models can be related by a functor preserving some of the Seely category structure.}


\subsection{Event structures with symmetry} We start with preliminaries
on \emph{event structures with symmetry}
\cite{DBLP:journals/entcs/Winskel07}, the mathematical structure on
which thin concurrent games rest.

\subsubsection{Event structures}
Specifically, we use prime event structures with binary conflict:

\begin{defi}\label{def:es}
An \textbf{event structure (\emph{es})} is a triple $E = (\ev{E},
\leq_E, \conflict_E)$, where
$\ev{E}$ is a countable
 set of \textbf{events}, $\leq_E$ is a partial
order called
\textbf{causal dependency} and $\conflict_E$ is an irreflexive symmetric
binary
relation on $\ev{E}$ called \textbf{conflict}, satisfying:
\[
\begin{array}{rl}
\text{\emph{finite causes}:}& \forall e\in \ev{E},~[e]_E = \{e'\in
\ev{E} \mid
e'\leq_E e\}~\text{is finite,}\\
\text{\emph{conflict inheritance}:}&
\forall e_1 \conflict_E e_2,~\forall e_2 \leq_E e'_2,~e_1 \conflict_E
e'_2\,.
\end{array}
\]
\end{defi}

We write $e \imc_E e'$ for \textbf{immediate causality}, \emph{i.e.} $e
<_E e'$ with no event in between. A notion of critical importance \changed{for working with} 
event structures is that of \emph{configurations}:

\begin{defi}
A (finite) \textbf{configuration} of event structure $E$ is a finite
$x \subseteq \ev{E}$ which is
\[
\begin{array}{rl}
\text{\emph{down-closed}:}& \forall e\in x,~\forall e' \in
\ev{E},~e'\leq_E e\implies
e' \in x.\\
\text{\emph{consistent}:}& \forall e,e'\in x,~\neg(e \conflict_E e')\,.
\end{array}
\]

We write $\conf{E}$ for the set of finite configurations on $E$.
\end{defi}

The set $\conf{E}$ is naturally ordered by inclusion; it is the
\emph{domain of configurations}. Configurations are typically ranged
over by variables $x, y, z$. For $x, y \in \conf{E}$, we write $x \cov
y$ if $x$ is immediately below $y$ in the inclusion order, \emph{i.e.}
there is $e \in \ev{E}$ such that $e \not \in x$ and $y = x \cup \{e\}$
-- in that case, we also write $x \vdash_E e$ and say that $x$
\textbf{enables} $e$.
Observe also that any $x \in \conf{E}$ inherits a partial order $\leq_x$, \changed{the restriction of $\leq_E$ to $x\times x$.
}  We usually consider a configuration $x \in \conf{E}$ as a partially ordered set.

Event structures are a so-called \emph{truly concurrent} model:
rather than presenting observable execution traces, they
list computational events along with their
causal dependence and independence. The causal order
$\leq_E$ is ``conjunctive'': for an event to occur, \emph{all} its
dependencies must be met first. The conflict relation $\conflict_E$
represents an irreconciliable non-deterministic choice. Finally,
\emph{configurations} provide the adequate notion of \emph{state}.

\subsubsection{Symmetry} Plain event structures are not expressive
enough for our purposes, notably to handle repetitions in games.
Instead, we use event structures \emph{with symmetry}:

\begin{defi}\label{def:isofam}%
An \textbf{isomorphism family} on event structure $E$ is a set
$\tilde{E}$ of
bijections between configurations of $E$, satisfying the additional
conditions:
\[
\begin{array}{rl}
\text{\emph{groupoid:}}& \tilde E\text{~contains identity bijections;
is closed
under composition and inverse.}\\
\text{\emph{restriction:}}&
\text{for all~$\theta : x \bij y \in \tilde{E}$ and $x \supseteq x' \in
\conf{E}$,}\\
&\text{there is a (necessarily) unique $\theta \supseteq \theta' \in
\tilde{E}$
such that $\theta' : x' \bij y'$.}\\
\text{\emph{extension:}}&\text{for all $\theta : x \bij y \in
\tilde{E}$, $x
\subseteq x' \in \conf{E}$,}\\
&\text{there is a (not necessarily unique) $\theta \subseteq \theta' \in
\tilde{E}$ such that $\theta' : x' \bij y'$.}
%
\end{array}
\]

The pair $(E, \tilde E)$ is called an \textbf{event structure with
symmetry (\emph{ess})}.
\end{defi}

We regard isomorphism families as
\emph{proof-relevant} equivalence relations: they convey the information
of which configurations are interchangeable, witnessed by an explicit
bijection. 

If $E$ is an ess, we call the elements of $\tilde{E}$ 
\textbf{symmetries}.  It is easy to prove that symmetries
are automatically order-isos
\cite{DBLP:journals/entcs/Winskel07}. 
We write $\theta : x \sym_{E} y$ to mean that
$\theta : x \simeq y$ is a bijection s.t. $\theta \in \tilde{E}$ \changed{with} $x = \dom(\theta)$ and $y = \cod(\theta)$. We also write $x
\sym_E y$ to mean that there is a symmetry $\theta$ s.t. $\theta :
x \sym_E y$. This induces an equivalence relation on configurations -- we
write $\sconf{E}$ for the set of equivalence classes, called
\textbf{symmetry classes}, and use $\x, \y, \z \in \sconf{E}$
to range over them. Symmetry classes are always non-empty 
sets of configurations, \changed{and the 
symmetry class of the empty configuration is always a singleton $\{\emptyset\}$}. Abusing notation we 
write $\emptyset \in \sconf{E}$, which should cause no confusion.

In \emph{thin concurrent games}, both games and strategies are certain ess. 

\subsection{Games}
\label{subsec:games}
 We introduce our games, and the corresponding
constructions.

\subsubsection{Definition}

First, we recall \emph{thin concurrent games} in the sense of \cite{cg2}:

\begin{defi}\label{def:tcg}
A \textbf{thin concurrent game (tcg)} is an ess $A = (\ev{A}, \leq_A,
\conflict_A)$ with isomorphism families $\tilde{A}, \ptilde{A},
\ntilde{A}$ s.t. $\ptilde{A} \subseteq \tilde{A}$, $\ntilde{A}
\subseteq \tilde{A}$, and 
\[
\pol_A : \ev{A} \to \{-, +\}
\]
a \textbf{polarity function} preserved by symmetries, and additionally
subject to the conditions: 
\[
\begin{array}{rl}
\text{\emph{orthogonality:}} & \text{for all $\theta \in \tilde{A}$, if
$\theta \in \ptilde{A} \cap \ntilde{A}$, then $\theta = \id_x$ for some
$x \in \conf{A}$,}\\
\text{\emph{$-$-receptivity:}} & \text{if $\theta \in \ntilde{A}$ and
$\theta \subseteq^- \theta' \in \tilde{A}$, then $\theta' \in
\ntilde{A}$,}\\
\text{\emph{$+$-receptivity:}} & \text{if $\theta \in \ptilde{A}$ and
$\theta \subseteq^+ \theta' \in \tilde{A}$, then $\theta' \in
\ptilde{A}$,}
\end{array}
\]
where $\theta \subseteq^p \theta'$ means that $\theta \subseteq \theta'$
adding only (pairs of) events of polarity $p$.
\end{defi}

We shall see examples in Section \ref{subsubsec:games_construction},
accompanying the constructions. Intuitively, negative events
correspond to Opponent moves, i.e. actions of the execution environment,
while events with positive polarity are Player moves, i.e. actions of the program under study. 
Symmetries correspond to changing the copy
indices arising from $\oc(-)$. Positive symmetries reindex Player events, while negative symmetries reindex Opponent events.

For us in this paper, a \emph{game} will be a tcg along
with a \emph{payoff function}:

\begin{defi}\label{def:game}
A \textbf{game} is a tcg $A = (\ev{A}, \leq_A, \conflict_A, \tilde{A},
\ptilde{A}, \ntilde{A}, \pol_A)$ with
\[
\kappa_A : \conf{A} \to \{-1, 0, +1\}
\]
a \textbf{payoff function} satisfying the following conditions:
\[
\begin{array}{rl}
\text{\emph{invariant:}} & \text{for all $\theta : x \sym_A y$, we have
$\kappa_A(x) = \kappa_A(y)$,}\\
\text{\emph{representable:}} & \text{postponed until Section
\ref{subsubsec:representability}.}
\end{array}
\]

Writing $\min(A)$ for the minimal events of $A$, a \textbf{$-$-game}
must additionally satisfy: 
\[
\begin{array}{rl}
\text{\emph{negative:}} & \text{for all $a \in \min(A)$, $\pol_A(a) =
-$,}\\
\text{\emph{initialized:}} & \kappa_A(\emptyset) \geq 0\,.
\end{array}
\]

Finally, a $-$-game $A$ is \textbf{strict} if
$\kappa_A(\emptyset) = 1$ and all its initial moves are in pairwise
conflict. It is \textbf{well-opened} if it is strict
with exactly  one initial move. 
\end{defi}

Contexts and types of $\nPCF$ will be represented as strict $-$-games, 
but more general games will arise during the model construction.

The condition \emph{representable} is not necessary to get a model of
$\nPCF$ but only for the collapse theorem. We shall only introduce and
describe it later on, when it becomes relevant.
The payoff function $\kappa_A$ assigns a value to each configuration.
Configurations with payoff $0$ are called \textbf{complete}: they
correspond to 
\changed{\emph{terminated}} 
 executions,
which have reached an adequate stopping point. 
Otherwise, $\kappa_A$ assigns a responsibility for why a
configuration is non-complete. If $\kappa_A(x) = -1$ then Player is responsible, otherwise it is Opponent. 

The payoff structure helps to manage the
mismatch between game semantics, which are inherently \emph{affine}, and
relational semantics, which are inherently \emph{linear}.
\changed{Using payoff we will restrict to the strategies that behave
linearly; then we can investigate the properties of our collapse
at the level of Seely categories. This can also be achieved
with other techniques; for example one make the weighted relational
model affine by decomposing the comonad $\oc$ on $\Rel{\R}$ as the
composition $\oc_\mathrm{contr} \circ \oc_{\mathrm{weak}}$ of a comonad
allowing weakening, and one allowing arbitrary duplication. A similar
construction appears in \cite[8.10]{panorama}. With either approach we
obtain a cartesian closed functor between the respective Kleisli
categories for $\oc$. }

From \emph{invariant}, all configurations in a
symmetry class $\x \in \sconf{A}$ have the same payoff, so we may 
write $\kappa_A(\x)$ unambiguously.
We now introduce constructions on games.

\subsubsection{Basic games}\label{subsubsec:games_construction} 
Firstly  
\begin{figure}
\begin{minipage}{.45\linewidth}
\[
\xymatrix@R=15pt@C=10pt{
&\qu^-
\ar@{.}[dl]
\ar@{.}[dr]\\
\ttrue^+\ar@{~}[rr]&&
\tfalse^+
}
\]
\caption{The $-$-game $\gbool$}
\label{fig:ar_bool}
\end{minipage}
\hfill
\begin{minipage}{.45\linewidth}
\[
\xymatrix@R=15pt@C=10pt{
&&\qu^-
\ar@{.}[dll]
\ar@{.}[dl]
\ar@{.}[d]
\ar@{.}[dr]\\
0^+     \ar@{~}[r]&
1^+     \ar@{~}[r]&
2^+     \ar@{~}[r]&
\dots
}
\]
\caption{The $-$-game $\gnat$}
\label{fig:ar_nat}
\end{minipage}
\end{figure}
we draw in Figures \ref{fig:ar_bool} and \ref{fig:ar_nat} the $-$-games
corresponding to the basic types $\tbool$ and $\tnat$. The diagrams
represent the event structures, read from top (the minimum) to bottom
(maximal events). Events are annotated with their polarity, and the
wiggly line indicates \emph{conflict} -- we adopt the convention that we
only draw \emph{minimal} conflict, \emph{i.e.} we omit it when it can be
deduced via \emph{conflict inheritance}. In Figure \ref{fig:ar_nat}, all
positive events are assumed to be in pairwise conflict (this is not
reflected in the diagram for readability). The isomorphism families are
not represented, but for these games they are trivial and only comprise
identity bijections between configurations. Finally, we have
\[
\kappa_\gx(\emptyset) = +1\,,
\qquad
\kappa_\gx(\{\qu\}) = -1\,,
\qquad
\kappa_\gx(\{\qu, v\}) = 0\,,
\]
with $\gx \in \{\gbool, \gnat\}$. This covers all possible
configurations on $\gbool$ and $\gnat$. 

Although there are more configurations on $\gbool$ and $\gnat$
than points in $\rintr{\tbool}$ and $\rintr{\tnat}$, the mismatch is
resolved when considering \emph{complete} configurations:

\begin{lem}
Writing $\nconf{A} = \{x \in \conf{A} \mid \kappa_A(x) = 0\}$, 
we have two bijections:
\[
\begin{array}{rcl}
\nconf{\gbool} &\bij& \rintr{\tbool}\\
\{\qu, b\} &\mapsto & b
\end{array}
\qquad
\qquad
\begin{array}{rcrcl}
\nconf{\gnat} &\bij& \rintr{\tnat}\\
\{\qu,n\} &\mapsto & n\,.
\end{array}
\]
\end{lem}

So the web $\rintr{\tx}$ corresponds to complete configurations of
$\gx$. For affine or linear languages with no
replication, this correspondence is preserved by all type constructors.
But in the presence of replication, one must consider \emph{complete
symmetry classes} instead:

\begin{lem}\label{lem:r_basic}
Writing $\wconf{A} = \{\x \in \sconf{A} \mid \kappa_A(\x) = 0\}$, we
have two bijections:
\[
\begin{array}{rcrcl}
s^\tbool &:& \wconf{\gbool} &\bij& \rintr{\tbool}
\end{array}
\qquad
\qquad
\begin{array}{rcrcl}
s^\tnat &:& \wconf{\gnat} &\bij& \rintr{\tnat}\,.
\end{array}
\]
\end{lem}

This is not saying much: the symmetries on $\gbool$ and $\gnat$ are trivial, so symmetry classes are in one-to-one correspondence with
configurations. This will not always be the case.  

\subsubsection{Basic constructions on ess} We start with some
constructions on plain ess. 

\begin{defi}
Consider $E_1$ and $E_2$ two event structures with symmetry.

Then, we define their \textbf{parallel composition} $E_1 \parallel E_2$
as comprising the components: 
\[
\begin{array}{rcrcl}
\text{\emph{events:}} &~~~& \ev{E_1\parallel E_2} &=& \{1\}\times
\ev{E_1} ~~ \uplus ~~ \{2\}\times \ev{E_2}\\
\text{\emph{causality:}} && (i, e) \leq_{E_1 \parallel E_2} (j, e')
&\Leftrightarrow&
i=j~\&~e \leq_{E_i}
e'\\
\text{\emph{conflict:}} && (i, e) \conflict_{E_1 \parallel E_2} (j, e')
&\Leftrightarrow& i=j ~\&~ e \conflict_{E_i} e'\,,\\
\text{\emph{symmetry:}} && \theta \in \tilde{E_1 \parallel E_2}
&\Leftrightarrow& \exists \theta_1 \in \tilde{E_1}, \theta_2 \in
\tilde{E_2}, \theta = \theta_1 \parallel \theta_2
\end{array}
\]
where, if $\theta_i : x_i \sym_{E_i} y_i$, we set $(\theta_1 \parallel
\theta_2)(i, e) = (i, \theta_i(e))$.
\end{defi}

Note that any configuration $x \in \conf{E_1 \parallel E_2}$ decomposes
uniquely as $(\{1\} \times x_1) \uplus (\{2\} \times x_2)$, which we
also write $x_1 \parallel x_2$. This is compatible with symmetry: if
$\theta : x_1 \parallel x_2 \sym_{E_1 \parallel E_2} y_1 \parallel
y_2$, then $\theta$ decomposes uniquely as $\theta_1 \parallel \theta_2$
with $\theta_i : x_i \sym_{E_i} y_i$. We may also observe that any
symmetry class $\x \in \sconf{E_1 \parallel E_2}$ has the form $\x_1
\parallel \x_2$, which is shorthand for $\{x_1 \parallel x_2 \mid x_1
\in \x_1,~x_2 \in \x_2\}$ for $\x_1 \in \sconf{E_1}$ and $\x_2 \in
\sconf{E_2}$.

\changed{We also use a variant of the above where components are in conflict:}

\begin{defi}
Let $E_1$ and $E_2$ be two event structures with symmetry.

Then, we define their \textbf{sum} $E_1 + E_2$
as comprising the components:
\[
\begin{array}{rcrcl}
\text{\emph{events:}} &~~~& \ev{E_1\parallel E_2} &=& \{1\}\times
\ev{E_1} ~~ \uplus ~~ \{2\}\times \ev{E_2}\\
\text{\emph{causality:}} && (i, e) \leq_{E_1 \parallel E_2} (j, e')
&\Leftrightarrow&
i=j~\&~e \leq_{E_i}
e'\\
\text{\emph{conflict:}} && (i, e) \conflict_{E_1 \parallel E_2} (j, e')
&\Leftrightarrow& i\neq j ~\vee~ e \conflict_{E_i} e'\,,\\
\text{\emph{symmetry:}} && \theta \in \tilde{E_1 \parallel E_2}
&\Leftrightarrow& \exists \theta_1 \in \tilde{E_1}, \theta_2 \in
\tilde{E_2}, \theta = \theta_1 \parallel \theta_2\,,
\end{array}
\]
where, necessarily, one of $\theta_1$ or $\theta_2$ must be empty.
\end{defi}

Any non-empty $x \in \conf{E_1 + E_2}$ may be written uniquely either as
$\{1\} \times x_1$ for $x_1 \in \conf{E_1}$, or as $\{2\} \times x_2$ for $x_2 \in
\conf{E_2}$. By convention, we write $(1, x_1) \in
\conf{E_1 + E_2}$ for the former and $(2, x_2) \in \conf{E_1 + E_2}$ for
the latter. Likewise, non-empty symmetries on $E_1 + E_2$ may be written
uniquely either as $(1, \theta_1)$ or as $(2, \theta_2)$ for $\theta_i \in
\tilde{E_i}$, defined in the obvious way. For $\x_1 \in \sconf{E_1}$, we
also write $(1, \x_1) \in \sconf{E_1\with E_2}$ as a shorthand for
$\{(1, x_1) \mid x_1 \in \x_1\}$ and likewise for $(2, \x_2) \in
\sconf{E_1 \with E_2}$ for $\x_2 \in \sconf{E_2}$.
Every non-empty symmetry class on $E_1 \with E_2$ may be written
uniquely via one of these two shapes. 

\subsubsection{Basic constructions on games}
\label{subsubsec:basic_constructions}

Our first construction is the \textbf{dual}
$A^\perp$ of a game $A$. We set $A^\perp$ as the same ess as $A$,
changing only the components relative to polarities: more precisely,
$\pol_{A^\perp} = - \pol_A$, $\ptilde{A^\perp} = \ntilde{A}$, 
$\ntilde{A^\perp} = \ptilde{A}$\changed{, and $\kappa_{A^\perp} = -\kappa_A.$}

Parallel composition splits into \emph{tensor} and \emph{par},
\changed{which differ only in their payoff function:}   

\begin{defi}\label{def:games_tensor}
Consider two games $A$ and $B$. 

We define their \textbf{tensor} $A \tensor B$ as having ess $A
\parallel B$, with the additional components:
\[
\begin{array}{rrcl}
\text{\emph{polarities:}} &
\pol_{A\tensor B}(1, a) &=& \pol_A(a)\\
&\pol_{A\tensor B}(2, b) &=& \pol_B(b)\\
\text{\emph{positive symmetries:}} & 
\theta_A \parallel \theta_B \in \ptilde{A\tensor B} &\Leftrightarrow&
\theta_A \in \ptilde{A} ~ \& ~ \theta_B \in \ptilde{B}\,\\
\text{\emph{negative symmetries:}} &
\theta_A \parallel \theta_B \in \ntilde{A\tensor B} &\Leftrightarrow& 
\theta_A \in \ntilde{A} ~ \& ~ \theta_B \in \ntilde{B}\,\\
\text{\emph{payoff:}} & 
\kappa_{A\tensor B}(x_A \parallel x_B) &=& \kappa_A(x_A) \tensor
\kappa_B(x_B)
\end{array}
\]
where the binary operation $\tensor$ on $\{-1, 0, +1\}$ is defined in
Figure \ref{fig:op_payoff}.
We also define the \textbf{par} $A \parr B$ with the same components,
except for $\kappa_{A\parr B}(x_A \parallel x_B) = \kappa_A(x_A) \parr
\kappa_B(x_B)$.
\end{defi}

The tensor of two $-$-games is again a $-$-game. 
\begin{figure}
\[
\begin{array}{c|ccc}
\tensor & -1 & 0 & +1\\
\hline
-1 & -1 & -1 & -1\\
0 & -1 & 0 & +1 \\
+1 & -1 & +1 & +1
\end{array}
\qquad\qquad\qquad\qquad
\begin{array}{c|ccc}
\parr & -1 & 0 & +1\\
\hline
-1 & -1 & -1 & +1\\
0 & -1 & 0 & +1 \\
+1 & +1 & +1 & +1
\end{array}
\]
\caption{Payoff for $\tensor$ and $\parr$}
\label{fig:op_payoff}
\end{figure}
The par also preserves $-$-games, but we 
\changed{will often use it on more general games: for $-$-games $A$ and $B$ we will eventually define a strategy \emph{from $A$ to $B$} as a strategy on $A^\perp \parr B$. This is a game but not a $-$-game.}

Analogously to Lemma \ref{lem:r_basic}, we have:

\begin{lem}\label{lem:r_tensor}
Consider $A$ and $B$ any games. Then, we have
\[
\begin{array}{rcrcl}
r^\tensor_{A, B} &:& \conf{A\tensor B} &\bij& \conf{A} \times
\conf{B}\\
&&x_A \parallel x_B &\mapsto& (x_A, x_B)
\end{array}
\quad
\begin{array}{rcrcl}
s^{\tensor}_{A, B} &:& \wconf{A\tensor B} &\bij& \wconf{A} \times \wconf{B}\\
&&\x_A \parallel \x_B &\mapsto& (\x_A, \x_B)
\end{array}
\]
and likewise, $s^\parr_{A, B} : \wconf{A\parr B} \bij \wconf{A} \times
\wconf{B}$ with the same function.
\end{lem}

The proof is straightforward, and uses that in Figure
\ref{fig:op_payoff}, for either $\tensor$ and $\parr$, a configuration
$x_A \parallel x_B$ has null payoff iff it has null payoff on both
sides. This connects the tensor product of games with that in the
relational model, which is defined as the cartesian product of sets.
\changed{We move to another construction on games.}
\begin{defi}\label{def:with}
For two strict $-$-games $A_1$ and $A_2$, we define their \textbf{with} $A_1 \with A_2$ as having ess $A_1 +
A_2$, with the additional components:
\[
\begin{array}{rrcll}
\text{\emph{polarities:}} &
\pol_{A\with B}(1, a) &=& \pol_A(a)\\
&\pol_{A\with B}(2, b) &=& \pol_B(b)\\
\text{\emph{positive symmetries:}} & 
(i, \theta) \in \ptilde{A_1 \with A_2} &\Leftrightarrow&
\theta \in \ptilde{A_i}\,\\
\text{\emph{negative symmetries:}} &
(i, \theta) \in \ntilde{A_1 \with A_2} &\Leftrightarrow& 
\theta \in \ntilde{A_i}\,,\\
\text{\emph{payoff:}} & 
\kappa_{A_1\with A_2}((i, x)) &=& \kappa_{A_i}(x)\,,&
(x \neq \emptyset)\\
&\kappa_{A_1\with A_2}(\emptyset) &=& 1\,,
\end{array}
\]
yielding a strict $-$-game.
\end{defi}
\changed{As we will see, this construction gives a cartesian product in our forthcoming category of strategies. It can also be applied to non-strict $-$-games, but then it is not a product: if one of the $A_i$ is not strict then the
corresponding projection does not respect payoff (in the sense of Definition \ref{def:exh_strat}), because we have set $\kappa_{A_1
\with A_2}(\emptyset) = 1$. On the other hand having $\kappa_{A_1 \with
A_2}(\emptyset) = 0$ breaks the correspondence with the relational
model, since the empty configuration does not correspond in a canonical way to one of the components.}

On complete symmetry classes this matches the corresponding
construction in $\Rel{\R}$:

\begin{lem}\label{lem:conf_with}
Consider $A$ and $B$ any strict $-$-games. Then, we have
\[
\begin{array}{rcrcl}
r^\with_{A, B} &:& \confn{A\with B} &\bij& \confn{A} + \confn{B}\\
s^\with_{A, B} &:& \wconf{A\with B} &\bij& \wconf{A} + \wconf{B}\,.
\end{array}
\]
\end{lem}
This generalizes directly to the $n$-ary case.
The bijection follows our notation $(i, \x)$ for symmetry classes of
non-empty configurations. As discussed above, the bijection
$s^\with_{A, B}$ relies on strictness, which
ensures that configurations with null payoff cannot be empty.

\subsubsection{Arrow}\label{subsubsec:arrow}
\changed{
Next we give the construction of a \emph{linear} function space. The event structure $A \lin B$ is easier to describe when $B$ is well-opened, so we consider this case first. 
}
\begin{defi}
\label{def:linearfunctionspace}
\changed{
Consider two games $A$ and $B$, where $B$ is well-opened and has unique initial move $b_0$. The \textbf{linear function space} $A \lin B$ has the following
components:}
\[
\begin{array}{rrcl}
\text{\emph{events:}} & 
\ev{A\lin B} &=& \ev{A^\perp \parallel B}\\
\text{\emph{causality:}} &
\leq_{A\lin B} &=& {\leq_{A^\perp \parallel B}} \uplus \{((2, b_0),
(1, a)) \mid a \in \ev{A}\}\\
\text{\emph{conflict:}} &
\conflict_{A\lin B} &=& \conflict_{A^\perp \parallel B}\\
\text{\emph{symmetries:}} &
\tilde{A\lin B} &=& \{\theta : x \sym_{A^\perp \parallel B} y \mid x, y
\in \conf{A\lin B}\}\\
\text{\emph{polarities:}} &
\pol_{A\lin B} &=& \pol_{A^\perp \parallel B}\\
\text{\emph{positive symmetries:}} &
\ptilde{A\lin B} &=& \{\theta : x \sym_{A^\perp \parallel B}^+ y \mid x,
y \in \conf{A\lin B}\}\\
\text{\emph{negative symmetries:}} &
\ntilde{A\lin B} &=& \{\theta : x \sym_{A^\perp \parallel B}^- y \mid x,
y \in \conf{A\lin B}\}\\
\text{\emph{payoff:}} &
\kappa_{A\lin B}(x_A \parallel x_B) &=& \kappa_{A^\perp}(x_A) \parr
\kappa_B(x_B) \qquad (x_B \neq \emptyset)\\
&\kappa_{A \lin B}(\emptyset) &=& 1\,.
\end{array}
\]

This is a well-opened $-$-game.
\end{defi}

Again, symmetry classes of $A\lin B$ relate with the
corresponding construction in $\Rel{\R}$:

\begin{lem}\label{lem:conf_lin_wo}
Consider two games $A$ and $B$ with $B$ well-opened. Then, we have
\[
\begin{array}{rcrcl}
r^\lin_{A, B} &:& \confn{A\lin B} &\bij& \conf{A} \times \confn{B}\\
s^\lin_{A, B} &:& \wconf{A\lin B} &\bij& \wconf{A} \times \wconf{B}
\end{array}
\]
\end{lem}

Again, the bijection $s^\lin_{A, B}$ relies on $B$
being \emph{strict}: if we had $\kappa_B(\emptyset) = 0$, then 
there would be no symmetry class on $A\lin B$
corresponding to a pair $(\x_A, \emptyset)$ with $\x_A$ non-empty.

It is easy to extend the construction of $A \lin B$ to the case where $B$ is
\emph{strict} but not well-opened: in that case, $B$ 
has a canonical form $B \iso \bigwith_{b \in \min(B)} B_b$
with $B_b$ well-opened, for the obvious notion of isomorphism \changed{between} games.
This lets us set:
\[
A \lin B \quad = \quad \bigwith_{b \in \min(B)} A \lin B_b\,,
\]
with in particular $A \lin \top = \top$ for $\top$ the empty game with
$\kappa_\top(\emptyset) = 1$. We retain:

\begin{lem}\label{lem:bij_lin_strict}
Consider two games $A$ and $B$ with $B$ strict. Then, we have 
bijections
\[
\begin{array}{rcrcl}
r^\lin_{A, B} &:& \confn{A\lin B} &\bij& \conf{A} \times \confn{B}\\
s^\lin_{A, B} &:& \wconf{A\lin B} &\bij& \wconf{A} \times \wconf{B}
\end{array}
\]
\end{lem}
\begin{proof}
Obtained by composition
\begin{eqnarray*}
\confn{A\lin B} &\bij& \sum_{b \in \min(B)} \confn{A \lin B_b}\\
&\bij& \sum_{b \in \min(B)} \conf{A} \times \confn{B_b}\\
&\bij& \conf{A} \times (\sum_{b \in \min(B)} \confn{B_b})\\
&\bij& \conf{A} \times \confn{B}
\end{eqnarray*}
using Lemmas \ref{lem:conf_with} and \ref{lem:conf_lin_wo}, together with the 
distributivity of cartesian product over disjoint union. The exact same
reasoning applies to complete symmetry classes.
\end{proof}

For $x_A \in \conf{A}$ and $x_B \in \confn{B}$, we write 
$x_A \lin x_B = (r_{A,B}^\lin)^{-1}(x_A, x_B) \in \confn{A\lin B}$, and
likewise $\x_A \lin \x_B \in \wconf{A\lin B}$ for $\x_A \in \wconf{A}$ and $\x_B
\in \wconf{B}$.

\subsubsection{Exponentials}
 \changed{We define $\oc(-)$, which will eventually extend to a linear
exponential comonad. It is the only source of non-trivial symmetries in
the games used for $\nPCF$. For a $-$-game $A$, $\oc A$ is understood
as an infinitary tensor of symmetric copies of $A$. The point is to
allow for duplication and weakening, and so in particular $\oc A$ is
\emph{not} strict. 
 } 

\begin{defi}\label{def:bang}
Consider $A$ a $-$-game. Then, we define the \textbf{bang} $\oc A$ with
the components:
\[
\begin{array}{rrcl}
\text{\emph{events:}} &
\ev{\oc A} &=& \mathbb{N} \times \ev{A}\\
\text{\emph{causality:}} &
(i, a_1) \leq_{\oc A} (j, a_2) &\Leftrightarrow&
i=j \,\wedge\, a_1 \leq_A a_2\\
\text{\emph{conflict:}} &
(i, a_1) \conflict_{\oc A} (j , a_2) &\Leftrightarrow&
i=j \,\wedge\, a_1 \conflict_A a_2\\
\text{\emph{symmetries:}} &
\theta \in \tilde{\oc A} &\Leftrightarrow& 
\exists \pi : \mathbb{N} \bij \mathbb{N},\quad
\exists (\theta_n)_{n \in \mathbb{N}} \in \tilde{A}^\mathbb{N}\,\\
&&&\forall (i, a) \in \dom(\theta), \theta(i, a) = (\pi(i),
\theta_i(a))\\
\text{\emph{polarities:}} &
\pol_{\oc A}(i, a) &=& \pol_A(a)\\
\text{\emph{positive symmetries:}} &
\theta \in \ptilde{\oc A} &\Leftrightarrow&
\exists (\theta_n)_{n\in \mathbb{N}} \in \ptilde{A}^{\mathbb{N}}\,,\\
&&&\forall (i, a) \in \dom(\theta), \theta(i, a) = (i, \theta_i(a))\\
\text{\emph{negative symmetries:}} &
\theta \in \ntilde{\oc A} &\Leftrightarrow&
\exists \pi : \mathbb{N} \bij \mathbb{N},\,
\exists (\theta_n)_{n \in \mathbb{N}} \in \ntilde{A}^{\mathbb{N}}\,,\\
&&&\forall (i, a) \in \dom(\theta), \theta(i, a) = (\pi(i),
\theta_i(a))\\
\text{\emph{payoff:}} &
\kappa_{\oc A}(\parallel_{i \in I} x_i) &=& \bigotimes_{i\in I}
\kappa_A(x_i) \qquad (I\subseteq \mathbb{N},\,\forall i \in I,\, x_i \neq
\emptyset)\\
&\kappa_{\oc A}(\emptyset) &=& 0
\end{array}
\]
where $\parallel_{i\in I} x_i = \biguplus_{i\in I} \{i\}\times x_i$.
This yields a
$-$-game $\oc A$.
\end{defi}

The definition of payoff uses implicitely that the tensor operation on $\{-1,
0, +1\}$ defined in Figure \ref{fig:op_payoff} is associative.
The $-$-game $\oc A$ is non-strict by intention:
Opponent is free to open \changed{any number of copies}, including zero.

\changed{
Again, by considering symmetry classes we recover the matching construction in $\Rel{\R}$.
}

\begin{lem}\label{lem:r_bang}
Consider $A$ a $-$-game. Then, we have bijections
\[
\begin{array}{rcrcl}
s^\oc_A &:& \wconf{\oc A} &\bij& \M_f(\wconfn{A})
\end{array}
\qquad
\begin{array}{rcrcl}
s^{\oc, \neq \emptyset}_A &:& \wconfn{\oc A} &\bij& \M^{\neq
\emptyset}_f(\wconfn{A})
\end{array}
\]
with $\wconfn{A}$ the non-empty complete classes, and $\M^{\neq
\emptyset}_f(X)$ the non-empty finite multisets.

In particular, if $A$ is a strict $-$-game, the first bijection
specializes to
\[
\begin{array}{rcrcl}
s^\oc_A &:& \wconf{\oc A} &\bij& \M_f(\wconf{A})\,.
\end{array}
\]
\end{lem}
\begin{proof}
We prove the first bijection. We first define $s^\oc_A$ on concrete
configurations of $\oc A$:
\[
\begin{array}{rcrcl}
s^\oc_A &:& \zconf{\oc A} &\to& \M_f(\wconfn{A})\\
&& \emptyset &\mapsto& []\\
&& \parallel_{i\in I} x_i &\mapsto& [\x_i \mid i \in I] \qquad
(\forall i \in I,\,x_i \neq \emptyset)
\end{array}
\]
where $\x_i$ denotes the symmetry class of $x_i$. It is straightforward
that if $x \sym_{\oc A} y$ then $s^\oc_A(x) = s^\oc_A(y)$, so $s^\oc_A$
lifts to $s^\oc_A : \wconf{\oc A} \to \M_f(\wconfn{A})$, keeping the same
notation.

\emph{Injective.} Assume $s^\oc_A(\parallel_{i\in I} x_i) =
s^\oc_A(\parallel_{j \in J} y_j)$ with every $x_i$ and $y_j$ non-empty.
We have
\[
[\x_i \mid i \in I] = [\y_j \mid j \in J]\,,
\]
meaning that there is a bijection $\pi : I \bij J$ such that for all $i
\in I$, we have $\x_i = \y_{\pi(i)}$. In turn, this means that for all
$i \in I$, there is some $\theta_i : x_i \sym_A y_{\pi(i)}$. Now,
completing $\pi$ to $\pi' : \mathbb{N} \bij \mathbb{N}$ arbitrarily
yields $\theta :\,\parallel_{i\in I} x_i \sym_{\oc A}\,\parallel_{j \in J}
y_j$ as required.

\emph{Surjective.} Consider $\mu \in \M_f(\wconfn{\oc A})$. If $\mu =
[]$, its pre-image is the empty complete symmetry class. Otherwise, write
$\mu = [\x_i \mid i \in I]$ with each $\x_i$ 
non-empty. For $i \in I$, fix some $x_i \in \x_i$.
\emph{W.l.o.g.} we may assume $I \subseteq_f \mathbb{N}$, so that 
$\parallel_{i\in I} x_i$ gives the required pre-image.

Clearly, $s^\oc_A$ restricts to $s^{\oc,
\neq \emptyset}_A : \wconfn{\oc A} \bij \M^{\neq
\emptyset}_f(\wconfn{A})$.
\end{proof}


Note that unless $A$ is strict, not all finite multisets in
$\M_f(\wconf{A})$ correspond to a complete symmetry class on $\oc A$.
For instance $[\emptyset]$, or any $[\emptyset, \dots, \emptyset]$ do
not.

\subsubsection{Interpretation of types and arenas} 
\changed{We give the complete} interpretation of types. \emph{Types} are interpreted as well-opened
$-$-games: for base types we use the well-opened games $\gbool$ and $\gnat$
defined in Section \ref{subsubsec:games_construction}, and we set
$\intr{A\to B} = \oc \intr{A} \lin \intr{B}$. \emph{Contexts} are
interpreted as strict $-$-games, with $\intr{x_1 : A_1, \dots, x_n :
A_n} = \with_{1\leq i \leq n} \intr{A_i}$.

Putting together Lemmas \ref{lem:r_basic}, \ref{lem:conf_with},
\ref{lem:conf_lin_wo} and \ref{lem:r_bang}, we immediately get:

\begin{lem}\label{lem:web_wconf}
For any type $A$ and context $\Gamma$, there are bijections:
\[
\begin{array}{rcrcl}
s^{\Ty}_A &:& \rintr{A} &\bij& \wconf{\intr{A}}
\end{array}
\qquad
\qquad
\begin{array}{rcrcl}
s^{\Ctx}_\Gamma &:& \mathcal{M}_f(\rintr{\Gamma}) &\bij& \wconf{\oc
\intr{\Gamma}}\,.
\end{array}
\]
\end{lem}

So the web of $A$ may  be regarded as the set of
complete symmetry classes of $\intr{A}$. 
The games obtained as the interpretation of types have a particular
shape:

\begin{defi}\label{def:arena}
An \textbf{arena} is a $-$-game $A$ satisfying:
\[
\begin{array}{rl}
\text{\emph{alternating:}} & \text{if $a_1 \imc_A a_2$, $\pol_A(a_1)
\neq
\pol_A(a_2)$,}\\
\text{\emph{forestial:}} & \text{if $a_1 \leq_A a$ and $a_2 \leq_A a$,
then $a_1
\leq_A a_2$ or $a_2 \leq_A a_1$.}
\end{array}
\]
\end{defi}

For any type $A$ and context $\Gamma$, it is straightforward
that the $-$-games $\intr{A}$ and
$\intr{\Gamma}$ are arenas. Arenas do not quite characterize the games
arising from the interpretation, but get close enough for the purposes of
this paper. More precisely, 
arenas guarantee that moves have at most one causal
immediate predecessor, hence recovering the notion of \emph{justifier}
familiar from more traditional game semantics \cite{cg3}. This will play
a very minor role in this paper: it is necessary for the
deadlock-freeness property used in Section \ref{subsubsec:wit_w_sym}.

\subsection{Strategies} Now that games are set, we introduce
\emph{strategies} used to interpret terms.

\subsubsection{Plain strategies}
We start by recalling the notion of \emph{causal strategy} in its
formulation in \cite{cg3} -- here, we only say \emph{strategy} as
it is the only sort of strategy we consider.

\begin{defi}\label{def:caus_strat}
A \textbf{prestrategy} on game $A$ comprises an
ess
$(\ev{\sigma}, \leq_{\sigma}, \conflict_\sigma, \tilde{\sigma})$
with
\[
\pr : \ev{\sigma} \to \ev{A}
\]
a function called the \textbf{display map}, subject to the following
conditions:
\[
\begin{array}{rl}
\text{\emph{rule-abiding:}} & \text{for all $x \in \conf{\sigma}$,
$\pr(x) \in \conf{A}$,}\\
\text{\emph{locally injective:}} & \text{for all $s_1, s_2 \in x \in
\conf{\sigma}$, if $\pr(s_1) = \pr(s_2)$ then $s_1 =
s_2$,}\\
\text{\emph{symmetry-preserving:}} &
\text{for all $\theta \in \tilde{\sigma}$, $\pr(\theta) =
\{(\pr(s_1), \pr(s_2)) \mid (s_1, s_2) \in \theta\} \in
\tilde{A}$,}\\
\text{\emph{$\sim$-receptive:}} &
\text{for all $\theta : x \sym_\sigma y$, and extensions $x
\vdash_\sigma s_1^-$,
$\pr(\theta) \vdash_{\tilde{A}} (\pr(s^-_1), a^-_2)$,}\\
&\text{there is a unique $s_2^- \in \ev{\sigma}$ s.t.
$\theta \vdash_{\tilde{\sigma}} (s^-_1, s^-_2)$ and
$\pr(s_2^-) = a_2^-$,}\\
\text{\emph{thin:}} & 
\text{for all $\theta : x \sym_\sigma y$, and extension $x
\vdash_\sigma s_1^+$,}\\
&\text{there is a \emph{unique} extension $y \vdash_\sigma s_2^+$
such that $\theta \vdash_{\tilde{\sigma}} (s_1^+, s_2^+)$.}
\end{array}
\]

Additionally, we say that $\sigma$ is a \textbf{strategy} if it
satisfies the further three conditions:
\[
\begin{array}{rl}
\text{\emph{negative:}} & \text{for all $s \in \ev{\sigma}$, if $s$ is
minimal then $s$ is negative,}\\
\text{\emph{courteous:}} & \text{for all $s_1 \imc_\sigma s_2$, if
$\pol(s_1) = +$ or $\pol(s_2) = -$ then
$\pr(s_1) \imc_A \pr(s_2)$,}\\
\text{\emph{receptive:}} & \text{for all $x \in \conf{\sigma}$, for all
$\pr(x) \vdash_A a^-$,}\\
&\text{there is a unique $x \vdash_\sigma s^- \in \conf{\sigma}$ such
that
$\pr(s) = a$,}
\end{array}
\]

We write $\sigma : A$ to mean that $\sigma$ is a strategy on game $A$.
\end{defi}

We disambiguate some notations used in the definition. First,
$\sigma$ implicitly comes with polarities,
imported from $A$ as $\pol_\sigma(s) = \pol_A(\pr(s))$. We often tag
events to indicate their polarity as in $s^-, s^+$; the sign is
not considered part of the variable name but conveys the  polarity information. We also used
the enabling relation on isomorphism families, defined
by $\theta \vdash_{\tilde{A}} (a_1, a_2)$ iff $(a_1, a_2)
\not \in \theta$ and $\theta \cup \{(a_1, a_2)\} \in \tilde{A}$.

A strategy is a causal presentation, in one global object, of the entire
computational behaviour of a program on the interface with its execution
environment. While the \emph{events} or \emph{moves} in $\ev{\sigma}$
are not technically moves of the game, the \emph{display map}
$\pr_\sigma$ associates to any move of $\sigma$ the corresponding move
in the game, and is subject to adequate conditions ensuring
compatibility with the structure. More precisely, 
conditions \emph{rule-abiding}, \emph{locally injective} and
\emph{symmetry-preserving} together amount to $\pr$ being a \textbf{map
of event structures with symmetry} \cite{DBLP:journals/entcs/Winskel07} -- the adequate
simulation maps between ess. 

Conditions \emph{courteous} and
\emph{receptive} are the usual conditions for concurrent strategies
\cite{DBLP:conf/lics/RideauW11,DBLP:journals/lmcs/CastellanCRW17},
expressing that the strategy should be invariant under asynchronous
delay. The condition \emph{$\sim$-receptive} forces
strategies to consider as symmetric pairs of Opponent events symmetric
in the game, and hence to treat them uniformly.  Finally, \emph{thin} is
a minimality condition forcing strategies to pick \emph{one} canonical
representative up to symmetry for positive moves. For further
explanations and discussions on those conditions, see \cite{cg2}.

\begin{figure}
\[
\xymatrix@C=1pt@R=10pt{
&\oc (\oc \gbool&&&&\lin&&&&\gbool)&&&&\lin&&&&\gbool\\
&&&&&&&&&&&&&&&&&\qu^-
	\ar@{-|>}[dllllllll]\\
&&&&&&&&&\qu^+_{\grey{0}}
	\ar@{.}@/^/[urrrrrrrr]
	\ar@{-|>}[dll]
	\ar@{-|>}[drr]
	\ar@{-|>}[dllllllll]\\
&\qu^-_{\grey{0,i}}
	\ar@{-|>}[d]
	\ar@{.}@/^/[urrrrrrrr]
&&&&&&\ttrue^-_{\grey{0}}
	\ar@{-|>}[d]
	\ar@{~}[rrrr]
	\ar@{.}@/^/[urr]
&&&&\tfalse^-_{\grey{0}}
	\ar@{-|>}[d]
	\ar@{.}@/_/[ull]\\
&\ttrue^+_{\grey{0,i}}
	\ar@{.}@/^/[u]
&&&&&&\qu^+_{\grey{1}}
	\ar@{-|>}[dl]
	\ar@{-|>}[dr]
	\ar@{-|>}[dlllllll]
	\ar@{.}@/^/[uuurrrrrrrrrr]
&&&&\qu^+_{\grey{1}}
	\ar@{-|>}[dl]
	\ar@{-|>}[dr]
	\ar@{-|>}[dlllllllll]
	\ar@{.}@/^/[uuurrrrrr]\\
\qu^-_{\grey{1,j}}
	\ar@{-|>}[d]
	\ar@{.}@/^/[urrrrrrr]
&&\qu^-_{\grey{1,k}}
	\ar@{-|>}[d]
	\ar@{.}@/^/[urrrrrrrrr]
&&&&\ttrue^-_{\grey{1}}
	\ar@{~}[rr]
	\ar@{.}@/^/[ur]&&
\tfalse^-_{\grey{1}}
	\ar@{-|>}[drrrrrrrr]
	\ar@{.}@/_/[ul]
&&\ttrue^-_{\grey{1}}
	\ar@{~}[rr]
	\ar@{.}@/^/[ur]
	\ar@{-|>}[drrrrrrrr]
&&\tfalse^-_{\grey{1}}
	\ar@{.}@/_/[ul]\\
\ttrue^+_{\grey{1, j}}
	\ar@{.}@/^/[u]&&
\ttrue^+_{\grey{1, k}}
	\ar@{.}@/^/[u]
&&&&&&&&&&&&&&\ttrue^+
	\ar@{.}@/^/[uuuuur]
&&\ttrue^+
	\ar@{.}@/_/[uuuuul]
}
\]
\caption{A strategy $\sigma : \oc (\oc \gbool \lin \gbool) \lin
\gbool$}
\label{fig:ex_strat}
\end{figure}

As an illustration, we show a strategy in Figure \ref{fig:ex_strat}.  We
take advantage of this to introduce our convention for drawing
strategies.  Ignoring dotted lines, the diagram represents the event
structure $(\ev{\sigma}, \leq_\sigma, \conflict_\sigma)$, via its
immediate causality relation $\imc$. For convenience, each node is
labeled with the corresponding move in the game as observed through
$\pr_\sigma$ -- as each move in the game originates in one of the ground
types appearing in the type, we attempt as much as possible to place
moves in the diagram in the corresponding column. The immediate
causality in the game is conveyed through the dotted lines. The
grey subscripts correspond to the \emph{copy indices}, \emph{i.e.} the
tags introduced in Definition \ref{def:bang} to address the components
of $\oc A$. For moves under several $\oc$, we write the copy indices in
a sequence ranging from the outermost to the innermost $\oc$. The
representation is symbolic: the actual strategy comprises branches as
shown for all $\grey{i}, \grey{j}, \grey{k} \in \mathbb{N}$. Finally,
the one component of $\sigma$ missing in this representation is the
isomorphism family $\tilde{\sigma}$ which is too unwieldy to draw -- in
this particular case it consists of all order-isomorphisms changing only
copy indices.

Only a few strategies are definable via $\nPCF$: for
instance, in \cite{cg3} the same strategies are used to model a
higher-order concurrent language with shared memory. One could, with an
adequate notion of innocence, characterize exactly those strategies
definable through $\nPCF$ -- this is done in
\cite{DBLP:phd/hal/Castellan17}. For our present purposes this is
not required; nevertheless we have to add two further conditions \changed{for the collapse functor to exist.} 

\subsubsection{Visible strategies} \emph{Visibility} is a locality
property for the control flow. First, we define:

\begin{defi}
Consider $E$ an event structure. 

A \textbf{grounded causal chain (gcc)} in $E$ is
$\rho = \{\rho_1, \dots, \rho_n\} \subseteq \ev{E}$ forming
\[
\rho_1 \imc_E \dots \imc_E \rho_n
\]
a chain with $\rho_1$ minimal with respect to $\leq_E$.
We write $\gcc(E)$ for the gccs in $E$.
\end{defi}

If $\sigma$ is forestial, then $\gcc(\sigma)$ comprises 
simply its \emph{finite branches}. But in general, a gcc $\rho \in
\gcc(\sigma)$ is not down-closed; it may be seen as an individual
\emph{thread}, and $\sigma$ may be regarded as a collection of such
threads with the information of their non-deterministic
branchings, as well as the points where they causally fork or merge.
Visibility states that each thread respects the local scope, \emph{i.e.}
it may not use resources introduced in another thread:

\begin{defi}\label{def:visible}
A strategy $\sigma : A$ is \textbf{visible} if for all $\rho \in
\gcc(\sigma)$, $\pr_{\sigma}(\rho) \in \conf{A}$.
\end{defi}

In this paper, we require visibility as it restricts the behaviour of
strategies just enough so that their interactions never deadlock.
Visibility is \changed{far from sufficient in the way of} capturing the behaviour of
 $\nPCF$ terms -- for those the causal structure is always
forest-shaped, which trivially entails visibility. We work with
visibility rather than a forest-shaped causality as it costs us nothing, and
makes our collapse theorem slightly more general.

Much more information on visibility may be found in \cite{cg3}.

\subsubsection{Exhaustive strategies} \emph{Exhaustivity} is an elegant
mechanism due to Melli\`es \cite{DBLP:conf/lics/Mellies05}. \changed{As
discussed previously, it ensures that strategies are \emph{linear}
rather than \emph{affine}}. It also enforces a form of
well-bracketing,  
forcing strategies to respect the call stack discipline. 

First, we must define a notion of \emph{stopping state} for a strategy.

\begin{defi}
Consider $A$ a game, and $\sigma : A$ a strategy.
A configuration $x \in \conf{\sigma}$ is \textbf{$+$-covered} if for all
$m \in x$, if $m$ is maximal in $x$ (with respect to $\leq_\sigma$),
then $\pol_\sigma(m) = +$.

We write $\confp{\sigma}$ for the set of all $+$-covered configurations
of $\sigma$.
\end{defi}

In other words, a configuration is $+$-covered if no Opponent move is
left unresponded.

\begin{defi}\label{def:exh_strat}
Consider $A$ a game, and $\sigma : A$ a strategy. We set the condition:
\[
\begin{array}{rl}
\text{\emph{exhaustive:}} & \text{for all $x \in \confp{\sigma}$,
$\kappa_A(\pr_\sigma(x)) \geq 0$.}
\end{array}
\]
\end{defi}

For example, $\sigma : \gbool_1 \lin \gbool_2$ (indices
for disambiguation) \emph{cannot} reply on $\gbool_2$ without evaluating
its argument, as $\kappa_{\gbool_1 \lin \gbool_2}(\{\qu_2^-,
b_2^+\}) = \kappa_{\gbool^{\perp}}(\emptyset) \parr
\kappa_{\gbool}(\{\qu^-, b^+\}) = -1 \parr 0 = -1$.

In this paper we include all details of the compositionality of
exhaustivity, which to the best of our knowledge do not appear in any
published source.

%
%
%

\subsection{Composition of plain strategies} Postponing for now the
stability under composition of visibility and exhaustivity, we
recall the composition of plain strategies.

For games $A$ and $B$, a \textbf{strategy from $A$ to $B$} is
a strategy on $A^\perp \parr B$, also written $A \vdash B$. Fix from
now on $A, B$ and $C$, and $\sigma : A \vdash B$, $\tau : B \vdash C$.
We aim to define $\tau \odot \sigma : A \vdash C$.  
The concrete definition of composition is covered at length elsewhere
\cite{cg2}; instead we give a characterization in terms of its
\emph{states}.

\subsubsection{Synchronization} Given configurations $x^\sigma \in
\conf{\sigma}$, $x^\tau \in
\conf{\tau}$, by convention we write
\[
\pr_\sigma(x^\sigma) = x^\sigma_A \parallel x^\sigma_B \in
\conf{A\vdash B}\,,
\qquad
\qquad
\pr_\tau(x^\tau) = x^\tau_B \parallel x^\tau_C \in \conf{B \vdash
C}\,,
\]
for the corresponding projections to the game. In defining composition,
the first stage is to capture when such configurations $x^\sigma \in
\conf{\sigma}$ and $x^\tau \in \conf{\tau}$ may successfully
\emph{synchronise}.

\begin{defi}\label{def:caus_comp}
Consider two configurations $x^\sigma \in \conf{\sigma}$ and $x^\tau \in
\conf{\tau}$.
They are \textbf{causally compatible} if \emph{(1)} 
\textbf{matching}: $x^\sigma_B = x^\tau_B = x_B$; and \emph{(2)} if the
bijection
\[
\varphi_{x^\sigma, x^\tau}
\quad
:
\quad
x^\sigma \parallel x^\tau_C 
\quad
\stackrel{\pr_\sigma \parallel x^\tau_C}{\simeq} 
\quad
x^\sigma_A \parallel x_B \parallel x^\tau_C 
\quad
\stackrel{x^\sigma_A \parallel \pr_\tau^{-1}}{\simeq}
\quad
x^\sigma_A \parallel x^\tau\,,
\]
obtained by composition using local injectivity of $\pr_\sigma$ and
$\pr_\tau$, is \textbf{secured}, \emph{i.e.} the relation
\[
(m, n) \vartriangleleft (m', n') 
\qquad
\Leftrightarrow
\qquad
m <_{\sigma \parallel C} m'
\quad
\vee
\quad
n <_{A \parallel \tau} n'\,,
\]
defined on (the graph of) $\varphi_{x^\sigma, x^\tau}$ by importing
causal constraints of $\sigma$ and $\tau$, is acyclic.
\end{defi}

Securedness eliminates deadlocks: two matching $x^\sigma \in
\conf{\sigma}$, $x^\tau \in \conf{\tau}$ agree on the final state in
$B$, but $\sigma$ and $\tau$ may impose
incompatible constraints making it unreachable.


\subsubsection{Plain composition} It turns out that up to an adequate
isomorphism, there is a \emph{unique} strategy $\tau \odot \sigma :
A\vdash C$ whose configurations are such synchronizations, in a way
compatible with symmetry. In order to state this we first disambiguate
our notion of isomorphism.

\begin{defi}
Consider $A$ a game, and $\sigma, \tau : A$ two strategies. 

An \textbf{isomorphism} $\varphi : \sigma \iso \tau$ is an invertible
map of ess such that $\pr_\tau \circ \varphi = \pr_\sigma$.
\end{defi}

Before we introduce composition, we extend two earlier notions
from configurations to symmetries. A symmetry $\theta
\in \tilde{\sigma}$ is \textbf{$+$-covered} if either of $\dom(\theta)$
or $\cod(\theta)$ is: as symmetries are order-isomorphisms and preserve
polarities, the difference is immaterial. We write $\tildep{\sigma}$ for
the set of $+$-covered symmetries of $\sigma$.
Likewise, $\theta^\sigma \in
\tilde{\sigma}$ and $\theta^\tau \in \tilde{\tau}$ are \textbf{causally
compatible} if they are \emph{matching}, \emph{i.e.} $\pr_\sigma
\theta^\sigma = \theta^\sigma_A \parallel \theta^\sigma_B$ and $\pr_\tau
\theta^\tau = \theta^\tau_B \parallel \theta^\tau_C$ with
$\theta^\sigma_B = \theta^\tau_B$; and \emph{secured}, \emph{i.e.}
either $\dom(\theta^\sigma), \dom(\theta^\tau)$ or $\cod(\theta^\sigma),
\cod(\theta^\tau)$ are.

Now we state the following proposition, whose proof may be found in
\cite{cg3}:

\begin{prop}\label{prop:char_comp}
Consider $\sigma : A \vdash B$ and $\tau : B \vdash C$ two strategies.

Then, there is a strategy $\tau \odot \sigma : A \vdash C$, unique up
to iso, such that there are order-isos:
\[
\begin{array}{rcrcl}
(- \odot -) 
&\!\!:\!\!&
\{(x^\tau, x^\sigma) \in \confp{\tau} \times \confp{\sigma} \mid
\text{$x^\sigma$ and $x^\tau$ causally compatible}\} 
&\!\!\simeq\!\!&
\confp{\tau \odot \sigma}\\
(- \odot -)
&\!\!:\!\!&
\{(\theta^\tau, \theta^\sigma) \in \tildep{\tau} \times
\tildep{\sigma} \mid \text{$\theta^\sigma$ and $\theta^\tau$ causally
compatible}\}
&\!\!\simeq\!\!&
\tildep{\tau\odot \sigma}
\end{array}
\]
commuting with $\dom, \cod$, and s.t. for $\theta^\sigma \in \tildep{\sigma}$
and $\theta^\tau \in \tildep{\tau}$ causally compatible,
\[
\pr_{\tau\odot \sigma}(\theta^\tau \odot \theta^\sigma) =
\theta^\sigma_A \parallel \theta^\tau_C\,.
\]
\end{prop}

So $\tau \odot \sigma$ is the unique (up
to iso) strategy whose configurations correspond to
matching pairs for which the causal constraints
imposed by $\sigma$ and $\tau$ are compatible. 

\changed{The idea of considering matching pairs of configurations
brings us close to the composition of relations, spans, or profunctors.
In general, however, the composition of strategies is more restricted,
because of \emph{causal compatibility}. This is a
well-known feature of game semantics, which allows for the
interpretation of more sophisticated language constructs. 
But for the strategies of this paper, matching configurations are
always causally compatible. This follows from \emph{visibility},
which we will exploit to prove that composition never \emph{deadlocks};
see Section \ref{subsubsec:wit_w_sym}. This is crucial for
constructing a functor to the relational model. }


\subsubsection{Copycat} As we have defined a notion of composition, it
is natural to introduce here the accompanying identity, the
\emph{copycat strategy}. For any game $A$, the copycat strategy $\cc_A :
A \vdash A$ is an asynchronous forwarder: any Opponent move on either
side enables (\emph{i.e.} is the unique causal dependency for) the
corresponding event on the other side.

We only give the definition of copycat on arenas, as it is sufficient
and slightly simpler:

\begin{defi}\label{def:copycat}
For each arena $A$, the \textbf{copycat strategy} $\cc_A : A \vdash
A$ comprises:
\[
\begin{array}{rcl}
\ev{\cc_A} &=& \ev{A \vdash A}\\
\pr_{\cc_A}(i, a) &=& (i, a)\\
(i, a) \leq_{\cc_A} (j,a') &\Leftrightarrow& \text{$a <_A a'$; or $a =
a'$, $\pol_{A \vdash A}(i, a) = -$ and $\pol_{A^\perp \parallel
A}(j, a') = +$}\\
(i, a) \conflict_{\cc_A} (j, a') &\Leftrightarrow& a \conflict_A a'\,,
\end{array}
\]
with symmetries those bijections of the form $\theta_1 \parallel
\theta_2 : x_1 \parallel x_2 \sym_{\cc_A} y_1 \parallel y_2$ such that
\[
\theta_1 : x_1 \sym_A y_1\,,
\quad
\theta_2 : x_2 \sym_A y_2\,,
\quad
\text{and}
\quad
\theta_1 \cap \theta_2 : x_1 \cap x_2 \sym_A y_1 \cap y_2\,.
\]
\end{defi}

As for composition, we shall rely on a characterization
of its $+$-covered configurations:

\begin{prop}\label{prop:cc_pcov}
Consider $A$ any arena. Then, we have:
\[
\begin{array}{rcl}
\confp{\cc_A} &=& \{x_A \parallel x_A \in \conf{A\parallel A} \mid x_A
\in \conf{A}\}\,\\ 
\tildep{\cc_A} &=& \{\theta_A \parallel \theta_A \in \tilde{A\parallel
A}\mid \theta_A \in
\tilde{A}\}\,.
\end{array}
\]
\end{prop}

In $+$-covered configurations, any Opponent moves must have caused a
Player move; so they must all have been forwarded. This makes
$+$-covered configurations completely balanced; once again it hints at
the link with relational semantics as through its $+$-covered
configurations, $\cc_A$ matches the identity relation on configurations.

\subsubsection{Congruence} Arenas and strategies up to isomorphism form
\changed{a bicategory, where the 2-cells are given by a notion of \emph{maps between strategies}: these are defined as maps of event structures that commute with the display maps \cite{cg2}. For the purposes of semantics, we usually consider strategies up to isomorphism, and this gives a category. Unfortunately this category is not the right one: the requirement that isomorphisms commute with the display map is too strict. For $\oc$ to satisfy the comonad laws, we require a more permissive equivalence relation between
strategies, which allows for the choice of copy indices to be different in each strategy. Formally, this weaker notion of isomorphism is given by a map of event structures which commutes with the display maps \emph{up to positive symmetry}:}
\begin{defi}\label{def:pos_iso}
Consider $\sigma, \tau : A$ two causal strategies on arena $A$.

A \textbf{positive isomorphism} $\varphi : \sigma \simstrat \tau$ is
an isomorphism of ess satisfying
\[
\pr_\tau \circ \varphi \sim^+ \pr_\sigma\,,
\]
\emph{i.e.} for all $x \in \conf{\sigma}$, $\{(\pr_\sigma(s),
\pr_\tau\circ \varphi(s)) \mid s \in x\} \in \ptilde{A}$. If there is a
positive iso $\varphi : \sigma \simstrat \tau$ we say $\sigma$ and
$\tau$ are \textbf{positively isomorphic}, and write $\sigma \simstrat
\tau$. 
\end{defi}

This means that $\sigma$ and $\tau$ are the same up to renaming of
events. Intuitively this renaming might cause a reindexing of
positive events, but it must keep the copy indices of negative events
unchanged.
Crucially, \changed{the induced notion of} equivalence is preserved by composition, \changed{\emph{i.e.} it is a congruence. }

\begin{prop}\label{prop:congruence}
Consider $\sigma, \sigma' : A \vdash B$, $\tau, \tau' : B \vdash
C$ s.t. $\sigma \simstrat \sigma'$ and $\tau \simstrat \tau'$.

Then, we have
$\tau \odot \sigma 
\simstrat
\tau' \odot \sigma'$.
\end{prop}

The proof is fairly elaborate \cite{cg2}. Details are out of scope,
however the main argument is reviewed later as Proposition
\ref{prop:sync_sym}: from two configurations able to synchronize
\emph{up to symmetry}, one can extract symmetric configurations
synchronizing \emph{on the nose}.

\subsection{Composition of visible and exhaustive strategies.} 
\changed{
The fact that visible strategies compose is detailed in
\cite{cg3}, and we do not repeat the details here. Note that for any arena $A$, the copycat strategy $\cc_A : A \vdash A$ is visible, and for arenas $A$, $B$ and $C$ and strategies $\sigma : A \vdash B$ and $\tau :
B \vdash C$, if $\sigma$ and $\tau$ are visible then so is $\tau\odot
\sigma$. We emphasize that these results apply to arenas but not to arbitrary games; in particular, the proof exploits that every event has a unique
immediate predecessor in the game \cite{cg3}.}

 The compositionality of exhaustivity does
not appear anywhere \changed{in the literature}, so we include the details, straightforward as
they are. We start with copycat.

\begin{prop}
For any arena $A$, $\cc_A : A \vdash A$ is exhaustive.
\end{prop}
\begin{proof}
By Lemma \ref{prop:cc_pcov}, $+$-covered configurations of copycat are
$x_A \parallel x_A \in \confp{\cc_A}$
for $x_A \in \conf{A}$. But then $\kappa_{A\vdash A}(x_A \parallel x_A)
= (-\kappa_A(x_A)) \parr \kappa_A(x_A)$, always non-negative.
\end{proof}

We now prove that exhaustive strategies are stable under
composition.

\begin{prop}
Consider $A, B$ and $C$ games, and $\sigma : A\vdash B$, $\tau : B
\vdash C$ exhaustive.

Then, $\tau \odot \sigma : A \vdash C$ is exhaustive.
\end{prop}
\begin{proof}
By Proposition \ref{prop:char_comp}, any $+$-covered configuration of
$\tau \odot \sigma$ has the form $x^\tau \odot x^\sigma \in \confp{\tau
\odot \sigma}$ for $x^\sigma \in \confp{\sigma}$ and $x^\tau \in
\confp{\tau}$ causally compatible. We write projections as:
\[
\pr_\sigma x^\sigma = x_A \parallel x_B\,,
\qquad
\qquad
\pr_\tau x^\tau = x_B \parallel x_C
\]
and $\pr_{\tau \odot \sigma}(x^\tau \odot x^\sigma) = x_A
\parallel x_C$. If $\kappa_A(x_A) = -1$ or $\kappa_C(x_C) = 1$, then
$\kappa_{A\vdash C}(x_A \parallel x_C) = 1$ and we are done. Hence, assume
$\kappa_A(x_A) \geq 0$ and $\kappa_C(x_C) \geq 0$. If
$\kappa_A(x_A) = 1$, then we must have $\kappa_B(x_B) = 1$ as well since
$\kappa_{A\vdash B}(x_A \parallel x_B) \geq 0$ since $\sigma$ is
exhaustive. But then, with the same reasoning by exhaustivity of $\tau$,
$\kappa_C(x_C) = 1$, contradiction.
Symmetrically, if $\kappa_C(x_C) = -1$ we may deduce that $\kappa_A(x_A)
= -1$, contradiction. The only case left has $\kappa_A(x_A) = 0$ and
$\kappa_C(x_C) = 0$, so that $\kappa_{A\vdash C}(x_A \parallel x_C)=0$ as
needed.
\end{proof}

%
%
%
%

\subsection{Categorical structure} 
\changed{
We now outline the categorical structure of our model, as required for $\nPCF$, exploiting the various constructions on games introduced in \ref{subsec:games}. Due to the 2-dimensional nature of the model (isomorphisms between strategies play an important role), there are subtleties in the presentation of this structure, which we explain now. 

Strategies on thin concurrent games form a bicategory, where
associativity and unit laws for composition hold only up to iso. Above
we introduced a weaker notion, positive isomorphism, which is useful
for considering strategies up to a choice of copy indices. This gives
rise to another bicategory with the same objects and morphisms, but
more 2-cells.

A natural step could be to present the categorical structure at this
level (as in \cite{paquet2020probabilistic}). This is mathematically
important but technical, because many additional
coherence laws involving 2-cells must be verified. For this paper this
is not necessary as any 2-dimensional structure disappears in the
collapse to $\Rel{\R}$. Thus we compromise and define a model
consisting of objects (arenas), morphisms (exhaustive, visible
strategies), and an equivalence relation on each hom-set ($\approx$),
such that the laws for composition hold up to $\approx$. In other
words, we record which strategies are positively isomorphic, but forget
the isos.  

We call this a $\sim$-category. There is a general theory of these, in which the usual
coherence laws for morphisms hold only up to equivalence\footnote{We
note that a $\sim$-category is an enriched bicategory
\cite{garner2016enriched}, where the enrichment is over a 2-category of
sets with equivalence relations, equivalence-preserving maps, and
equivalence between maps, with the latter defined pointwise. This gives
a formal justification for our definitions.}. In particular there are
canonical notions of ($\sim$-)functor, ($\sim$-)comonad, monoidal
$\sim$-category, and so on, which we use below. 
}

\subsubsection{$\sim$-categories} We start by defining
\emph{$\sim$-categories}:
%
%

\begin{defi}
A 
(small) 
\textbf{$\sim$-category} $\C$ consists in a set of
\emph{objects} $\C_0$; for each $A, B$, a set of \emph{morphisms}
$\C(A, B)$ with an equivalence relation $\sim$; a \emph{composition}
operation 
\[
(- \circ -) : \C(B, C) \times \C(A, B) \to \C(A, C)
\]
for all $A, B, C \in \C_0$; an \emph{identity} morphism $\id_A \in \C(A, A)$ for
all $A \in \C_0$, subject to:
\[
\begin{array}{rl}
\text{\emph{associativity:}} & \text{for all $f \in \C(A, B), g \in
\C(B, C), h \in \C(C, D)$, $(h \circ g) \circ f \sim h \circ (g \circ
f)$,}\\
\text{\emph{identity:}} & \text{for all $f \in \C(A, B)$, $\id_B \circ f
\sim f \circ \id_A \sim f$,}\\
\text{\emph{congruence:}} & \text{for all $f \sim f' \in \C(A, B)$, $g
\sim g' \in \C(B, C)$, $g' \circ f' \sim g \circ f$.}
\end{array}
\]
\end{defi}

In particular, the development above already gives us our main
$\sim$-category of interest:

\begin{prop}\label{prop:sim_cat}
There is $\Strat$, a $\sim$-category with arenas as objects; morphisms
from $A$ to $B$ the exhaustive, visible strategies on $A \vdash B$; and
equivalence relation $\simstrat$.
\end{prop}
\begin{proof}
The required structure was provided above; it remains to prove that
associativity of composition and neutrality of copycat hold up to
$\simstrat$. In other words, we must provide
\[
\begin{array}{rcrcl}
\alpha_{\sigma, \tau, \delta} &:& (\delta \odot \tau) \odot \sigma
&\simstrat& \delta \odot (\tau \odot \sigma)\\
\lambda_\sigma &:& \cc_B \odot \sigma &\simstrat& \sigma\\
\rho_\sigma &:& \sigma \odot \cc_A &\simstrat& \sigma\,,
\end{array}
\]
for all $\sigma : A \vdash B$, $\tau : B \vdash
C$ and $\delta : C \vdash D$.
On $+$-covered configurations, they are:
\[
\begin{array}{rcrcl}
\alpha_{\sigma, \tau, \delta} &:& \confp{(\delta \odot \tau) \odot
\sigma} &\iso& \confp{\delta \odot (\tau \odot \sigma)}\\
&& (x^\delta \odot x^\tau) \odot x^\sigma &\mapsto & x^\delta \odot
(x^\tau \odot x^\sigma)
\end{array}
\quad
\begin{array}{rcrcl}
\lambda_\sigma &:& \confp{\cc_B \odot \sigma} &\iso& 
\confp{\sigma}\\
&& (x_B \parallel x_B) \odot x^\sigma & \mapsto& x^\sigma
\end{array}
\]
and symmetrically for $\rho_\sigma$; using Propositions
\ref{prop:char_comp} and \ref{prop:cc_pcov}. In fact, these bijections 
between configurations are sufficient to define the required
isomorphisms of strategies (provided one checks that they are
order-isomorphisms, compatible with symmetry, and that they commute with
display maps) -- see \cite{cg3}. In this paper, we 
only need the above characterization of their action on configurations.
Congruence is Proposition \ref{prop:congruence}. 
\end{proof}

Functors between $\sim$-categories \changed{must preserve} $\sim$, and preserve composition and identities up to $\sim$.
Of course, any $\sim$-category quotients to a category whose morphisms
are equivalence classes. But it is preferable to refrain from
quotienting: this way, the interpretation yields a concrete strategy
rather than an equivalence class. This is particularly relevant for
recursion which involves an complete partial order on concrete
strategies, whereas it is unknown if positive isomorphism classes
satisfy the adequate completeness properties. Note that this subtlety is
present in all game semantics involving explicit copy indices, including
 in AJM games \cite{ajm}, though it is usually handled implicitly.

\subsubsection{Relative Seely categories.} 
\changed{
Following Section~\ref{subsubsec:basic_constructions}, some
constructions are only available for \emph{strict} games:
$A\with B$ is only defined when $A$ and $B$ are strict
(Definition~\ref{def:with}), and $A \lin B$ is only defined when $B$ is
strict (Definition~\ref{def:linearfunctionspace}).  
As we explain in more details below, this means that some of the Seely
category structure in $\Strat$ only exists \emph{relative} to the
inclusion functor $\Strat_s \hookrightarrow \Strat$, where $\Strat_s$
is the full sub-$\sim$-category of \emph{strict} arenas. We introduce a
notion of \emph{relative Seely category}, which generalises Seely
categories, in which some constructions ($\with$ and $\lin$) are only
available for a sub-family of objects.

Additionally, while $\oc$ is always available in $\Strat$
(there is a comonad $\oc : \Strat \to \Strat$), the collapse of
Section~\ref{sec:coll_rw} only preserves $\oc$ on strict
arenas. So to precisely capture the logical content of this
collapse operation, it makes sense to consider $\oc$ as a
\emph{relative comonad} $\Strat_s \hookrightarrow \Strat$, in the sense
of \cite{relativemonads}. The axioms for relative Seely categories
ensure that the induced Kleisli category (with objects are the strict
ones), is cartesian closed. 

Relative Seely categories model the fragment of Intuitionistic Linear Logic with
\begin{eqnarray*}
S, T &::=& \top \mid S \with T \mid A \lin S\\
A, B &::=& 1 \mid A \tensor B \mid S \mid \oc S
\end{eqnarray*}
as formulas, separated into \emph{strict} $S, T$ and \emph{general} $A,
B$ formulas. This corresponds to the logical structure preserved by our collapse; note that this covers all formulas involved in Girard's call-by-name translation for $\nPCF$ types.  

In the following we make use of the standard notions of relative adjunctions and relative comonads. For the reader unfamiliar with these, the definition also contains explicit data. For full details see Appendix~\ref{app:relative}. 

}
\begin{defi}  
\label{def:relativeseely}
\changed{
A \textbf{relative Seely category} is a symmetric monoidal category $(\C, \tensor, 1)$ equipped with a full subcategory $\C_s$ together with the following data and axioms:

\begin{itemize}
\item $\C_s$ has finite products $(\with, \top)$ preserved by the inclusion functor $J : \C_s \hookrightarrow \C$.
\item For every $B \in \C$ there is a functor $B \lin - : \C_s \to \C_s$, such that there is a natural bijection 
\[
\Lambda(-) : \C(A\tensor B, S) \bij \C(A, B \lin S).
\]
for every $A \in \C$ and $S \in \C_s$. In other words, the functors 
\[
- \tensor B : \C \to \C \qquad \qquad J(B \lin -) : \C_s \to \C 
\]
 form a $J$-relative adjunction $- \tensor B \dashv_J J(- \lin B)$. 
\item There is a $J$-relative comonad $\oc : \C_s \to \C$. Concretely we have, for every $S \in \C_s$, an object $\oc S \in \C$ and a morphism $\der_S : \oc S \to S$, and for every $\sigma: \oc S \to T$, a \textbf{promotion} $\sigma^\dagger : \oc S \to \oc T$, subject to three axioms \cite{relativemonads}.

\item The functor $\oc : \C_s \to \C$ is symmetric strong monoidal
$(\C_s, \with, \top) \to (\C, \tensor, 1)$, so there are 
\[m_0 : 1 \to \oc \top \qquad \qquad m_{S, T} : \oc S \tensor \oc T \to \oc(S \with T)\]
isos for $S, T \in \C_s$, natural in $S, T$ and satisfying the axioms for Seely categories \cite{panorama}. 
\end{itemize}
}
\end{defi}
\changed{
Note that whenever $\C_s = \C$ this is precisely a Seely category in
the usual sense; in particular any Seely category is canonically a
relative Seely category. 

For any relative Seely category, the Kleisli category associated with
the relative comonad $\oc$ is cartesian closed. This category, denoted
$\C_\oc$, has objects those of $\C_s$, and $\C_\oc(S, T)  = \C(\oc S,
T)$.  The proof is essentially as for Seely categories;
details are in Appendix~\ref{app:kleisli}. 
\begin{lem}
For a relative Seely category $\C$, the Kleisli category $\C_\oc$ is
cartesian closed with finite products given as in $\C_s$, and function
space $S \tto T = \oc S \lin T$.  
\end{lem}

For $\Strat$ these definitions must be taken in
$\sim$-categorical form, this is what we do next. The generalisation is
very straightforward, but we give the main definitions along the
way.
}

\subsubsection{Symmetric monoidal structure} We first extend the
symmetric monoidal structure.

Symmetric monoidal
$\sim$-categories are defined as expected; where the usual data 
\changed{
required to preserve $\sim$ and satisfy laws up to $\sim$.
For $\sim$-categories $\C$ and $\D$, their \emph{product} $\C \times \D$ has objects $\C_0\times \D_0$, morphisms $(A, B) \to (C, D)$ the pairs $(f, g) \in \C(A, C) \times \D(B, D)$, with obvious identity and compositions, and $(f, g) \sim (f', g')$ if $f \sim f'$ and $g \sim g'$.    

\begin{defi}
A \textbf{symmetric monoidal $\sim$-category} is a $\sim$-category $\C$ equipped with an object $1 \in \C_0$, a $\sim$-functor 
\[
(- \tensor -) : \C \times \C \to \C,
\]
and, for all $A, B, C \in \C_0$, morphisms
\[
\begin{array}{rcrcl}
\alpha_{A, B, C} &:& (A\tensor B) \tensor C &\to& A \tensor (B\tensor
C)\\
\lambda_A &:& 1 \tensor A &\to& A\\
\rho_A &:& A \tensor 1 &\to& A\\
s_{A, B} &:& A \tensor B &\to& B \tensor A
\end{array}
\]
invertible up to $\sim$, and such that the usual naturality and coherence conditions for a symmetric monoidal category \cite{maclane} hold up to $\sim$.
\end{defi}
}

For $\Strat$, on arenas the tensor is an instance of the tensor
of games given in Definition \ref{def:games_tensor}. 
The unit $1$ is the empty arena where the empty configuration has payoff $0$.
\changed{This operation extends to a tensor product of strategies, defined using parallel composition of event structures. Rather than giving the concrete construction we state a characterization in terms of $+$-covered configurations: }
\begin{prop}\label{prop:def_tensor}
Consider $A, B, C, D$ arenas, and $\sigma : A \vdash B$, $\tau : C
\vdash D$ strategies.

Then, there is a strategy $\sigma \tensor \tau : A \tensor C
\vdash B \tensor D$, unique up to iso, s.t. there are 
\[
\begin{array}{rcrcl}
(- \tensor -) &:& \confp{\sigma} \times \confp{\tau} &\simeq&
\confp{\sigma \tensor \tau}\\
(- \tensor -) &:& \tildep{\sigma} \times \tildep{\tau} &\simeq&
\tildep{\sigma \tensor \tau}
\end{array}
\]
order-isos commuting with $\dom, \cod$, and s.t. for all $\theta^\sigma \in
\tildep{\sigma}$ and $\theta^\tau \in \tildep{\tau}$, 
\[
\pr_{\sigma \tensor \tau}(\theta^\sigma \tensor \theta^\tau) =
(\theta^\sigma_A \parallel \theta^\tau_C) \parallel (\theta^\sigma_B
\parallel \theta^\tau_D)\,.
\]
\end{prop}

This is clear from the construction of the tensor in \cite{cg2}
with uniqueness coming from Lemma \changed{4.11} of \cite{cg3}. 
Bifunctoriality up
to isomorphism is proved in \cite{cg2}. The components of
the symmetric monoidal structure are also given in \cite{cg2}; they are
immediate variants of the copycat strategy of Definition
\ref{def:copycat}. They also satisfy obvious characterizations of their
$+$-covered configurations, with for instance for the associator:
\[
((x_A \parallel x_B) \parallel x_C) \parallel (x_A \parallel (x_B
\parallel x_C)) \in \confp{\alpha_{A, B, C}}
\]
for $x_A \in \conf{A}, x_B \in \conf{B}$ and $x_C \in \conf{C}$ and
$\confp{\alpha_{A, B, C}}$ containing exactly those, and similarly for
$\lambda_A, \rho_A$ and $s_{A, B}$. For this paper, the only new things
to check are that the structural isomorphisms are exhaustive, and that
the tensor preserves exhaustivity.
The former is straightforward via characterization of
configurations; we detail the latter.

\begin{prop}
If $\sigma : A \vdash B$ and $\tau : C \vdash D$ are exhaustive, so
is $\sigma \tensor \tau : A \tensor C \vdash B \tensor D$.
\end{prop}
\begin{proof}
Consider $x^\sigma \tensor x^\tau \in \confp{\sigma \tensor \tau}$,
\emph{i.e.} $x^\sigma \in \confp{\sigma}$ and $x^\tau \in \confp{\tau}$.
If $\kappa_A(x^\sigma_A) = -1$ or $\kappa_C(x^\tau_C) = -1$, then we are
done; thus assume $\kappa_A(x^\sigma_A) \geq 0$ and $\kappa_C(x^\tau_C)
\geq 0$. If $\kappa_{B\tensor D}(x^\sigma_B \parallel x^\tau_D) = 1$
then we are done; thus assume $\kappa_{B \tensor D}(x^\sigma_B
\parallel x^\tau_D) \leq 0$. If $\kappa_{B\tensor D}(x^\sigma_B
\parallel x^\tau_D) = -1$, then $\kappa_B(x^\sigma_B) = -1$ or
$\kappa_D(x^\tau_D) = -1$, say the former. Since $\sigma$ is
exhaustive and $x^\sigma \in \confp{\sigma}$, we have
$\kappa_A(x^\sigma_A) = -1$, contradiction -- and likewise,
$\kappa_D(x^\tau_D) = -1$ yields a contradiction. So, $\kappa_{B\tensor
D}(x^\sigma_B \parallel x^\tau_D) = 0$, \emph{i.e.}
$\kappa_B(x^\sigma_B) = 0$ and $\kappa_D(x^\tau_D) = 0$. Finally, assume
$\kappa_A(x^\sigma_A) = 1$. Then, by exhaustivity of $\sigma$,
$\kappa_B(x^\sigma_B) = 1$ as well, contradiction. So
$\kappa_A(x^\sigma_A) = 0$, and likewise $\kappa_C(x^\tau_C) = 0$.
Summing up, the overall payoff of $\pr_{\sigma \tensor \tau}(x^\sigma
\tensor x^\tau)$ is $0$, as required.
\end{proof}

Overall, we have completed the symmetric monoidal 
 structure of
$\Strat$:

\begin{prop}
Equipped with the above, $\Strat$ is a symmetric monoidal $\sim$-category.
\end{prop}

\subsubsection{Closed structure.}
\label{subsubsec:semiclosed_strat}

\changed{The (relative) closed structure is easily derived, using the
linear function space construction of Section
\ref{subsubsec:arrow}.} For $A, B, C$ arenas with $C$ strict, the
currying bijection $\Lambda(-) : \Strat(A \tensor B, C) \to
\Strat(A, B \lin C)$
leaves the ess unchanged and only affects the display map: for $\sigma
: A \tensor B \to C$ and $x^\sigma \in \confn{\sigma}$ with
$\pr_\sigma(x^\sigma) = (x^\sigma_A \parallel x^\sigma_B) \parallel
x^\sigma_C$, 
\[
\pr_{\Lambda(\sigma)}(x^\sigma) = x^\sigma_A \parallel (x^\sigma_B \lin
x^\sigma_C)
\]
using the notation introduced in Section \ref{subsubsec:arrow}.

The evaluation morphism $\evm_{B, C} : (B \lin C) \tensor B \to C$ is a copycat-like strategy having 
\[
((x_B \lin x_C) \parallel x_B) \parallel x_C \in \confp{\evm_{B, C}}
\]
as $+$-covered non-empty configurations, for any $x_B \in \conf{B}$ and
$x_C \in \confn{C}$.

\subsubsection{Cartesian products.} \label{subsubsec:cart_id}
We now introduce the cartesian structure.

\changed{


\begin{defi}
\label{def:simproduct}
A $\sim$-category $\C$ \textbf{has binary products} if for any $A, B \in \C$, there exists an object $A \with B \in \C$, and \textbf{projections} $\pi_A \in \C(A\with
B, A)$ and $\pi_B \in \C(A \with B, B)$, such that for every $\Gamma \in \C$, $\sigma \in \C(\Gamma, A)$, and $\tau \in \C(\Gamma, B)$, there exists $\tuple{\sigma, \tau}$, unique up to $\sim$, such that 
\[
\sigma \sim \pi_A \circ \tuple{\sigma, \tau} \qquad \qquad \tau \sim \pi_B \circ \tuple{\sigma, \tau}.
\] 

An object $\top \in \C$ is \textbf{terminal} if for any $\Gamma \in
\C$ there is $\term_\Gamma \in \C(\Gamma, \top)$, unique
up to $\sim$.
\end{defi}

We prove that $\Strat_s$ has binary products and a terminal object (so it has all \emph{finite} products), and these are preserved by the inclusion functor $\Strat_s \hookrightarrow \Strat$. 

Clearly, the arena $\top$ with no events and
$\kappa_\top(\emptyset) = 1$ is strict, and terminal in $\Strat$. For any $A$ and
$B$ strict arenas, their product is given by the construction $A \with
B$ of Definition \ref{def:with}. We give a characterisation of the pairing construction $\tuple{-, =}$. }
\begin{prop}\label{prop:def_pairing}
For arenas $\Gamma, A, B$, with $A, B$ strict, and strategies $\sigma :
\Gamma \vdash A, \tau : \Gamma \vdash B$, there is a strategy $\tuple{\sigma, \tau} : \Gamma \vdash A \with
B$, unique up to iso, s.t. there are order-isos:
\[
\begin{array}{rcl}
\confp{\sigma} + \confp{\tau}  &\simeq& \confp{\tuple{\sigma, \tau}}\\
\tildep{\sigma} + \tildep{\tau} & \simeq& \tildep{\tuple{\sigma, \tau}}
\end{array}
\]
commuting with $\dom, \cod$, and such that for all $\theta^\sigma \in
\tildep{\sigma}$ and $\theta^\tau \in \tildep{\tau}$, we have
\[
\pr_{\tuple{\sigma, \tau}}(\inj_\sigma(\theta^\sigma)) =
\theta^\sigma_\Gamma \parallel (\theta^\sigma_A \parallel \emptyset)\,,
\qquad
\pr_{\tuple{\sigma, \tau}}(\inj_\tau(\theta^\tau)) =
\theta^\tau_\Gamma \parallel (\emptyset \parallel \theta^\tau_B)
\]
writing $\inj_\sigma : \confp{\sigma} \to \confp{\tuple{\sigma, \tau}}$
and $\inj_\tau : \confp{\tau} \to \confp{\tuple{\sigma, \tau}}$ for the
induced injections.
\end{prop}
\changed{
This result and the universal property of Definition~\ref{def:simproduct} follow easily from the construction of the cartesian product of arenas given in \cite{cg3}. }Projections are the obvious copycat strategies, with $+$-covered configurations 
\[
(x_A \parallel \emptyset) \parallel x_A \in \confp{\pi_A}\,,
\qquad
(\emptyset \parallel x_B) \parallel x_B \in \confp{\pi_B}
\]
for $x_A \in \conf{A}, x_B \in \conf{B}$. Their exhaustivity is easy;
note however that this uses the fact that $A$ and $B$
are strict. We detail the proof that pairing preserves
exhaustivity:
\begin{prop}
Consider $\Gamma, A, B$ arenas with $A, B$ strict, and $\sigma : \Gamma
\vdash A, \tau : \Gamma \vdash B$.

If $\sigma$ and $\tau$ are exhaustive, then so is $\tuple{\sigma, \tau}
: \Gamma \vdash A \with B$.
\end{prop}
\begin{proof}
Consider a $+$-covered configuration of $\tuple{\sigma, \tau}$, say for
instance $\inj_\sigma(x^\sigma) \in \confp{\tuple{\sigma, \tau}}$ for
$x^\sigma \in \confp{\sigma}$; with $\pr_{\tuple{\sigma, \tau}}
\inj_\sigma(x^\sigma) = x^\sigma_\Gamma \parallel (x^\sigma_A \parallel
\emptyset) \parallel x^\sigma_A$. If $\kappa_{A\with B}(x^\sigma_A
\parallel \emptyset) = 1$, then we are done; assume $\kappa_{A \with
B}(x^\sigma_A \parallel \emptyset) \leq 0$. Likewise, assume
$\kappa_\Gamma(x^\sigma_\Gamma) \geq 0$. If
$\kappa_\Gamma(x^\sigma_\Gamma) = 1$, then we must have
$\kappa_A(x^\sigma_A) = 1$ as well since $\sigma$ is exhaustive -- but
this implies $\kappa_{A\with B}(x^\sigma_A \parallel \emptyset) = 1$ as
well, contradiction. Likewise, if $\kappa_{A\with B}(x^\sigma_A
\parallel \emptyset) = -1$, then $\kappa_A(x^\sigma_A) = -1$ and so
$\kappa_\Gamma(x^\sigma_\Gamma) = -1$ since $\sigma$ is exhaustive,
contradiction. The only case left has $\kappa_\Gamma(x^\sigma_\Gamma) =
0$ and $\kappa_{A\with B}(x^\sigma_A \parallel \emptyset) = 0$, so
the display of $\inj_\tau(x^\sigma)$ has payoff $0$. The symmetric
reasoning applies to $+$-covered configurations of the form
$\inj_\tau(x^\tau)$ for $x^\tau \in \confp{\tau}$, which concludes the
proof.
\end{proof}
\changed{
In summary, we have the following:
\begin{prop}
The full sub-$\sim$-category $\Strat_s$ of strict arenas has finite products, preserved by the inclusion functor $\Strat_s \hookrightarrow \Strat$. 
\end{prop}
}

\subsubsection{Exponential.}\label{subsubsec:exp_strat}
We define a $\sim$-comonad $\oc : \Strat \to \Strat$. By a general result \cite{relativemonads} the restriction of $\oc$ to the sub-$\sim$-category $\Strat_s$ gives a relative comonad $\oc : \Strat_s \to \Strat$, and below we will show that this satisfies the axioms of Definition~\ref{def:relativeseely}. 



\changed{For an object $A \in \Strat$}, $\oc A$ is given by Definition \ref{def:bang}.
Following the same pattern as for the other structure we introduce the
functorial action of $\oc$ as follows:
\begin{prop}\label{prop:char_pcov_bang}
Consider $A, B$ arenas, and $\sigma : A \vdash B$ a strategy.

Then there is a strategy $\oc \sigma : \oc A \vdash \oc B$ where the ess
is obtained as in Definition \ref{def:bang} (without the clauses on
positive and negative symmetries); there is an order-iso
\[
\,[-] : \mathsf{Fam}\left(\confpn{\sigma}\right)  \bij
\confp{\oc \sigma}
\]
with $\mathsf{Fam}(X)$ the set of families of elements of $X$ indexed by
finite subsets of $\mathbb{N}$, and where $\confpn{\sigma}$ denotes the
set of non-empty $+$-covered configurations of $\sigma$. Moreover, 
\[
\pr_{\oc \sigma}\left(\left[(x^i)_{i\in I}\right]\right)
=
(\parallel_{i\in I} x^i_A) \parallel (\parallel_{i \in I} x^i_B)
\]
where for all $i\in I$, $\pr_\sigma(x^i) = x^i_A \parallel x^i_B$.
\end{prop}

The proof is direct from the definition of $\oc \sigma$ (see
\cite{cg3}). Though the symmetries of $\oc \sigma$ may also be described
in a similar style (as for earlier operations), we refrain from doing so
as it is slightly more heavy notationally: indeed, whereas for tensor
and pairing the symmetries are constructed exactly as for
configurations, in $\oc \sigma$ symmetries span components as they may
freely exchange copy indices as stated in Definition \ref{def:bang}.  

\changed{The $\sim$-comonad structure of $\oc$ is given by natural transformations $\der_A : \oc A \to A$ and $\dig_A : \oc A \to \oc \oc A$ whose components are, as usual, relabeled copycat strategies, characterized by the shape of their $+$-covered configurations:}
\[
\begin{array}{rclcl}
\left(\parallel_{\tuple{i,j} \in \mathbb{N}} x_A^{i,j}\right)
&\parallel& \left( \parallel_{i\in \mathbb{N}} \parallel_{j\in
\mathbb{N}} x_A^{i,j}\right) 
&\in& \confp{\dig_A}\\
\{0\} \times x_A &\parallel& x_A &\in&\confp{\der_A}
\end{array}
\]
whenever these sets are finite, $x_A, x^{i, j}_A \in \conf{A}$, and using a fixed bijection $\tuple{-, -} : \mathbb{N}^2 \bij \mathbb{N}$.

It is routine that
these are exhaustive -- for $\der_A$, the case
of the empty set uses condition \emph{initialized}. We must also prove
that the functorial operation preserves exhaustivity:

\begin{prop}
Consider $A, B$ arenas, and $\sigma : A \vdash B$ a strategy.

If $\sigma$ is exhaustive, then so is $\oc \sigma : \oc A \vdash \oc B$.
\end{prop}
\begin{proof}
Consider $[(x^i)_{i\in I}] \in \confp{\oc \sigma}$, where $I \subseteq_f
\mathbb{N}$ and $x^i \in \confpn{\sigma}$ for all $i \in I$. We have
\[
\pr_{\oc \sigma}\left(\left[(x^i)_{i\in I}\right]\right)
=
(\parallel_{i\in I} x^i_A) \parallel (\parallel_{i \in I} x^i_B)
\]
where $\pr_\sigma(x^i) = x^i_A \parallel x^i_B \in \conf{A \vdash B}$.
If $\kappa_{\oc B}(\parallel_{i\in I} x_B^i) = 1$ then we are done, so
assume $\kappa_{\oc B}(\parallel_{i\in I} x_B^i) \leq 0$. Symmetrically,
assume $\kappa_{\oc A}(\parallel_{i\in I} x_A^i) \geq 0$. If
$\kappa_{\oc B}(\parallel_{i\in I} x_B^i) = -1$, then there must be some
$i_0 \in \mathbb{N}$ such that $\kappa_B(x_B^{i_0}) = -1$. But because
$\sigma$ is exhaustive, this implies that $\kappa_A(x_A^{i_0}) = -1$ as
well; but then $\kappa_{\oc A}(\parallel_{i\in I} x_A^i) = -1$,
contradiction. So, $\kappa_{\oc B}(\parallel_{i\in I} x_B^i) = 0$, which
entails that for all $i \in I$, $\kappa_B(x_B^i) = 0$ (note that as $x^i
\in \confpn{\sigma}$, all $x_B^i$ are non-empty). Now, if $\kappa_{\oc
A}(\parallel_{i\in I} x_A^i) = 1$, then there must be some $i_0 \in I$
such that $\kappa_A(x_A^{i_0}) = 1$. But as $\sigma$ is exhaustive, this
entails that $\kappa_B(x_B^{i_0}) = 1$ as well, contradiction.
\end{proof}

\changed{The functor $\oc$ preserves $\simstrat$ and satisfies all necessary axioms up to $\simstrat$ \cite{cg3}, so that:
\begin{lem}
The functor $\oc : \Strat \to \Strat$ is a $\sim$-comonad, which restricts to a $J$-relative comonad $\Strat_s \to \Strat$, for $J$ the inclusion functor.     
\end{lem}

Finally we have, for strict $B, C$, a strategy $m_{B, C} : \oc B \tensor \oc C \to \oc(B \with C)$ characterized by 
\[
\begin{array}{rclcl}
\left( \left(\parallel_{i\in \mathbb{N}} x_B^i\right) \parallel
\left(\parallel_{j\in \mathbb{N}} x^j_C\right)\right)
&\parallel&
\left(
\begin{array}{r}
\biguplus_{j\in \mathbb{N}} \{2j+1\} \times (\emptyset \parallel
x_C^j)\\
\uplus~\biguplus_{i\in \mathbb{N}} \{2i\}\times (x_B^i \parallel
\emptyset)
\end{array}\right)
&\in& \confp{m_{B,C}}
\end{array}
\]
for 
$x_B^i \in \conf{B}, x_C^j \in \conf{C}$. Likewise, the strategy $m_0 : \oc \top \to 1$ with only the
empty configuration. The axioms can be verified, and we finally obtain:
\begin{prop}
The $\sim$-category $\Strat$, equipped with the full sub-$\sim$-category $\Strat_s$ of strict arenas, and all components outlined above, is a relative Seely $\sim$-category.  
\end{prop}
}

In particular, the Kleisli $\sim$-category $\Strat_\oc$, with objects strict arenas, is cartesian closed. 


\subsection{Interpretation of $\nPCF$} Combinators of the simply-typed
$\lambda$-calculus are handled in the standard way following the
cartesian closed structure of $\Strat_\oc$ \cite{lambekscott}, the only
structure left to finalize the interpretation is the interpretation of
$\nPCF$ primitives, and recursion. 

\subsubsection{Base types and primitives} The arenas for the base types
are given in Section \ref{subsubsec:games_construction}.

For the $\PCF$ primitives, 
\begin{figure}
\[
\xymatrix@R=-6pt@C=-1pt{
\mathsf{if} &\!\!\!\!\!\!:\!\!\!\!\!\!& \oc (\gbool & \with& \gx & \with
&\gx)
&\vdash& \gx\\
&&&&&&&&\qu^-\\
&&\qu^+_{\grey{0}}\\
&&\ttrue^-_{\grey{0}}\\
&&&&\qu^+_{\grey{1}}\\
&&&&v^-_{\grey{1}}\\
&&&&&&&&v^+
}
\qquad
\xymatrix@R=-6pt@C=-1pt{
\mathsf{if} &\!\!\!\!\!\!:\!\!\!\!\!\!& \oc (\gbool &\with & \gx & \with
&\gx)
&\vdash& \gx\\
&&&&&&&&\qu^-\\
&&\qu^+_{\grey{0}}\\
&&\tfalse^-_{\grey{0}}\\
&&&&&&\qu^+_{\grey{1}}\\
&&&&&&v^-_{\grey{1}}\\
&&&&&&&&v^+
}
\qquad
\xymatrix@R=-6pt@C=0pt{
\mathsf{succ} &\!\!\!\!\!\!:\!\!\!\!\!\!&
\oc \gnat &\vdash& \gnat\\
&&&&\qu^-\\
&&\qu^+_{\grey{0}}\\
&&n^-_{\grey{0}}\\
&&&&(n+1)^+
}
\]
\medskip
\[
\xymatrix@R=-6pt@C=-3pt{
\mathsf{iszero} &\!\!\!\!\!\!:\!\!\!\!\!\!&
\oc \gnat &\vdash & \gbool\\
&&&&\qu^-\\
&&\qu^+_{\grey{0}}\\
&&0_{\grey{0}}^-\\
&&&&\ttrue^+
}
\quad
\xymatrix@R=-6pt@C=0pt{
\mathsf{iszero} &\!\!\!\!\!\!:\!\!\!\!\!\!&
\oc \gnat &\vdash &\gbool\\
&&&&\qu^-\\
&&\qu^+_{\grey{0}}\\
&&(n+1)_{\grey{0}}^-\\
&&&&\tfalse^+
}
\quad
\xymatrix@R=-6pt@C=-3pt{
\mathsf{pred} &\!\!\!\!\!\!:\!\!\!\!\!\!&
\oc \gnat &\vdash& \gnat\\
&&&&\qu^-\\
&&\qu^+_{\grey{0}}\\
&&0^-_{\grey{0}}\\
&&&&0^+
}
\quad
\xymatrix@R=-6pt@C=-3pt{
\mathsf{pred} &\!\!\!\!\!\!:\!\!\!\!\!\!&
\oc \gnat &\vdash& \gnat\\
&&&&\qu^-\\
&&\qu^+_{\grey{0}}\\
&&(n+1)^-_{\grey{0}}\\
&&&&n^+
}
\]
\caption{Strategies for basic $\PCF$ combinators}
\label{fig:intr_pcf}
\end{figure}
\begin{figure}
\begin{minipage}{.55\linewidth}
\begin{eqnarray*}
\intr{\ite{M}{N_1}{N_2}} &=& \mathsf{if}\odot_\oc \tuple{\intr{M},
\intr{N_1}, \intr{N_2}}\\
\intr{\tsucc\,M} &=& \mathsf{succ} \odot_\oc \intr{M}\\
\intr{\tpred\,M} &=& \mathsf{pred} \odot_\oc \intr{M}\\
\intr{\iszero\,M} &=& \mathsf{iszero} \odot_\oc \intr{M}
\end{eqnarray*}
\caption{Interpretation of basic combinators}
\label{fig:intr_basic_pcf}
\end{minipage}
\hfill
\begin{minipage}{.44\linewidth}
\[
\xymatrix@R=12pt@C=5pt{
&\gbool\\
&\qu^-
	\ar@{-|>}[dl]
	\ar@{-|>}[dr]\\
\ttrue^+\ar@{.}@/^/[ur]
	\ar@{~}[rr]&&
\tfalse^+
	\ar@{.}@/_/[ul]
}
\]
\caption{Interpretation of $\coin$}
\label{fig:intr_coin}
\end{minipage}
\end{figure}
the constants $\Gamma \vdash \ttrue, \tfalse : \tbool$ and $\Gamma
\vdash n : \tnat$ are interpreted by the obvious strategies which reply
accordingly to the initial move. Conditionals along with the primitives
$\iszero, \tpred$, and $\tsucc$ are interpreted following the clauses of
Figure \ref{fig:intr_basic_pcf}, using the strategies shown in Figure
\ref{fig:intr_pcf} -- note that we omit arrows to avoid clutter.

\subsubsection{Recursion} The fixpoint combinator is defined as the
least upper bound of its finite approximants, leveraging completeness
properties of the following partial order:

\begin{defi}
Consider $A$ a game, and $\sigma, \tau : A$ strategies.

We write $\sigma \cleq \tau$ if $\conf{\sigma} \subseteq \conf{\tau}$ --
so in particular $\ev{\sigma} \subseteq \ev{\tau}$, and additionally:
\[
\begin{array}{rl}
\text{\emph{(1)}} & 
\text{for all $s_1, s_2 \in \ev{\sigma}$, $s_1 \leq_\sigma s_2$ iff
$s_1 \leq_\tau s_2$,}\\
\text{\emph{(2)}} & 
\text{for all $s_1, s_2 \in \ev{\sigma}$, $s_1 \conflict_\sigma s_2$
iff $s_1 \conflict_\tau s_2$,}\\
\text{\emph{(3)}} &
\text{for all $x, y \in \conf{\sigma}$ and bijection $\theta : x \simeq
y$, we have $\theta \in \tilde{\sigma}$ iff $\theta \in
\tilde{\tau}$,}\\
\text{\emph{(4)}} &
\text{for all $s \in \ev{\sigma}$, $\pr_\sigma(s) = \pr_\tau(s)$,}
\end{array}
\]
\emph{i.e.} all components compatible with the inclusion.
\end{defi}

Strategies on $A$, ordered by $\cleq$, form a \emph{directed complete
partial order}, with respect to which all operations on strategies are
easily shown to be continuous. We have:

\begin{prop}\label{prop:sup_pcov}
Consider $A$ a game, and $D$ a directed set of strategies on $A$. Then:
\[
\confp{\vee D} = \bigcup_{\sigma \in \D} \confp{\sigma}\,,
\qquad
\qquad
\tildep{\vee D} = \bigcup_{\sigma \in \D} \tildep{\sigma}\,.
\]

Moreover, if every $\sigma \in D$ is exhaustive, then so is $\vee D : A$.
\end{prop}

The supremum also preserves visibility; so for any
$A, B$ the homset $\Strat(A, B)$ is a dcpo.
Before defining the recursion operator via the usual fixpoint formula,
we must deal with the minor inconvenience that the dcpo of
strategies on $A$ does not have a least element: strategies minimal for
$\cleq$ still have -- by receptivity -- events matching the
minimal events of $A$, but they are free to \emph{name} those
arbitrarily. We solve this as in \cite{cg2}: we choose one minimal
$\bot_A : A$. For any $\sigma : A$, we pick an isomorphic
$\sigma \iso \sigma^\flat : A$ s.t. $\bot_A \cleq \sigma^\flat$,
obtained by renaming minimal events. We write $\D_A$ for the pointed
dcpo of strategies above $\bot_A$. 

As all operations on strategies examined so far are continuous, for any
arena $A$ we have
\[
\begin{array}{rcrcl}
F &:& \D_{\oc \top \vdash (A \to A) \to A} &\to& \D_{\oc \top \vdash (A \to A) \to A}\\
&& \sigma &\mapsto& (\lambda f^{A \to A}.\,f\,(\sigma\,f))^\flat\,,
\end{array}
\]
written in $\lambda$-calculus
syntax relying on the constructions on strategies corresponding to the
cartesian closed structure of $\Strat_\oc$, continuous. By Kleene's
fixpoint theorem
\[
\Y_A = \bigvee_{n \in \mathbb{N}} F^n(\bot) \in \Strat_\oc(\top, (A \to
A) \to A)
\]
is a least fixpoint of $F$. Finally, in the presence of a context
$\Gamma$, we set $\Y_{\Gamma, A} = \Y_A \odot_\oc e_{\Gamma}$ where
$e_{\Gamma} \in \Strat_\oc(\Gamma, \top)$ is the terminal morphism,
yielding $\Y_{\Gamma, A} \in \Strat_\oc(\Gamma, \oc(\oc A \lin A) \lin
A)$.

%

This concludes the interpretation of $\nPCF$ in $\Strat$. From
\cite{cg3}, it is adequate with respect to may-convergence -- however
this statement will play no role here. Most of the rest of the
paper will study the collapse of $\Strat$ (and its
interpretation of $\nPCF$) onto $\Rel{\N}$. 

\subsection{Collapsing $\Strat$ to
$\Rel{\N}$}\label{subsec:intro_collapse}
Back to the question asked in Section \ref{subsubsec:question}: what
does the $\N$-weighted relational model count, on arbitrary higher-order
types? 

We shall answer this by providing a \emph{collapse}
interpretation-preserving functor, written 
\[
\coll(-) : \Strat \to \Rel{\N}\,.
\]

At first this seems simple: recall from Lemma \ref{lem:web_wconf} that for
any type $A$, there is 
$s^{\Ty}_A : \rintr{A} \bij \wconf{\intr{A}}$
a bijection through which we regard $\rintr{A}$ -- the
\textbf{web} of $A$, as it is usually called -- as the set of symmetry
classes of null payoff of $\intr{A}$. Given $\sigma
: A$ and $\x_A \in \wconf{\intr{A}}$, we must associate a
\emph{weight}
$\left(\coll(\sigma)\right)_{\x_A} \in \mathbb{N} \cup \{+\infty\}$.
It seems natural for this weight to be
\[
(\coll(\sigma))_{\x_A} = \sharp \left(\wit_\sigma(\x_A)\right)
\]
the cardinality of an adequately chosen set of \emph{witnesses} of
$\x_A$. On ground types this seems rather clear: for instance,
considering $\vdash M : \tbool$ defined as $M =
\ite{\coin}{\ttrue}{\ttrue}$, we have
\[
\intr{M} = 
\raisebox{10pt}{$
\xymatrix@C=5pt@R=10pt{
&\qu^-	\ar@{-|>}[dl]
	\ar@{-|>}[dr]\\
\ttrue^+\ar@{.}@/^/[ur]
	\ar@{~}[rr]&&
\ttrue^+\ar@{.}@/_/[ul]
}$}
\qquad
\qquad
\qquad
\begin{array}{rcl}
\rintr{M}_\ttrue &=& 2\\
\rintr{M}_\tfalse &=& 0
\end{array}
\]
so one may guess that $\wit(\x_A) = \{x^\sigma \in \conf{\sigma} \mid
\pr_\sigma(x^\sigma) \in \x_A\}$ -- however, this is inadequate beyond
ground types: if $x^\sigma \in \conf{\sigma}$
includes any copyable Opponent move, then $x^\sigma$
automatically has countably many symmetric copies prompted by copies of
that Opponent move, displayed to the same configuration of the game up to
symmetry. So this notion of witnesses is useless save in the
most simple cases, as it makes no account of symmetry.

So what is the right set of witnesses up to symmetry? \emph{How to count
configurations of a strategy up to symmetry}? This question, the crux of
the paper, is investigated next.

\section{Witnesses and Composition}
\label{sec:wit_comp}

To set up our collapse, the most challenging proof obligation -- by far
-- is preservation of composition. Here we present the right definition of
collapse, along with the proof that it preserves composition. But rather
than arriving there directly, we take a more indirect route, showing
some of the subtleties constraining the solution.

\subsection{Finding the Right Witnesses}

A fitting starting point for this discussion is to recall 
\begin{eqnarray}
(\beta \circ \alpha)_{x, z} &=& \sum_{y \in Y} \alpha_{x, y} \cdot
\beta_{y, z}\label{eq1}
\end{eqnarray}
the composition of $\alpha \in \Rel{\N}(X, Y)$ and $\beta \in
\Rel{\N}(Y, Z)$, which we must relate to the composition of strategies.
As explained above, for $\sigma \in \Strat(A, B)$ and $\x_A \in
\wconf{A}, \x_B \in \wconf{B}$, we expect $(\coll(\sigma))_{\x_A, \x_B}$
to be the cardinality of a well-chosen set $\wit_\sigma(\x_A, \x_B)$ of
\emph{witnesses} for $\x_A \parallel \x_B \in \wconf{A \vdash B}$. 
Thus, \eqref{eq1} strongly hints at a bijection
\begin{eqnarray}
\wit_{\tau \odot \sigma}(\x_A, \x_C) &\bij& \sum_{\x_B \in \wconf{B}}
\wit_\sigma(\x_A, \x_B) \times \wit_\tau(\x_B, \x_C)\label{eq2}
\end{eqnarray}
for all $\sigma \in \Strat(A, B)$, $\tau \in \Strat(B, C)$, $\x_A \in
\wconf{A}$ and $\x_C \in \wconf{C}$.

\subsubsection{Witnesses without symmetry} \label{subsubsec:wit_w_sym}
Let us start by ignoring symmetry,
%
%
%
and first count witnesses for plain complete configurations $x_A \in
\nconf{A}$. It seems natural to define: 
\[
\wit_\sigma(x_A, x_B) = \{x^\sigma \in \confp{\sigma} \mid
\pr_\sigma(x^\sigma) = x_A \parallel x_B\}\,.
\]
for all $\sigma \in
\Strat(A,B)$, $x_A \in \nconf{A}$ and $x_B \in \nconf{B}$.

Working with $+$-covered configurations seems reasonable, as
Proposition \ref{prop:char_comp} entails:
\[
(- \odot -) : 
\{(x^\tau, x^\sigma) \in \confp{\tau} \times \confp{\sigma} \mid
\text{$x^\sigma$ and $x^\tau$ causally compatible}\} 
\simeq
\confp{\tau \odot \sigma}
\]
preserving the display maps for all $\sigma : A\vdash B$ and $\tau : B
\vdash C$, which seems close to \eqref{eq2}. The two obstacles to
\eqref{eq2} are that: \emph{(1)} this bijection has an additional
requirement that synchronized $x^\sigma \in \confp{\sigma}$ and $x^\tau
\in \confp{\tau}$ should be \emph{causally compatible} -- a constraint
which does not appear in \eqref{eq2}; and \emph{(2)} we must be sure
that for all $x^\sigma \in \confp{\sigma}$ and $x^\tau \in \confp{\tau}$
synchronizable, if $x_A \in \nconf{A}$ and $\x_C \in
\nconf{C}$, then $x^\sigma_B = x^\tau_B \in \nconf{B}$.

For \emph{(1)}, we invoke the ``\emph{deadlock-free lemma}'' from \cite{cg3}:

\begin{lem}\label{lem:deadlock_free}
Consider $A, B, C$ arenas, $\sigma : A \vdash B$ and $\tau : B
\vdash C$ visible strategies, $x^\sigma \in \conf{\sigma}$ and
$x^\tau \in \conf{\tau}$ with a symmetry $\theta : x^\sigma_B \sym_B
x^\tau_B$.
Then, the composite bijection
\[
\varphi
\quad
:
\quad
x^\sigma \parallel x^\tau_C 
\quad
\stackrel{\pr_\sigma \parallel x^\tau_C}{\simeq} 
\quad
x^\sigma_A \parallel x^\sigma_B \parallel x^\tau_C 
\quad
\stackrel{x^\sigma_A \parallel \theta \parallel x^\tau_C}{\simeq}
\quad
x^\sigma_A \parallel x^\tau_B \parallel x^\tau_C
\quad
\stackrel{x^\sigma_A \parallel \pr_\tau^{-1}}{\simeq}
\quad
x^\sigma_A \parallel x^\tau\,,
\]
is secured, in the sense that the relation $\vartriangleleft$, defined
on the graph of $\varphi$ with
\[
(l, r) \vartriangleleft (l', r')
\]
whenever $l \mathrel{(<_\sigma \parallel <_C)} l'$ or $r \mathrel{(<_A
\parallel <_\tau)} r'$, is acyclic\footnote{For $\theta$ an identity,
this exactly means that $x^\sigma$ and $x^\tau$ satisfy the
\emph{secured} condition of Definition \ref{def:caus_comp}.}. 
\end{lem}

The quite subtle proof is out of the scope of the paper, the interested
reader is referred to \cite{cg3}. It is one of the main properties of 
\emph{visibility} that composition
never deadlocks, so that when dealing with visible strategies, the
\emph{causally compatible} requirement of Proposition
\ref{prop:char_comp} is redundant -- the lemma above also covers the
case of synchronization through a symmetry $\theta$, which will be
necessary later on in the paper.

Next, for \emph{(2)}, we prove the following property:

\begin{lem}\label{lem:sync_zero}
Consider games $A, B$ and $C$; and exhaustive $\sigma : A
\vdash B$ and $\tau : B \vdash C$.

For all $x^\sigma \in \confp{\sigma}$ and $x^\tau \in \confp{\tau}$, if
$\kappa_A(x^\sigma_A) = 0$ and $\kappa_C(x^\tau_C) = 0$, then $\kappa_B(x_B)
= 0$.
\end{lem}
\begin{proof}
Seeking a contradiction, assume $\kappa_B(x_B) = 1$. Then
$\kappa_{B\vdash C}(x_B \parallel x^\tau_C) = -1$, contradicting that
$\tau$ is exhaustive. Symmetrically, $\kappa_B(x_B) = -1$ contradicts
that $\sigma$ is exhaustive. 
\end{proof}

From these two statements, we may now deduce as claimed:

\begin{cor}
Consider $A, B, C$ arenas, $\sigma \in \Strat(A, B), \tau \in \Strat(B,
C)$. Then, we have
\begin{eqnarray}
\wit_{\tau \odot \sigma}(x_A, x_C) &\bij& \sum_{x_B \in \nconf{B}}
\wit_\sigma(x_A, x_B) \times \wit_\tau(x_B, x_C)\label{eq3}
\end{eqnarray}
for all $x_A \in \nconf{A}$ and $x_C \in \nconf{B}$.
\end{cor}
\begin{proof}
We construct the bijection by building functions in both directions.

Take $y \in \wit_{\tau \odot \sigma}(x_A, x_C)$, \emph{i.e.} $y \in
\confp{\tau \odot \sigma}$ s.t. $\pr_{\tau \odot \sigma}(y) = x_A
\parallel x_C$. By Proposition \ref{prop:char_comp}, $y = x^\tau \odot
x^\sigma$ s.t. $x^\sigma \in \confp{\sigma}$ and $x^\tau \in
\confp{\tau}$ are causally compatible, and with 
\[
\pr_\sigma(x^\sigma) = x_A \parallel x_B\,
\qquad
\qquad
\pr_\tau(x^\tau) = x_B \parallel x_C
\]
for some $x_B \in \conf{B}$. But by Lemma \ref{lem:sync_zero} we have
$\kappa_B(x_B) = 0$, so we return $(x_B, x^\sigma, x^\tau)$.

Now, take $x_B \in \nconf{B}$, $x^\sigma \in \wit_\sigma(x_A, x_B)$ and
$x^\tau \in \wit_\tau(x_B, x_C)$. By Lemma \ref{lem:deadlock_free},
$x^\sigma$ and $x^\tau$ are causally compatible, so 
$x^\tau \odot x^\sigma \in \confp{\tau \odot \sigma}$ s.t.
$\pr_{\tau \odot \sigma}(x^\tau \odot x^\sigma) = x_A \parallel x_C$.

That these constructions are inverses of one another follows from
Proposition \ref{prop:char_comp}.
\end{proof}

This is a good starting point, which we must now extend to deal
with symmetry.

\subsubsection{Witnesses as symmetry classes} How shall we extend
\eqref{eq3} to account for symmetry?

Since points of the web correspond to symmetry classes of configurations
of the game, it would seem sensible for witnesses to be symmetry classes
of configurations of the strategy:
\[
\swit_\sigma(\x_A, \x_B) = 
\{\x^\sigma \in \spconf{\sigma} \mid \pr_\sigma(\x^\sigma) = \x^\sigma_A
\parallel \x^\sigma_B\}
\]
for $\sigma : A \vdash B$, $\x_A \in \wconf{A}$ and $\x_B \in
\wconf{B}$; where $\spconf{\sigma}$ is the set of symmetry classes
of $+$-covered configurations of $\sigma$; and where
$\pr_\sigma(\x^\sigma) \in \tilde{A\vdash B}$ is well-defined since
$\pr_\sigma$ preserves symmetry. With this definition, are we able to
extend the bijection \eqref{eq3}?

Consider first $\y \in \swit_{\tau \odot \sigma}(\x_A, \x_C)$. Take any
representative $x^\tau \odot x^\sigma \in \y$. The configurations
$x^\sigma$ and $x^\tau$ synchronize in $x^\sigma_B = x^\tau_B = x_B$,
with symmetry class $\x_B \in \wconf{B}$. So taking $\x^\sigma$ and
$\x^\tau$ the respective symmetry classes of $x^\sigma$ and $x^\tau$, we obtain
\[
\x_B \in \wconf{B}\,,
\qquad
\x^\sigma \in \swit_\sigma(\x_A, \x_B)\,,
\qquad
\x^\tau \in \swit_\tau(\x_B, \x_C)
\]
as required. From the characterization of the symmetries of composition
in Proposition \ref{prop:char_comp}, one deduces
that this does not depend on the choice of the representative $x^\tau
\odot x^\sigma \in \y$.

What about the other direction? Consider now $\x^\sigma \in
\swit_\sigma(\x_A, \x_B)$ and $\x^\tau \in \swit_\tau(\x_B, \x_C)$ for
some $\x_B \in \wconf{B}$, and pick representatives $x^\sigma \in
\x^\sigma$ and $x^\tau \in \x^\tau$. Displaying them as
\[
\pr_\sigma(x^\sigma) = x^\sigma_A \parallel x^\sigma_B
\qquad
\qquad
\pr_\tau(x^\tau) = x^\tau_B \parallel x^\tau_C\,,
\]
we must compute their synchronization. But we may not have $x^\sigma_B =
x^\tau_B$, indeed we only know that $x^\sigma_B, x^\tau_B \in \x_B$ so
that there must be some $\theta : x^\sigma_B \sym_B x^\tau_B$.
Fortunately, thin concurrent games come with tools to compute such
synchronizations up to symmetry -- in particular, at the heart of the
proof of Proposition \ref{prop:congruence} is the following property:

\begin{restatable}{prop}{syncsym}\label{prop:sync_sym}
Consider $\sigma : A \vdash B$ and $\tau : B \vdash C$ two strategies.

For $x^\sigma \in \confp{\sigma}, x^\tau \in \confp{\tau}$ and
$\theta : x^\sigma_B \sym_B x^\tau_B$ s.t. the composite
bijection is \emph{secured}:
\[
x^\sigma \parallel x^\tau_C 
\quad
\stackrel{\pr_{\sigma}\parallel C}{\simeq}
\quad
x^\sigma_A \parallel x^\sigma_B \parallel x^\tau_C 
\quad
\stackrel{A \parallel \theta \parallel C}{\sym}
\quad
x^\sigma_A \parallel x^\tau_B \parallel x^\tau_C
\quad
\stackrel{A \parallel \pr_{\tau}^{-1}}{\simeq}
\quad
x^\sigma_A \parallel x^\tau\,,
\]
then there are (necessarily unique) $y^\sigma \in \confp{\sigma}$ and
$y^\tau \in \confp{\tau}$ causally compatible, and
\[
\varphi^\sigma : x^\sigma \sym_\sigma y^\sigma\,,
\qquad
\qquad
\varphi^\tau : x^\tau \sym_\tau y^\tau\,,
\]
such that
we have $\varphi^\sigma_A \in \ntilde{A}$, $\varphi^\tau_C \in
\ptilde{C}$, and $\varphi^\tau_B \circ \theta = \varphi^\sigma_B$. 
\end{restatable}

See Appendix \ref{app:syncsym}.
Intuitively, we play $\tilde{\sigma}$ and $\tilde{\tau}$
against each other. By \emph{$\sim$-receptivity} and \emph{extension}
they adjust their copy indices interactively until they reach an
agreement, \emph{i.e.} pairs of configurations matching on the nose.

In the situation at hand the securedness assumption is automatic by
Lemma \ref{lem:deadlock_free}, so that we get indeed $x^\sigma
\sym_\sigma y^\sigma$ and $x^\tau \sym_\tau y^\tau$ matching on the
nose, synchronizing to
\[
y^\tau \in y^\sigma \in \confp{\tau \odot \sigma}
\]
whose symmetry class yields $\y \in \swit_{\tau \odot \sigma}(\x_A,
\x_C)$ as needed.

Having given constructions in both directions, it might seem that we are
essentially done, with a version of \eqref{eq2} essentially following
the case without symmetry. But remember that above we started from
witnesses $\x^\sigma \in \swit_\sigma(\x_A, \x_B)$ and $\x^\tau \in
\swit_\tau(\x_B, \x_C)$, for which we first chose representatives
$x^\sigma \in \confp{\sigma}, x^\tau \in \confp{\tau}$ and a symmetry
$\theta : x^\sigma_B \sym_B x^\tau_B$. So one should not overlook the
proof obligation that the construction is invariant under the choice of
these representatives.  Unfortunately, the symmetry class $\y \in
\swit_{\tau \odot \sigma}(\x_A, \x_C)$ \emph{does} depend on the choice
of $x^\sigma, x^\tau$, and $\theta$  -- indeed, this was the error made
in \cite{DBLP:conf/lics/CastellanCPW18}. 

In fact, \eqref{eq2} does \emph{not} hold for this notion of witness.

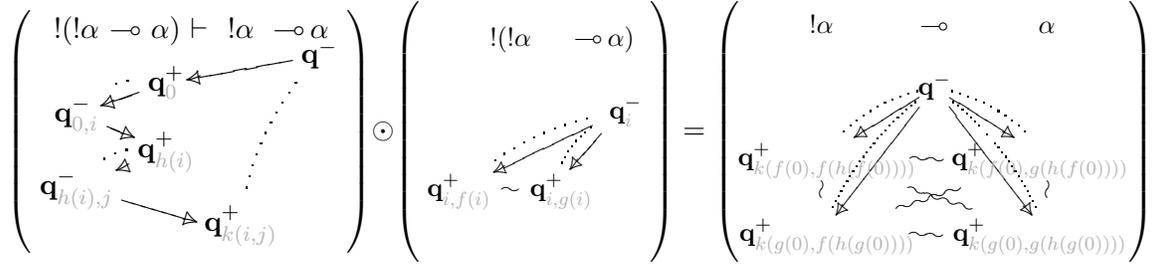
\begin{figure}
\[
\left(
\raisebox{38pt}{$
\xymatrix@R=-8pt@C=-8pt{
\oc (\oc \alpha &\lin& \alpha)& \vdash &\oc \alpha &\lin &\alpha\\
&&&&&&\qu^-
	\ar@{-|>}[dllll]\\
&&\qu^+_{\grey{0}}
	\ar@{-|>}[dll]\\
\qu^-_{\grey{0,i}}
	\ar@{.}@/^/[urr]
	\ar@{-|>}[drr]\\
&&\qu^+_{\grey{h(i)}}
	\ar@{-|>}[dll]\\
\qu^-_{\grey{h(i),j}}
	\ar@{.}@/^/[urr]
	\ar@{-|>}[drrr]\\
&&&&\qu^+_{\grey{k(i,j)}}
	\ar@{.}@/^/[uuuuurr]
}$}
\right)
\odot 
\left(
\raisebox{35pt}{$
\scalebox{.9}{$
\xymatrix@R=12pt@C=-7pt{
&\oc(\oc \alpha &&\lin &&\alpha)\\
&&&&&\,\,\qu^-_{\grey{i}}
	\ar@{-|>}[dlllll]
	\ar@{-|>}[dlll]\\
\qu^+_{\grey{i,f(i)}}
	\ar@{.}@/^/[urrrrr]
	\ar@{~}[rr]&&
\qu^+_{\grey{i,g(i)}}\!\!\!\!
	\ar@{.}@/^/[urrr]\\~
}$}$}
\right)
=
\left(
\raisebox{40pt}{$
\scalebox{.9}{$
\xymatrix@R=10pt@C=0pt{
\oc \alpha &&\!\!\lin\!\! &&\alpha\\
&&\!\!\qu^-\!\!
	\ar@{-|>}[dll]
	\ar@{-|>}[drr]
	\ar@{-|>}[ddll]
	\ar@{-|>}[ddrr]\\
\!\!\qu^+_{\grey{k(f(0), f(h(f(0))))}}\!\!\!\!\!
	\ar@{~}[rrr]
	\ar@{~}[drrrr]
	\ar@{~}[d]
	\ar@{.}@/^/[urr]&&&&
\!\!\!\!\!\qu^+_{\grey{k(f(0), g(h(f(0))))}}\!\!
	\ar@{~}[d]
	\ar@{.}@/_/[ull]\\
\!\!\qu^+_{\grey{k(g(0), f(h(g(0))))}}\!\!\!\!\!
	\ar@{~}[rrr]
	\ar@{~}[urrrr]
	\ar@{.}@/^/[uurr]&&&&
\!\!\!\!\!\qu^+_{\grey{k(g(0), g(h(g(0))))}}\!\!
	\ar@{.}@/_/[uull]
}$}$}
\right)
\]
\caption{Composition of $\sigma$ and $\tau$}
\label{fig:ex_composition}
\end{figure}
\begin{exa} Consider the arena $\alpha$ with one
event $\qu^-$, $\kappa_\alpha(\emptyset) = 1$ and
$\kappa_\alpha(\{\qu^-\}) = 0$, and
\[
\sigma : \oc (\oc \alpha \lin \alpha) \vdash \oc \alpha \lin \alpha\,,
\qquad
\qquad
\tau : \oc (\oc \alpha \lin \alpha)
\]
two strategies as pictured in Figure \ref{fig:ex_composition}. Their
assignment of copy indices uses functions
\[
f : \mathbb{N} \to \mathbb{N}\,,
\qquad
g : \mathbb{N} \to \mathbb{N}\,,
\qquad
h : \mathbb{N} \to \mathbb{N}\,,
\qquad
k : \mathbb{N}^2 \to \mathbb{N}\,,
\]
whose precise identity has no impact on the discussion. We are
interested in their composition, which unfolds as in Figure
\ref{fig:ex_composition}, yielding four pairwise conflicting,
non-symmetric positive moves. 
Enriching $\nPCF$ with a type $\alpha$ with no
constant, the substitution
\[
\left(\lambda
x^\alpha.\,f\,(f\,x)\right)
\,[(\lambda y^\alpha.\,\ite{\coin}{y}{y})/f]
\]
gives a perfect syntactic counterpart to the composition in Figure
\ref{fig:ex_composition}. Because there are two calls to the
non-deterministic choice, this reduces to $\lambda x^\alpha.\,x$, but in
four different ways. Observe that in the copy indices of each positive
move of $\tau \odot \sigma$, one can read back the way the two
non-deterministic choices were resolved: the upper row corresponds to
the first call yielding $\qu^+_{\grey{i,f(i)}}$, the leftmost column to the
second call yielding $\qu^+_{\grey{i,f(i)}}$, and so on.

From the composition, $\tau \odot \sigma$ has four
witnesses for
$\left(\raisebox{10pt}{$
\scalebox{.8}{$
\xymatrix@R=5pt@C=0pt{
\qu^-	\ar@{.}[d]\\
\qu^+
}$}$}\right) \in \wconf{\oc \alpha \lin \alpha}$, while
\[
\sharp \swit_\sigma\left(\raisebox{10pt}{$
\scalebox{.8}{$
\xymatrix@R=5pt@C=0pt{
\qu^-   \ar@{.}[d]&\qu^-\ar@{.}[d]\\
\qu^+&\qu^+
}$}$}, 
\raisebox{10pt}{$
\scalebox{.8}{$
\xymatrix@R=5pt@C=0pt{
\qu^-   \ar@{.}[d]\\
\qu^+
}$}$}
\right) = 1\,,
\qquad
\qquad
\sharp\swit_\tau\left(\raisebox{10pt}{$
\scalebox{.8}{$
\xymatrix@R=5pt@C=0pt{
\qu^-   \ar@{.}[d]&\qu^-\ar@{.}[d]\\
\qu^+&\qu^+
}$}$}\right) = 3\,,
\]
as $\sigma$ and $\tau$ cannot synchronize on any other symmetry class,
this contradicts \eqref{eq2}.

Indeed, up to symmetry, there are exactly \emph{three} configurations of
$\tau$ corresponding to two calls: \emph{(1)} both choices may be
resolved with $\qu_{\grey{i,f(i)}}$; \emph{(2)} both choices may be
resolved with $\qu_{\grey{i,g(i)}}$; and \emph{(3)} we may have one of
each. The point is that there is a symmetry
\begin{eqnarray}
\left(\raisebox{10pt}{$
\scalebox{.8}{$
\xymatrix@R=5pt@C=0pt{
\qu^-_{\grey{i}}   \ar@{.}[d]&\qu^-_{\grey{j}}\ar@{.}[d]\\
\qu^+_{\grey{i,f(i)}}&\qu^+_{\grey{j, g(j)}}
}$}$}\right)
&\sym_{\oc (\oc \alpha \lin \alpha)}&
\left(\raisebox{10pt}{$
\scalebox{.8}{$
\xymatrix@R=5pt@C=0pt{
\qu^-_{\grey{j}}   \ar@{.}[d]&\qu^-_{\grey{i}}\ar@{.}[d]\\
\qu^+_{\grey{j,f(j)}}&\qu^+_{\grey{i,g(i)}}
}$}$}\right)\label{eq4}
\end{eqnarray}
swapping the two calls, even though they do give rise to separate
configurations in $\tau \odot \sigma$.
\end{exa}

So, symmetry classes of $+$-covered configurations are not the right
witness: they count only once symmetry classes that should intuitively
weight more, as they admit endo-symmetries that may affect the result.
Indeed one can correct this accounting by appropriately weighting groups
of endo-symmetries of symmetry classes -- see Appendix
\ref{app:weights}. This suggests links with generalized species of
structures \cite{fiore2008cartesian}; 
but our original question regarding what \emph{concrete objects} the
weighted relational model \emph{counts} remains open.

\subsubsection{Concrete witnesses} We now provide an alternative, more
concrete notion of witness.

Intuitively, we must find a refinement of symmetry still letting us
consider strategies up to their specific choice of copy indices, but
nevertheless keeping the two configurations of \eqref{eq4} separate.
Concretely, one may observe that the \emph{swap} of \eqref{eq4} is only
possible if Opponent changes their copy indices, exchanging $\grey{i}$
and $\grey{j}$. Here we use a fundamental property of our setting:
being \emph{thin concurrent games} (Definition \ref{def:tcg}), arenas
have sets of \emph{positive} and \emph{negative} symmetries
$\ptilde{A}$ and $\ntilde{A}$; so we may consider configurations of
strategies up to \emph{positive} symmetry only. This has a very strong
consequence:

\begin{lem}\label{lem:pos_symm}
Consider $A$ a game, $\sigma : A$ a strategy on $A$, and $\theta \in
\tilde{\sigma}$.

If $\pr_\sigma(\theta) \in \ptilde{A}$, then, $x = y$ and $\theta =
\id_x$.
\end{lem}
\begin{proof}
A quite direct consequence of \emph{thin}, which
prevents Player from imposing symmetries not prompted by a prior
Opponent exchange. See Lemma 3.28 in \cite{cg2}.
\end{proof}

In other words, sub-groupoids of $\tilde{\sigma}$ mapping to positive
symmetries of the game are all reduced to identities. Consider
now fixed, for any arena $A$, the choice of a representative $\rep{\x}_A
\in \x_A$ for any $\x_A \in \wconf{A}$. This invites the definition of
\textbf{positive witnesses}:
\begin{eqnarray}
\wit_\sigma^+(\x_A, \x_B) &=& \{x^\sigma \in \confp{\sigma} \mid
x^\sigma_A \sym_A^- \rep{\x}_A ~\&~x^\sigma_B \sym_B^+
\rep{\x}_B\}\,,\label{eq:defwitp}
\end{eqnarray}
for $\x_A \in \wconf{A}$ and $\x_B \in \wconf{B}$, which will indeed
turn out to be the right one. Notice the pleasant fact that this ranges
over actual configurations of $\sigma$ rather than symmetry
classes\footnote{An early sign that $\wit^+$ is better behaved is that
unlike $\swit$,
it does not depend on the choice of the symmetry for $\sigma$ -- recall
from Section A.1.2 in
\cite{cg2} that the symmetry is \emph{not} unique.}.

\subsection{Representability} We introduce our last technical
ingredient, \emph{representability}.

\subsubsection{Canonical representatives} The definition of $\wit^+$ in
\eqref{eq:defwitp} includes a significant subtlety: the dependency
on the choice of the representatives $\rep{\x}_A \in \x_A$, $\rep{\x}_B
\in \x_B$. Unfortunately, not only the set $\wit_\sigma^+(\x_A, \x_B)$
depends on the choice of $\rep{\x}_A$ and $\rep{\x}_B$, but even its
\emph{cardinality}:

\begin{figure}
\begin{minipage}{.24\linewidth}
\begin{eqnarray*}
\rep{\x}_A &=& 
\raisebox{10pt}{$
\scalebox{.8}{$
\xymatrix@R=5pt@C=5pt{
\qu^-_{\grey{0}}
	\ar@{.}[d]&
\qu^-_{\grey{1}}
	\ar@{.}[d]\\
\qu^+_{\grey{0, 0}}&
\qu^+_{\grey{1, 1}}
}$}$}\\
\rep{\x}_B &=& 
\raisebox{10pt}{$
\scalebox{.8}{$
\xymatrix@R=5pt@C=5pt{
\qu^-
        \ar@{.}[d]\\
\qu^+_{\grey{0}}
}$}$}
\end{eqnarray*}
\end{minipage}
\hfill
\begin{minipage}{.75\linewidth}
\[
\wit_\sigma^+(\x_A, \x_B) = 
\left\{
\raisebox{30pt}{$
\scalebox{.8}{$
\xymatrix@R=-8pt@C=-8pt{
\oc (\oc \alpha &\lin& \alpha)& \vdash &\oc \alpha &\lin &\alpha\\
&&&&&&\qu^-
	\ar@{-|>}[dllll]\\
&&\qu^+_{\grey{0}}
	\ar@{-|>}[dll]\\
\qu^-_{\grey{0,0}}
	\ar@{.}@/^/[urr]
	\ar@{-|>}[drr]\\
&&\qu^+_{\grey{h(0)}}
	\ar@{-|>}[dll]\\
\qu^-_{\grey{h(0),1}}
	\ar@{.}@/^/[urr]
	\ar@{-|>}[drrr]\\
&&&&\qu^+_{\grey{k(0,1)}}
	\ar@{.}@/^/[uuuuurr]
}$}$}
,
\raisebox{30pt}{$
\scalebox{.8}{$
\xymatrix@R=-8pt@C=-8pt{
\oc (\oc \alpha &\lin& \alpha)& \vdash &\oc \alpha &\lin &\alpha\\
&&&&&&\qu^-
	\ar@{-|>}[dllll]\\
&&\qu^+_{\grey{0}}
	\ar@{-|>}[dll]\\
\qu^-_{\grey{0,1}}
	\ar@{.}@/^/[urr]
	\ar@{-|>}[drr]\\
&&\qu^+_{\grey{h(1)}}
	\ar@{-|>}[dll]\\
\qu^-_{\grey{h(1),0}}
	\ar@{.}@/^/[urr]
	\ar@{-|>}[drrr]\\
&&&&\qu^+_{\grey{k(1,0)}}
	\ar@{.}@/^/[uuuuurr]
}$}$}
\right\}
\]
\end{minipage}

\caption{Witnesses for non-canonical representatives}
\label{fig:wit_noncan_rep}
\end{figure}

\begin{figure}
\begin{minipage}{.24\linewidth}
\begin{eqnarray*}
\rep{\x}_A &=& 
\raisebox{10pt}{$
\scalebox{.8}{$
\xymatrix@R=5pt@C=5pt{
\qu^-_{\grey{0}}
	\ar@{.}[d]&
\qu^-_{\grey{1}}
	\ar@{.}[d]\\
\qu^+_{\grey{0, 0}}&
\qu^+_{\grey{1, 0}}
}$}$}\\
\rep{\x}_B &=& 
\raisebox{10pt}{$
\scalebox{.8}{$
\xymatrix@R=5pt@C=5pt{
\qu^-
        \ar@{.}[d]\\
\qu^+_{\grey{0}}
}$}$}
\end{eqnarray*}
\end{minipage}
\hfill
\begin{minipage}{.75\linewidth}
\[
\wit_\sigma^+(\x_A, \x_B) = 
\left\{
\raisebox{30pt}{$
\scalebox{.8}{$
\xymatrix@R=-8pt@C=-8pt{
\oc (\oc \alpha &\lin& \alpha)& \vdash &\oc \alpha &\lin &\alpha\\
&&&&&&\qu^-
	\ar@{-|>}[dllll]\\
&&\qu^+_{\grey{0}}
	\ar@{-|>}[dll]\\
\qu^-_{\grey{0,0}}
	\ar@{.}@/^/[urr]
	\ar@{-|>}[drr]\\
&&\qu^+_{\grey{h(0)}}
	\ar@{-|>}[dll]\\
\qu^-_{\grey{h(0),0}}
	\ar@{.}@/^/[urr]
	\ar@{-|>}[drrr]\\
&&&&\qu^+_{\grey{k(0,0)}}
	\ar@{.}@/^/[uuuuurr]
}$}$}
\right\}
\]
\end{minipage}

\caption{Witnesses for canonical representatives}
\label{fig:wit_can_rep}
\end{figure}
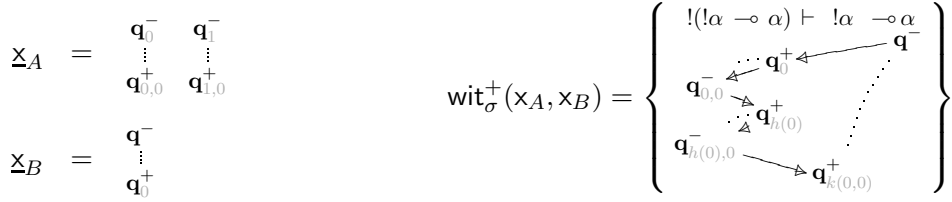

\begin{exa}
Consider $A = \oc (\oc \alpha \lin \alpha)$, $B = \oc \alpha \lin
\alpha$, $\x_A =
\left(
\raisebox{10pt}{$\scalebox{.8}{$\xymatrix@R=5pt@C=5pt{\qu^-\ar@{.}[d]&\qu^-\ar@{.}[d]\\\qu^+&\qu^+}$}$}\right)$
and $\x_B = \left(
\raisebox{10pt}{$\scalebox{.8}{$\xymatrix@R=5pt@C=5pt{\qu^-\ar@{.}[d]\\\qu^+}$}$}\right)$.

We show in Figures \ref{fig:wit_noncan_rep} and \ref{fig:wit_can_rep}
the sets $\wit^+_\sigma(\x_A, \x_B)$ for $\sigma : A$ the left hand side
strategy in Figure \ref{fig:ex_composition}, with the choice of
representatives as shown. Up to positive symmetry, Player is free to
associate $\qu^+_{\grey{0}}$ to either minimal event of $\rep{\x}_A$.
But as the symmetry is positive, the copy index of the subsequent
Opponent move is forced by $\rep{\x}_A$ and this association. For Figure
\ref{fig:wit_can_rep} the two choices make no difference as the bottom
events of $\rep{\x}_A$ have the same index $\grey{0}$. In contrast,
in Figure \ref{fig:wit_noncan_rep} these indices differ, and so yield
distinct concrete witnesses.   
\end{exa}

To explain this mismatch, it is helpful to explicitly factor in the
positive symmetries by
\begin{eqnarray}
\pswit_\sigma(\x_A, \x_B) &=& \{(\theta_A, x^\sigma, \theta_B) \mid 
x^\sigma \in \confp{\sigma},~\theta_A : \rep{\x}_A {\sym_A^-}
x^\sigma_A,~\theta_B : x^\sigma_B \sym_B^+
\rep{\x}_B\}\label{eq:witsym}
\end{eqnarray}
the set of \textbf{$\sim^+$-witnesses}, allowing us to prove the following property:

\begin{restatable}{prop}{cardinv}
Consider $A, B$ arenas, $\sigma : A \vdash B$, and $\x_A \in
\wconf{A}, \x_B \in \wconf{B}$. 

Then, the cardinality of $\pswit_\sigma^+(\x_A, \x_B)$ does not depend
on $\rep{\x}_A \in \x_A$ and $\rep{\x}_B \in \x_B$.
\end{restatable}
\begin{proof}
See Appendix \ref{app:invwit}.
\end{proof}

For instance, in Figure \ref{fig:wit_noncan_rep}, each positive witness
admits exactly \emph{one} pair of symmetries making it a
$\sim^+$-witness, whereas for Figure \ref{fig:wit_can_rep}, the unique
witness yields \emph{two} $\sim^+$-witnesses. 

So which of Figure \ref{fig:wit_noncan_rep} and \ref{fig:wit_can_rep} is
right, if any? Letting the weighted
relational model be the judge, Figure \ref{fig:wit_can_rep} is better:
indeed $\sigma$ is the strategy of a pure $\lambda$-term to which the
weighted relational model associates only weights $0$ and $1$. Tracking
down the pathological behaviour in Figure \ref{fig:wit_noncan_rep}, the
crux of the issue is the inability, in the representative $\rep{\x}_A$
from Figure \ref{fig:wit_noncan_rep}, to exchange the minimal
moves while staying in $\rep{\x}_A$. The only negative symmetry
\[
\raisebox{10pt}{$
\scalebox{.8}{$
\xymatrix@R=5pt@C=5pt{
\qu^-_{\grey{0}}
        \ar@{.}[d]&
\qu^-_{\grey{1}}
        \ar@{.}[d]\\
\qu^+_{\grey{0, 0}}&
\qu^+_{\grey{1, 1}}
}$}$}
\sym_A^-
\raisebox{10pt}{$
\scalebox{.8}{$
\xymatrix@R=5pt@C=5pt{
\qu^-_{\grey{1}}
        \ar@{.}[d]&
\qu^-_{\grey{0}}
        \ar@{.}[d]\\
\qu^+_{\grey{1, 0}}&
\qu^+_{\grey{0, 1}}
}$}$}
\]
(signified here by moves having the same position in
the diagram) 
with domain $\rep{\x}_A$ exchanging $\qu^-_{\grey{0}}$ and
$\qu^-_{\grey{1}}$ 
has codomain distinct from $\rep{\x}_A$, intuitively causing the extra
witness. In contrast, the analogous swap is an
endosymmetry of the representative $\rep{\x}_A$ for Figure
\ref{fig:wit_can_rep}. 

The next definition aims to capture the
representatives for which $\wit^+$ is well-behaved:

\begin{defi}\label{def:canonical}
Consider $A$ a game, and $x \in \conf{A}$.

We say that $x$ is \textbf{canonical} iff any $\theta : x \sym_A x$
factors uniquely as
\[ 
x \stackrel{\theta^-}{\sym_A^-} x \stackrel{\theta^+}{\sym_A^+} x\,,
\]
with in particular $x$ in the middle.
\end{defi}

This definition extends a basic property of thin concurrent games:

\begin{lem}\label{lem:factor}
Consider $A$ a thin concurrent game, and $\theta : x \sym_A y$ any
symmetry.

Then, there exist unique $z\in \conf{A}$, $\theta^- : x \sym_A z$ and
$\theta^+ : z \sym_A y$ such that $\theta = \theta^+ \circ \theta^-$.
\end{lem}
\begin{proof}
See Lemma 3.19 in \cite{cg2}.
\end{proof}

Likewise, any $\theta : x \sym_A y$ factors uniquely as
$\theta^- \circ \theta^+$ for some $\theta^- \in \ntilde{A}$,
$\theta^+ \in \ptilde{A}$. 

The representative $\rep{\x}_A$ of Figure \ref{fig:wit_can_rep} is
canonical, while that of Figure \ref{fig:wit_noncan_rep} is not. In
our collapse, we will need to ensure that we compute witnesses only on
canonical representatives.

\subsubsection{Representability}\label{subsubsec:representability}
This asks two questions: \emph{(1)}
does there always exist a canonical representative for any symmetry
class?; and \emph{(2)} is the cardinality of witnesses now invariant
under the choice of a canonical representative?

For \emph{(1)}, for tcgs as in Definition \ref{def:tcg} or games as in
Definition \ref{def:game}, the answer is no -- see Appendix
\ref{app:nonrep}.
For arenas as in Definition \ref{def:arena}, we do not know. Likewise,
though it is not hard to prove \emph{(2)} for games arising from $\PCF$
types, we do not have an answer in general.  So we must instead ask
games to carry an explicit choice of canonical representatives:

\begin{defi}
A game $A$ is \textbf{representable} when it comes equipped with a
function
\[
\rep{(-)}_A : \wconf{A} \to \nconf{A}
\]
such that for all $\x_A \in \wconf{A}$, $\rep{\x}_A \in \x_A$ is
canonical.
\end{defi}

This provides the condition left missing in Definition \ref{def:game};
but leaves us with the proof obligation of constructing the
representation function for all game constructions. For basic games
$1, \top, \gbool, \gnat, \alpha$ which have trivial symmetry,
the representation function is obvious. 
\begin{figure}
\begin{minipage}{.3\linewidth}
\begin{eqnarray*}
\rep{\x}_{A^\perp} &=& \rep{\x}_A\\
\rep{\x \parallel \y}_{A\tensor B} &=& \rep{\x}_A \parallel \rep{\x}_B
\end{eqnarray*}
\end{minipage}
\hfill
\begin{minipage}{.3\linewidth}
\begin{eqnarray*}
\rep{(1, \x)}_{A\with B} &=& (1, \rep{\x}_{A})\\
\rep{(2, \x)}_{A\with B} &=& (2, \rep{\x}_{B})
\end{eqnarray*}
\end{minipage}
\hfill
\begin{minipage}{.3\linewidth}
\begin{eqnarray*}
\rep{\x \parallel \y}_{A\parr B} &=& \rep{\x}_A \parr \rep{\x}_B\\
\rep{\x \lin \y}_{A\lin B} &=& \rep{\x}_A \lin \rep{\y}_B
\end{eqnarray*}
\end{minipage}
\caption{Representation functions for symmetry-free game constructions}
\label{fig:rep_basic}
\end{figure}
For the game constructions only propagating symmetry (\emph{i.e.} dual,
tensor, par, with, and linear arrow), we set the representation
function as specified in Figure \ref{fig:rep_basic} -- it is direct that
it preserves canonicity. Most importantly, for the bang construction, we
set
\[
\rep{[\x^1, \dots, \x^n]}_{\oc A} = \parallel_{1\leq i \leq n}
\rep{\x}^i_{A} \in \nconf{\oc A}
\]
relying on Lemma \ref{lem:r_bang}, assuming chosen a
sequential writing $[x_1, \dots, x_n]$ for every multiset.

This definition indeed always yield a canonical representative:

\begin{lem}
Consider $A$ a representable $-$-game.

Then $\oc A$, equipped with the function above, is representable.
\end{lem}
\begin{proof}
We must show that for all $\x = [\x^1, \dots, \x^n] \in \wconf{\oc A}$,
$\rep{\x}_{\oc A}$ is canonical. Consider
\[
\theta\,\,:\,\, \parallel_{1\leq i \leq n} \rep{\x}^i_{A}\,\,\sym_{\oc A}\,\,
\parallel_{1\leq i \leq n} \rep{\x}^i_{A}
\]
any symmetry. By definition, there is $\pi : \mathbb{N} \bij \mathbb{N}$
a permutation, and a family $(\theta_i)_{i \in \mathbb{N}} \in
\tilde{A}^{\mathbb{N}}$ s.t. for all $(i, a) \in \rep{\x}_{\oc A}$,
we have $\theta(i, a) = (\pi(i), \theta_i(a))$. But then, we have
$\theta_i : \rep{\x}^i_A \sym_A \rep{\x}^{\pi(i)}_A$
which means that $\x^i = \x^{\pi(i)}$, so that $\rep{\x}^i_A =
\rep{\x}^{\pi(i)}_A$ -- hence $\theta_i$ is an endo-symmetry.

Now, we use that $\rep{\x}^i_A$ is canonical, which entails that
$\theta_i$ factors as
\[
\rep{\x}^i_A\,\,\stackrel{\theta_i^-}{\sym_A^-}\,\,\rep{\x}^i_A
\,\,\stackrel{\theta_i^+}{\sym_A^+}\,\,\rep{\x}^i_A
\]
using which we may finally factor $\theta$ as $\theta^+ \circ \theta^-$
where $\theta^+(i, a) = (i, \theta^+_{\pi^{-1}(i)}(a))$ and $\theta^-(i, a) =
(\pi(i), \theta^-_i(a))$ -- it is direct by definition that $\theta^+ :
\rep{\x}_{\oc A} \sym_{\oc A}^+ \rep{\x}_{\oc A}$ and $\theta^- :
\rep{\x}_{\oc A} \sym_{\oc A}^- \rep{\x}_{\oc A}$.
\end{proof}

From now on, all games come equipped with a representation. 
As stated above, not every tcg admits a representation 
(see Appendix \ref{app:nonrep} for a counter-example). However,
non-representable games seem to lie outside of the interpretation of any
reasonable type.  

\subsection{Preservation of Composition}\label{subsec:pres_comp_main}
For $A, B$ arenas and $\sigma :
A \vdash B$, we may finally set:
\begin{eqnarray}
(\coll(\sigma))_{\x_A, \x_B} &=& \sharp \wit_\sigma^+(\x_A, \x_B)\,,
\label{eq:coll}
\end{eqnarray}
for any $\x_A \in \wconf{A}$, $\x_B \in \wconf{B}$, and where
$\wit_\sigma^+(\x_A, \x_B)$ is defined as in \eqref{eq:defwitp} using
the canonical representatives. 
Our aim is to prove the following equality:
\[
(\coll(\tau \odot \sigma))_{\x_A, \x_C} = \sum_{\x_B \in \wconf{B}}
(\coll(\sigma))_{\x_A, \x_B} \times (\coll(\tau))_{\x_B, \x_C}\,,
\]
and the natural route seems to be by setting up a bijection
\begin{eqnarray}
\wit_{\tau \odot \sigma}^+(\x_A, \x_C) &\bij& \sum_{\x_B \in \wconf{B}}
\wit_\sigma^+(\x_A, \x_B) \times \wit_\tau^+(\x_B, \x_C)\,.\label{eq6}
\end{eqnarray}

However, while it must hold, there does not seem to be any simple way of
constructing this bijection explicitly. We use an
indirect route, considering witnesses with symmetry.

\subsubsection{Interaction witnesses with symmetry} In \eqref{eq6} above,
going from right to left is problematic as we get triples $(\x_B,
x^\sigma, x^\tau)$ where, in general, there is no reason to have $x^\sigma_B =
x^\tau_B$. We do have $x^\sigma_B \sym_B^+ \rep{\x}_B$ and
$\rep{\x}_B \sym_B^- x^\tau_B$, but with no specified symmetries.

So to approach \eqref{eq6}, we shall start by studying synchronizable
pairs of 
\[
(\theta_A^-, x^\sigma, \theta_B^+) \in \pswit_\sigma(\x_A, \x_B)\,,
\qquad
(\Omega_B^-, x^\tau, \Omega_C^+) \in \pswit_\tau(\x_B, \x_C)\,,
\]
witnesses with symmetry as in \eqref{eq:witsym}. So we
have $x^\sigma \in \confp{\sigma}$, $x^\tau \in \confp{\tau}$ and
symmetries
\[
\xymatrix{
\rep{\x}_A&
x^\sigma_A
	\ar[l]_{\theta_A^-}&
x^\sigma_B
	\ar[r]^{\theta_B^+}&
\rep{\x}_B&
x^\tau_B\ar[l]_{\Omega_B^-}&
x^\tau_C\ar[r]^{\Omega_C^+}&
\rep{\x}_C
}
\]
allowing us to synchronize $x^\sigma$ and $x^\tau$ using Proposition
\ref{prop:sync_sym}.

We shall compare those with witnesses in the composition, starting by
defining:

\begin{defi}
Consider $A, B$ and $C$ arenas; $\sigma \in \Strat(A,B)$ and $\tau \in
\Strat(B, C)$ strategies; and $\x_A \in \wconf{A}$, $\x_B \in \wconf{B}$
and $\x_C \in \wconf{C}$.

The \textbf{interaction witnesses} on $\x_A, \x_B, \x_C$ is the set
of all $x^\tau \odot x^\sigma \in \confp{\tau \odot \sigma}$ s.t.
\[
x^\sigma_A \sym_A^- \rep{\x}_A\,,
\qquad
\qquad
x^\sigma_B = x^\tau_B \in \x_B\,,
\qquad
\qquad
x^\tau_C \sym_C^+ \rep{\x}_B\,,
\]
we write $\wit^+_{\sigma, \tau}(\x_A, \x_B, \x_C)$ for this set.
Likewise, the \textbf{$\sim^+$-interaction witnesses} comprise
\[
\theta_A^- : x^\sigma_A \,{\sym_A^-}\, \rep{\x}_A\,,
\qquad
x^\tau \odot x^\sigma \in \confp{\tau \odot \sigma}\,,
\qquad
\theta_C^+ : x^\tau_C \,{\sym_C^+}\, \rep{\x}_C\,,
\]
with $x^\sigma_B = x^\tau_B \in \x_B$. We write $\pswit_{\sigma,
\tau}(\x_A, \x_B, \x_C)$ for this set.
\end{defi}

Note that we obviously have, for any $\x_A \in \wconf{A}$ and $\x_C \in
\wconf{C}$
\begin{eqnarray}
\wit^+_{\tau \odot \sigma}(\x_A, \x_C) &\bij& \sum_{\x_B \in \wconf{B}}
\wit^+_{\sigma, \tau}(\x_A, \x_B, \x_C)\label{eq7}
\end{eqnarray}
as interaction witnesses are exactly witnesses of the composition with a
specified symmetry class $\x_B \in \wconf{B}$ in the middle. To
establish \eqref{eq6}, we shall now study a connection between
$\sim^+$-interaction witnesses and synchronizable pairs of
$\sim^+$-witnesses. 

\subsubsection{Synchronization up to symmetry} The property we shall
prove is a quantitative elaboration on Proposition \ref{prop:sync_sym},
so we start with an explicit reformulation of Proposition \ref{prop:sync_sym}.

\begin{lem}\label{lem:qbip1}
Consider $A, B, C$ arenas, $\sigma \in \Strat(A, B)$ and $\tau \in
\Strat(B, C)$. 

Then, for any $\x_A \in \wconf{A}, \x_B \in \wconf{B}$ and $\x_C \in
\wconf{C}$, for any pair of $\sim^+$-witnesses
\[
(\theta_A^-, x^\sigma, \theta_B^+) \in \pswit_\sigma^+(\x_A, \x_B)\,,
\qquad
(\Omega_B^-, x^\tau, \Omega_C^+) \in \pswit_\tau^+(\x_B, \x_C)\,,
\]
there are unique $\omega^\sigma : x^\sigma \sym_\sigma
y^\sigma, \nu^\tau : x^\tau \sym_\tau y^\tau$,
$\Theta_B : \rep{\x}_B \sym_B y_B$ and $\sim^+$-interaction witness
\[
(\psi_A^-, y^\tau \odot y^\sigma, \psi_C^+) \in \pswit_{\sigma, \tau}(\x_A,
\x_B, \x_C)
\]
with $y^\sigma_B = y^\tau_B = y_B$, such that the following diagrams commute:
\[
\xymatrix@R=4pt{
&x^\sigma_A  
	\ar[dl]_{\theta_A^-}
        \ar[dd]^{\omega_A^\sigma}&
x^\sigma_B   
	\ar[r]^{\theta_B^+}
        \ar[dd]_{\omega^\sigma_B}&
\rep{\x}_B      
        \ar[dd]^{\Theta_B}&
x^\tau_B   
	\ar[dd]^{\nu^\tau_B}
        \ar[l]_{\Omega_B^-}&
x^\tau_C   
	\ar[dr]^{\Omega_C^+}
        \ar[dd]_{\nu^\tau_C}\\
\rep{\x}_A&&&&&&\rep{\x}_C\\
&y^\sigma_A  
	\ar[ul]^{\psi_A^-}&
y^\sigma_B   \ar@{=}[r]&
y_B&
y^\tau_B   \ar@{=}[l]&
y^\tau_C   \ar[ur]_{\psi_C^+}
}
\]
\end{lem}
\begin{proof}
\emph{Existence.} First, by Lemma \ref{lem:deadlock_free}, the bijection
induced by $(\Omega_B^-)^{-1} \circ \theta_B^+ : x_B^\sigma \sym_B
x^\tau_B$ is secured. Thus we can apply Proposition \ref{prop:sync_sym},
yielding $y^\tau \odot y^\sigma \in \confp{\tau \odot \sigma}$ along
with 
\[
\omega^\sigma : y^\sigma \sym_\sigma x^\sigma\,,
\qquad
\qquad
\nu^\tau : y^\tau \sym_\tau y^\tau\,,
\]
such that $\nu^\tau_B \circ ((\Omega_B^-)^{-1} \circ \theta_B^+) =
\omega^\sigma_B$. Furthermore we may set $\psi_A^- = \theta_A^- \circ
(\omega^\sigma_A)^{-1}$ and $\psi_C^+ = \Omega_C^+ \circ
(\nu^\tau_C)^{-1}$ -- overall, the following diagrams commute
\[
\xymatrix@R=4pt{
&x^\sigma_A  
	\ar[dl]_{\theta_A^-}
        \ar[dd]^{\omega^\sigma_A}&
x^\sigma_B   
	\ar[r]^{\theta_B^+}
        \ar[dd]_{\omega^\sigma_B}&
\rep{\x}_B&
x^\tau_B   
	\ar[dd]^{\nu^\tau_B}
        \ar@{<-}[l]_{(\Omega_B^-)^{-1}}&
x^\tau_C   
	\ar[dr]^{\Omega_C^+}
        \ar[dd]_{\nu^\tau_C}\\
\rep{\x}_A&&&&&&\rep{\x}_C\\
&y^\sigma_A  
	\ar[ul]^{\psi_A^-}&
y^\sigma_B   
	\ar@{=}[r]&
y_B     \ar@{=}[r]&
y^\tau_B &
y^\tau_C   \ar[ur]_{\psi_C^+}
}
\]
which we complete by setting $\Theta_B : \rep{\x}_B \to y_B$ as either
path around the center diagram. 

\emph{Uniqueness.} First, $y^\tau \odot y^\sigma \in \confp{\tau \odot
\sigma}$, $\omega^\sigma, \nu^\tau$ are unique by the uniqueness clause
in Proposition \ref{prop:sync_sym}. It follows that $\psi_A^-$ and
$\psi_C^+$ and $\Theta_B$ are determined by the diagram.
\end{proof}

In particular, to $(\theta_A^-, x^\sigma, \theta_B^+) \in
\pswit_\sigma^+(\x_A, \x_B)$ and $(\Omega_B^-, x^\tau, \Omega_C^+) \in
\pswit_\tau^+(\x_B, \x_C)$ we have associated a $\sim^+$-interaction
witness $(\psi_A^-, y^\tau \odot y^\sigma, \psi_C^+) \in
\pswit_{\sigma, \tau}(\x_A, \x_B, \x_C)$. Counting-wise, it seems
$\sim^+$-interaction witnesses have fewer degrees of liberty than pairs
of $\sim^+$-witnesses as the latter have no symmetry on $B$; in the
lemma above this is mitigated by the fact that the construction also
extracts $\Theta_B$. 
The next step is to reverse this: from 
\[
(\psi_A^-, y^\tau \odot y^\sigma, \psi_C^+) \in \pswit_{\sigma,
\tau}(\x_A, \x_B, \x_C) 
\]
and $\Theta_B : \rep{\x}_B \sym_B y_B$, we must make $\Theta_B$ act on
$y^\sigma$ and $y^\tau$ to recover the original $x^\sigma$ and $x^\tau$.

\subsubsection{Negative symmetries acting on strategies} 
In thin concurrent games, strategies can always adjust their copy
indices to match a change in Opponent's copy indices:
%

\begin{lem}\label{lem:b4}
Consider $A$ a game, $\sigma : A$ a strategy, $x^\sigma \in \conf{\sigma}$ and
$\theta^- : x^\sigma_A \sym_A^- y_A$. 

Then, there are \emph{unique} $\varphi : x^\sigma \sym_\sigma y^\sigma$
and $\theta^+ : y_A \sym_A^+ y^\sigma_A$ such that
\[
\pr_\sigma \varphi = \theta^+ \circ \theta^- : x^\sigma_A \sym_A
y^\sigma_A\,.
\]
\end{lem}
\begin{proof}
See Lemma B.4 in \cite{cg2}.
\end{proof}

Opponent changes their copy
indices by applying the negative symmetry $\theta^- : x^\sigma_A
\sym_A^- y_A$, and Player adapts by applying the unique $\varphi :
x^\sigma \sym_\sigma y^\sigma$. The resulting configuration $y^\sigma_A$
might not be equal to $y_A$, but it is \emph{positively symmetric},
reflecting Player's adjusted indices.

\subsubsection{Symmetries acting on $\sim^+$-interaction witnesses}
We use this to reverse Lemma \ref{lem:qbip1}.

\begin{lem}\label{lem:qbip2}
Consider $A, B, C$ arenas, $\sigma \in \Strat(A, B)$ and $\tau \in
\Strat(B, C)$.

Then, for any $\x_A \in \wconf{A}, \x_B \in \wconf{B}$ and $\x_C \in
\wconf{C}$, $\sim^+$-interaction witness
\[
(\psi_A^-, y^\tau \odot y^\sigma, \psi_C^+) \in \pswit_{\sigma, \tau}(\x_A, \x_B,
\x_C) 
\]
with $y^\sigma_B = y^\tau_B = y_B$, any $\Theta_B :
\rep{\x}_B \sym_B y_B$,  there are unique 
$\omega^\sigma : x^\sigma \sym_\sigma y^\sigma$,
$\nu^\tau : x^\tau \sym_\tau y^\tau$ and
\[
(\theta_A^-, x^\sigma, \theta_B^+) \in \pswit_\sigma^+(\x_A, \x_B)\,,
\qquad
(\Omega_B^-, x^\tau, \Omega_C^+) \in \pswit_\tau^+(\x_B, \x_C)\,,
\]
a pair of $\sim^+$-witnesses, such that the following diagram commutes:
\[
\xymatrix@R=7pt{
&x^\sigma_A  
	\ar[dl]_{\theta_A^-}
        \ar[dd]^{\omega^\sigma_A}&
x^\sigma_B   
	\ar[r]^{\theta_B^+}
        \ar[dd]_{\omega^\sigma_B}&
\rep{\x}_B      
        \ar[dd]^{\Theta_B}&
x^\tau_B\ar[dd]^{\nu^\tau_B}
        \ar[l]_{\Omega_B^-}&
x^\tau_C\ar[dr]^{\Omega_C^+}
        \ar[dd]_{\nu^\tau_C}\\
\rep{\x}_A&&&&&&\rep{\x}_C\\
&y^\sigma_A  
	\ar[ul]^{\psi_A^-}&
y^\sigma_B   
	\ar@{=}[r]&
y_B&
y^\tau_B\ar@{=}[l]&
y^\tau_C\ar[ur]_{\psi_C^+}
}
\]  
\end{lem}
\begin{proof}
The first step is to factor $\Theta_B^{-1}$ in two ways, as in the
diagram
\[
\xymatrix@R=7pt{
&&&z_B^1
        \ar[r]^{\Phi_B^+}&
\rep{\x}_B&
z_B^2   \ar[l]_{\Psi_B^-}\\
\rep{\x}_A&&&&&&&&\rep{\x}_C\\
&y^\sigma_A  
	\ar[ul]^{\psi_A^-}&
y^\sigma_B   
	\ar@{=}[rr]&&
y_B     \ar[uul]^{\Phi_B^-}
        \ar[uu]_{\Theta_B^{-1}}
        \ar[uur]_{\Psi_B^+}
        \ar@{=}[rr]&&
y^\tau_B&
y^\tau_C   
	\ar[ur]_{\psi_C^+}
}
\] 
following Lemma \ref{lem:factor}. By Lemma \ref{lem:b4} we can make
$\Phi_B^-$ act on $x^\sigma$. This yields
\[
\lambda_A^- : x^\sigma_A \sym_A^- y^\sigma_A\,,
\qquad
\omega^\sigma : x^\sigma \sym_\sigma y^\sigma\,,
\qquad
\Delta_B^+ : x^\sigma_B \sym_B^+ z^1_B\,,
\]
unique such that the following diagram commutes:
\[
\xymatrix@R=7pt{
&&x^\sigma_A 
	\ar[dl]_{\lambda_A^-}
        \ar[dd]^{\omega^\sigma_A}&
x^\sigma_B   
	\ar[dd]^{\omega^\sigma_B}
        \ar[r]^{\Delta_B^+}&
z^1_B   \ar@[grey][r]^{\grey{\Phi_B^+}}&
\grey{\rep{\x}_B}&
\grey{z^2_B}
        \ar@[grey][l]_{\grey{\Psi_B^-}}\\
\rep{\x}_A&
y^\sigma_A   \ar[l]^{\psi_A^-}
        \ar@{=}[dr]&&&&&&&&\grey{\rep{\x}_C}\\
&&y^\sigma_A&y^\sigma_B
        \ar@{=}[rr]&&
y_B     \ar[uul]^{\Phi_B^-}
        \ar@[grey][uu]_{\grey{\Theta_B^{-1}}}
        \ar@[grey][uur]_{\grey{\Psi_B^+}}
        \ar@[grey]@{=}[rr]&&
\grey{y^\tau_B}&
\grey{y^\tau_C}
        \ar@[grey][ur]_{\grey{\psi_C^+}}
}
\]
leaving in grey the irrelevant parts of the full diagram for context.
Setting
$\theta_A^- = \psi_A^- \circ \lambda_A^-$ and $\theta_B^+ = \Phi_B^+
\circ \Delta_B^+$, we have found data making the following diagram commute:
\[
\xymatrix@R=7pt@C=15pt{
&x^\sigma_A  
	\ar[dl]_{\theta_A^-}
        \ar[dd]^{\omega^\sigma_A}&
x^\sigma_B   
	\ar[dd]^{\omega^\sigma_B}
        \ar[rr]^{\theta_B^+}&&
\rep{\x}_B
        \ar[dd]_{\Theta_B}&
\grey{z^2_B}
        \ar@[grey][l]_{\grey{\Psi_B^-}}\\
\rep{\x}_A&&&&&&&&\grey{\rep{\x}_C}\\
&y^\sigma_A  
	\ar[ul]^{\psi_A^-}&
y^\sigma_B   
	\ar@{=}[rr]&&
y_B     \ar@[grey]@{=}[rr]
        \ar@[grey][uur]_{\grey{\Psi_B^+}}&&
\grey{y^\tau_C}&
\grey{y^\tau_C}
        \ar@[grey][ur]_{\grey{\psi_C^+}}
}
\]

We shall now prove uniqueness of this data. Assume that we have other
symmetries
$\gamma_A^- : u^\sigma_A \sym_A^- \rep{\x}_A$, 
$\varpi^\sigma : u^\sigma \sym_S y^\sigma$ and
$\gamma_B^+ : u^\sigma_B \sym_B^+ \rep{\x}_B$
making the following diagram commute:
\[
\xymatrix@R=4pt@C=15pt{
&u^\sigma_A  
	\ar[dl]_{\gamma_A^-}
        \ar[dd]^{\varpi^\sigma_A}&
u^\sigma_B   
	\ar[dd]^{\varpi^\sigma_B}
        \ar[rr]^{\gamma_B^+}&&
\rep{\x}_B
        \ar[dd]_{\Theta_B}&
\grey{z^2_B}
        \ar@[grey][l]_{\grey{\Psi_B^-}}\\
\rep{\x}_A&&&&&&&&\grey{\rep{\x}_C}\\
&y^\sigma_A  
	\ar[ul]^{\psi_A^-}&
y^\sigma_B   
	\ar@{=}[rr]&&
y_B     \ar@[grey]@{=}[rr]
        \ar@[grey][uur]_{\grey{\Psi_B^+}}&&
\grey{y^\tau_C}&
\grey{y^\tau_C}
        \ar@[grey][ur]_{\grey{\psi_C^+}}
}
\]

Then, it follows that the following diagram also commutes:
\[
\xymatrix@R=4pt{
&&u^\sigma_A 
	\ar[dl]_{(\psi_A^-)^{-1} \circ \gamma_A^-}
        \ar[dd]^{\varpi^\sigma_A}&
u^\sigma_B   
	\ar[dd]^{\varpi^\sigma_B}
        \ar[rr]^{(\Phi_B^+)^{-1} \circ \gamma_B^+}&&
z^1_B   \ar@[grey][r]^{\grey{\Phi_B^+}}&
\grey{\rep{\x}_B}&
\grey{z^2_B}
        \ar@[grey][l]_{\grey{\Psi_B^-}}\\
\rep{\x}_A&
y^\sigma_A   
	\ar[l]^{\psi_A^-}
        \ar@{=}[dr]&&&&&&&&&\grey{\rep{\x}_C}\\
&&y^\sigma_A&y^\sigma_B
        \ar@{=}[rrr]&&&
y_B     \ar[uul]^{\Phi_B^-}
        \ar@[grey][uu]_{\grey{\Theta_B^{-1}}}
        \ar@[grey][uur]_{\grey{\Psi_B^+}}
        \ar@[grey]@{=}[rr]&&
\grey{y^\tau_B}&
\grey{y^\tau_C}
        \ar@[grey][ur]_{\grey{\psi_C^+}}
}
\]

By uniqueness for Lemma \ref{lem:b4}, it follows that $u^\sigma =
x^\sigma$, $\omega^\sigma = \varpi^\sigma$,  $\lambda_A^- =
(\psi_A^-)^{-1} \circ \gamma_A^-$ so $\gamma_A^- = \theta_A^-$, and
$(\Phi_B^+)^{-1} \circ \gamma_B^+ = \Delta_B^+$ so $\gamma_B^+ =
\theta_B^+$. Altogether, we have proved that there are 
\[
\theta_A^- : x^\sigma_A \sym_\sigma \rep{\x}_A\,,
\qquad
\omega^\sigma : x^\sigma \sym_\sigma y^\sigma\,,
\qquad
\theta_B^+ : x^\sigma_B \sym_B^+ \rep{\x}_B\,,
\]
unique making the following diagram commute:
\[
\xymatrix@R=4pt@C=15pt{
&x^\sigma_A  
	\ar[dl]_{\theta_A^-}
        \ar[dd]^{\omega^\sigma_A}&
x^\sigma_B   
	\ar[dd]^{\omega^\sigma_B}
        \ar[rr]^{\theta_B^+}&&
\rep{\x}_B
        \ar[dd]_{\Theta_B}&
\grey{z^2_B}
        \ar@[grey][l]_{\grey{\Psi_B^-}}\\
\rep{\x}_A&&&&&&&&\grey{\rep{\x}_C}\\
&y^\sigma_A  \ar[ul]^{\psi_A^-}&
y^\sigma_B   \ar@{=}[rr]&&
y_B     \ar@[grey]@{=}[rr]
        \ar@[grey][uur]_{\grey{\Psi_B^+}}&&
\grey{y^\tau_C}&
\grey{y^\tau_C}
        \ar@[grey][ur]_{\grey{\psi_C^+}}
}
\]

The lemma follows by performing the exact same reasoning on the right
hand side.
\end{proof}

\subsubsection{The interaction bijection} 
If $B$ is a game and $\x_B \in
\wconf{B}$, let us write $\tilde{\x_B}$ for the set of endosymmetries on
$\rep{\x}_B$, \emph{i.e.} symmetries $\theta_B : \rep{\x}_B \sym_B
\rep{\x}_B$. We shall use the fact that for any $x, y \in \x_B$, there
are exactly as many symmetries $x \sym_B y$ as in $\tilde{\x_B}$.
Indeed let us fix, for any $x \in \x_B$, a symmetry $\kappa_x : x
\sym_B \rep{\x}_B$. For any $x, y \in \x_B$, we then have
\[
\begin{array}{rcrcl}
(-)[x, y] &:& \tilde{\x_B} &\to& \tilde{B}(x, y)\\
&&\theta &\mapsto& \kappa_y^{-1} \circ \theta \circ \kappa_x
\end{array}
\]
writing $\tilde{B}(x, y)$ for the set of all $\theta : x \sym_B y$. It
is elementary that this is a bijection.

Using this, we finally have, for any $A, B, C$ games and $\sigma \in
\Strat(A, B)$, $\tau \in \Strat(B, C)$:

\begin{cor}\label{cor:main}
Fix $\x_A \in \wconf{A}, \x_B \in \wconf{B}$ and $\x_C \in
\wconf{C}$. 
Then, there is a bijection
\[
\Upsilon
\quad:\quad
\pswit_\sigma(\x_A, \x_B) \times \pswit_\tau(\x_B, \x_C) 
\quad\bij\quad
\pswit_{\sigma, \tau}(\x_A, \x_B, \x_C) \times \tilde{\x_B}
\]
such that for any 
$\Upsilon(x^\sigma, x^\tau) = (y^\tau \odot y^\sigma, \Theta)$, we have
$x^\sigma \sym_\sigma y^\sigma$ and $x^\tau \sym_\tau y^\tau$.
\end{cor}
\begin{proof}
Given $(\theta_A^-, x^\sigma, \theta_B^+) \in \pswit_\sigma(\x_A,
\x_B)$ and $(\Omega_B^-, x^\tau, \Omega_C^+) \in \pswit_\tau(\x_B,
\x_C)$, we get
\[
(\psi_A^-, y^\tau \odot y^\sigma, \psi_C^+) \in \pswit_{\sigma,
\tau}(\x_A,
\x_B, \x_C)
\]
and $\Theta_B : \rep{\x}_B \sym_B y_B$ by Lemma \ref{lem:qbip1}; so we
set $\Upsilon((\theta_A^-, x^\sigma, \theta_B^+), (\Omega_B^-, x^\tau,
\Omega_C^+))$ as:
\[
((\psi_A^-, y^\tau \odot y^\sigma, \psi_C^+), (\kappa_{y_B} \circ
\Theta_B))\,.
\]

Reciprocally, given $(\psi_A^-, y^\tau \odot y^\sigma, \psi_C^+) \in
\pswit_{\sigma, \tau}(\x_A, \x_B, \x_C)$ and $\Xi_B \in \tilde{\x_B}$,
we set $\Theta_B = \kappa_{y_B}^{-1} \circ \Xi_B$ and apply Lemma
\ref{lem:qbip2} to get back two $\sim^+$-witnesses:
\begin{eqnarray*}
(\theta_A^-, x^\sigma, \theta_B^+) &\in& \pswit_\sigma(\x_A, \x_B)\\
(\Omega_B^-, x^\tau, \Omega_C^+) &\in& \pswit_\tau(\x_B, \x_C)
\end{eqnarray*}

That these constructions are inverses of each other is an immediate
consequence from the uniquess properties in Lemmas \ref{lem:qbip1} and
\ref{lem:qbip2}.
\end{proof}

\subsubsection{Preservation of composition} 
Extending earlier notations, for any $\x_A \in \wconf{A}$, we write
$\ptilde{\x_A}$ for the group of positive endo-symmetries on
$\rep{\x}_A$, and likewise for $\ntilde{\x_A}$. As for general
symmetries, if $x, y \in \x_A$ s.t. $x \sym_A^+ \rep{\x}_A$ and $y
\sym_A^+ \rep{\x}_A$, then there is a bijection
\[
\begin{array}{rcrcl}
(-)[x, y] &:& \ptilde{\x_A} &\bij & \ptilde{A}(x, y)\\
&& \theta^+ &\mapsto& (\kappa_y^+)^{-1} \circ \theta^+ \circ \kappa_x^+
\end{array}
\]
for $\ptilde{A}(x, y)$ the set of $\theta^+ : x \sym_A^+ y$; and having
chosen a positive symmetry $\kappa_x^+ : x \sym_A^+ \rep{\x}_A$ for all
$x$ positively symmetric to $\rep{\x}_A$ -- the same hold for negative
symmetries. Thus: 

\begin{lem}\label{lem:pswit_wit}
Consider strategies $\sigma \in \Strat(A, B)$, $\tau \in \Strat(B, C)$,
and symmetry classes $\x_A \in \wconf{A}, \x_B \in \wconf{B}$ and $\x_C
\in \wconf{C}$. Then we have bijections
\[
\begin{array}{rcrcl}
\Psi &:& \pswit_\sigma(\x_A, \x_B) &\bij& \ntilde{\x_A} \times 
\ptilde{\x_B} \times \wit^+_\sigma(\x_A, \x_B)\\
\Xi &:& \pswit_{\sigma, \tau}(\x_A, \x_B, \x_C) &\bij& \ntilde{\x_A} \times 
\ptilde{\x_C} \times \wit^+_{\sigma,\tau}(\x_A, \x_B,\x_C)
\end{array}
\]
s.t. for all $(\theta_A^-, x^\sigma, \theta_B^+) \in
\pswit_\sigma(\x_A, \x_B)$, writing $\Psi(\theta_A^-, x^\sigma,
\theta_B^+) = (\psi_A^-, \psi_B^+, y^\sigma)$, $x^\sigma =
y^\sigma$; and likewise, writing $\Xi(\theta_A^-, x^\tau\odot x^\sigma,
\theta_C^+) =
(\psi_A^-, \psi_C^+, y^\tau \odot y^\sigma)$, then $x^\sigma = y^\sigma$ and
$x^\tau = y^\tau$.
\end{lem}
\begin{proof}
To $(\theta_A^-, x^\sigma, \theta_B^+) \in \pswit_\sigma(\x_A, \x_B)$,
simply associate 
\[
\Psi(\theta_A^-, x^\sigma, \theta_B^+) = 
(\theta_A^- \circ (\kappa_{x^\sigma_A}^-)^{-1},
\theta_B^+ \circ (\kappa_{x^\sigma_B}^+)^{-1},
 x^\sigma) \in \ntilde{\x_A} \times 
\ptilde{\x_B} \times \wit^+_\sigma(\x_A, \x_B)\,,
\]
it is straightforward that this is a bijection. The proof for $\Xi$ is
the same.
\end{proof}

We now compose these bijections, to obtain:

\begin{lem}\label{lem:mainbij}
Consider strategies $\sigma \in \Strat(A, B)$, $\tau \in \Strat(B, C)$,
and symmetry classes $\x_A \in \wconf{A}, \x_B \in \wconf{B}$ and $\x_C
\in \wconf{C}$. Then we have a bijection:
\[
\begin{array}{rcl}
\Phi &:& \ntilde{\x_A} \times \tilde{\x_B} \times \ptilde{\x_C} \times
\wit^+_{\sigma, \tau}(\x_A, \x_B, \x_C)\\
&& \qquad \bij \quad
\ntilde{\x_A} \times \tilde{\x_B} \times \ptilde{\x_C} \times
\wit^+_{\sigma}(\x_A, \x_B) \times \wit^+_\tau(\x_B, \x_C)
\end{array}
\]
such that writing $\Phi(\theta_A^-, \theta_B, \theta_C^+, x^\tau \odot
x^\sigma) = (\varphi_A^-, \varphi_B, \varphi_C^+, y^\sigma, y^\tau)$, 
$x^\sigma \sym_\sigma y^\sigma$
and $x^\tau \sym_\tau y^\tau$.
\end{lem}
\begin{proof}
The bijection is obtained through the following composition:
\begin{eqnarray*}
&& \ntilde{\x_A} \times \tilde{\x_B} \times \ptilde{\x_C} \times
\wit^+_{\sigma, \tau}(\x_A, \x_B, \x_C)\\
&\bij& \tilde{\x_B} \times \pswit_{\sigma, \tau}(\x_A, \x_B, \x_C)\\
&\bij& \pswit_\sigma(\x_A, \x_B) \times \pswit_\tau(\x_B, \x_C)\\
&\bij& \ntilde{\x_A} \times \wit^+_\sigma(\x_A, \x_B) \times
\ptilde{\x_B} \times \ntilde{\x_B} \times \wit^+_\tau(\x_B, \x_C) \times
\ptilde{\x_C}\\
&\bij& \ntilde{\x_B} \times \tilde{\x_B} \times \ptilde{\x_C} \times
\wit^+_\sigma(\x_A, \x_B) \times \wit^+_\tau(\x_B, \x_C)\,.
\end{eqnarray*}
using $\Xi$ in Lemma \ref{lem:pswit_wit}, then $\Upsilon^{-1}$ in
Corollary \ref{cor:main}, then $\Psi$ from Lemma \ref{lem:pswit_wit} for
$\sigma$ and $\tau$, and finally the bijection $\tilde{\x_B} \bij
\ptilde{\x_B} \times \ntilde{\x_B}$ coming from the canonicity of
$\rep{\x}_B$.
That $\Phi$ preserves symmetry classes in $\sigma$ and $\tau$
is an immediate verification.
\end{proof}

From $\Phi$, it immediately follows for $\x_A \in \wconf{A}$, $\x_B \in
\wconf{B}$ and $\x_C \in \wconf{C}$ we have
\begin{eqnarray}
\sharp \wit^+_{\sigma, \tau}(\x_A, \x_B, \x_C) = \sharp
\wit^+_{\sigma}(\x_A, \x_B) \times \sharp \wit^+_\tau(\x_B,
\x_C)\,,\label{eq:sumup}
\end{eqnarray}
however there is no clear way to realize the corresponding bijection
directly, without invoking symmetries. Of course the bijection must
exist for cardinality reasons, but then it may not preserve symmetry
classes in $\sigma$ and $\tau$ -- which is necessary for the
quantitative generalization in Section \ref{sec:coll_rw}. Nevertheless,
\eqref{eq:sumup} allows us to conclude the core result of the paper:

\begin{cor}\label{cor:functor}
Consider $\sigma \in \Strat(A, B)$ and $\tau \in \Strat(B, C)$. Then,
\[
\coll(\tau \odot \sigma)_{\x_A, \x_C} = \sum_{\x_B \in \wconf{B}}
\coll(\sigma)_{\x_A, \x_B} \times \coll(\tau)_{\x_B, \x_C}
\]
for all $\x_A \in \wconf{A}$ and $\x_C \in \wconf{C}$.
\end{cor}
\begin{proof}
We perform the following direct computation, using \eqref{eq7} and \eqref{eq:sumup}.
\begin{eqnarray*}
\sharp \wit^+_{\tau \odot \sigma}(\x_A, \x_C)&=& \sum_{\x_B \in
\wconf{B}} \sharp \wit^+_{\sigma, \tau}(\x_A, \x_B, \x_B)\\
&=& \sum_{\x_B \in \wconf{B}} \sharp \wit^+_\sigma(\x_A, \x_B) \times
\wit^+_\tau(\x_B, \x_C)\,.\qedhere
\end{eqnarray*}
\end{proof}

\section{Preservation of the Interpretation}
\label{sec:pres_intr}

Now that preservation of composition is clear, we deal with the rest of the interpretation.

\subsection{Structure-preserving functors} We first set up the
categorical machinery.

\subsubsection{Cartesian closed functors} 
We start with cartesian closed functors, the 
appropriate notion of morphisms between cartesian closed categories, 
preserving the interpretation of the simply-typed $\lambda$-calculus.
This can be straightforwardly adaptated to $\sim$-categories.

\begin{defi}
Let $\C, \D$ be cartesian closed $\sim$-categories. A $\sim$-functor 
\[
F : \C \to \D
\]
is \textbf{cartesian closed} if it comes equipped with 
for any $A, B \in \C_0$, maps
\[
\begin{array}{rcrcl}
k^\top &:& \top &\to& F\top\\
k_{A, B}^\with &:& FA \with FB &\to& F(A\with B)\\
k_{A, B}^\tto &:& FA \tto FB &\to& F(A\tto B)
\end{array}
\]
invertible up to $\sim$; and such that 
the following diagrams commute up to $\sim$:
\[
\scalebox{.8}{$
\xymatrix{
&FA\\
FA \with FB
	\ar[rr]^{k_{A, B}^\with}
	\ar[dr]_{\pi_2}
	\ar[ur]^{\pi_1}&&
F(A \with B)
	\ar[dl]^{F(\pi_2)}
	\ar[ul]_{F(\pi_1)}\\
&FB
}$}
\qquad
\qquad
\raisebox{-15pt}{$
\scalebox{.8}{$
\xymatrix@C=50pt{
F(A\tto B) \with FA
	\ar[r]^{k^\with_{A\tto B, A}}&
F((A\tto B)\with A)
	\ar[d]^{F(\evm_{A, B})}\\
(FA \tto FB)\with FA
	\ar[u]^{k^\tto_{A, B} \with FA}
	\ar[r]_{\evm_{FA, FB}}&
FB
}
$}$}
\]
\end{defi}

We use notions of cartesian closed categories with
explicit structure, which must be preserved up to isomorphism. It is not
necessary to require that $t^\with$ and $t^\tto$ are natural (up to
$\sim$); this automatically follows. Likewise, preservation
of projections and evaluation suffice to ensure that pairing and
currying are also preserved.

In fact, cartesian closed $\sim$-functors ensure
preservation up to isomorphism of the interpretation of the simply-typed
$\lambda$-calculus, in the following sense: assume chosen 
\[
k^\alpha : \intr{\alpha}_\D \to F(\intr{\alpha}_\C)
\]
an isomorphism for any base type $\alpha$. Then, by induction on types
one can form isos
\[
\begin{array}{rcrcl}
k^\Ty_A &:& \intr{A}_\D &\to& F(\intr{A}_\C)\\
k^\Ctx_\Gamma &:& \intr{\Gamma}_\D &\to& F(\intr{\Gamma}_\D)
\end{array}
\]
in $\D$ for every type $A$ and context $\Gamma$; it is then a lengthy 
exercise to prove that 
\[
\xymatrix{
\intr{\Gamma}_\D
	\ar[r]^{\intr{M}_\D}
	\ar[d]_{k^\Ctx_\Gamma}
	\ar@{}[dr]|\sim&
\intr{A}_\C
	\ar[d]^{k^\Ty_A}\\
F(\intr{\Gamma}_\C)
	\ar[r]_{F(\intr{M}_\C)}&
F(\intr{A}_\C)
}
\]
for every simply-typed $\lambda$-term $\Gamma \vdash M : A$; \emph{i.e.}
$F$ preserves the interpretation up to iso.

\subsubsection{Relative Seely $\sim$-functors} 

\begin{defi}\label{def:relseelyfunctors}
\changed{
Let $\C, \D$ be relative Seely $\sim$-categories. A \textbf{relative Seely ($\sim$-)functor} $\C \to \D$ is a functor $F : \C \to \D$ which restricts to $F : \C_s \to \D_s$, equipped with:
\begin{itemize}
\item 
For every $A, B\in \C$, morphisms
\[
\begin{array}{rcrcl}
t^\tensor_{A,B} &:& FA \tensor FB &\to& F(A\tensor B)\\
t^1 &:& 1 &\to& F1;
\end{array}
\]
making $(F, t^\tensor, t^1)$ a symmetric monoidal functor $(\C, \tensor, 1) \to (\D, \tensor, 1)$;
\item for every $S, T\in \C_s$, morphisms
\[
\begin{array}{rcrcl}
t^\with_{S,T} &:& FS \with FT &\to& F(S\with T)\\
t^\top &:& \top &\to& F\top;
\end{array}
\]
\item For every $A \in \C$ and $S \in \C_s$, a morphism
\[
\begin{array}{rcrcl}
t^\lin_{A,S} &:& FA \lin FS &\to& F(A\lin S);
\end{array}
\]
\item For every $S \in \C_s$, a morphism
\[
\begin{array}{rcrcl}
t^\oc_S &:& \oc FS &\to& F\oc S;
\end{array}
\]
\end{itemize}
all of which are invertible up to $\sim$ and satisfy the coherence axioms of Figure~\ref{fig:seely-functors} up to $\sim$, such that for every $S, T \in \C_s$ and $f : \oc S \to T$, the diagram }
\begin{equation}
\label{eq:preservationpromotion}
\begin{tikzcd}[column sep=3em]
\oc F S \arrow[swap]{d}{t^\oc_S} \arrow{r}{(Ff \circ t^\oc_S)^\dagger} & \oc F T \arrow{d}{t^\oc_T} \\ 
F \oc S \arrow{r}{F(f^\dagger)} & F(\oc T)  
\end{tikzcd}
\end{equation}
commutes up to $\sim$.
\end{defi}

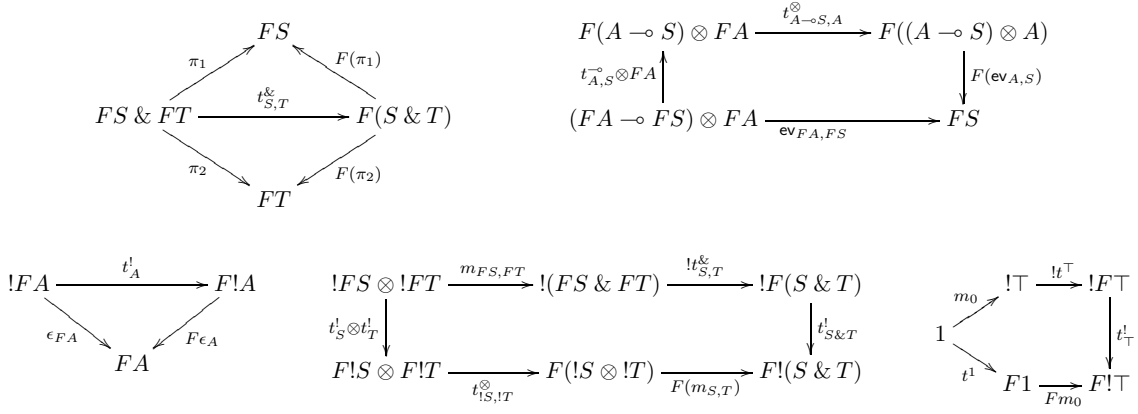
\begin{figure}
\begin{mathpar}
\scalebox{.8}{$
\xymatrix{
&FS\\
FS \with FT
        \ar[rr]^{t_{S, T}^\with}
        \ar[dr]_{\pi_2}
        \ar[ur]^{\pi_1}&&
F(S \with T)
        \ar[dl]^{F(\pi_2)}
        \ar[ul]_{F(\pi_1)}\\
&FT
}$}
\and
\scalebox{.8}{$
\xymatrix@C=50pt{
F(A\lin S) \tensor FA
        \ar[r]^{t^\tensor_{A\lin S, A}}&
F((A\lin S)\tensor A)
        \ar[d]^{F(\evm_{A, S})}\\
(FA \lin FS)\tensor FA
        \ar[u]^{t^\lin_{A, S} \tensor FA}
        \ar[r]_{\evm_{FA, FS}}&
FS
}
$}
\and
\scalebox{.8}{$
\xymatrix{
\oc F A	\ar[rr]^{t^\oc_A}
	\ar[dr]_{\der_{FA}}&&
F\oc A	\ar[dl]^{F\der_A}\\
&FA
}
$}
\and
\scalebox{.8}{$
\xymatrix@C=40pt{
\oc FS \tensor \oc FT
	\ar[r]^{m_{FS, FT}}
	\ar[d]_{t^\oc_S \tensor t^\oc_T}&
\oc (FS \with FT)
	\ar[r]^{\oc t^\with_{S, T}}&
\oc F(S\with T)
	\ar[d]^{t^\oc_{S \with T}}\\
F\oc S \tensor F \oc T
	\ar[r]_{t^\tensor_{\oc S, \oc T}}&
F(\oc S\tensor \oc T)
	\ar[r]_{F(m_{S, T})}&
F\oc (S\with T)
}
$}
\and
\scalebox{.8}{$
\xymatrix@C=20pt@R=10pt{
&\oc \top
	\ar[r]^{\oc t^\top}&
\oc F\top
	\ar[dd]^{t^\oc_\top}\\
1	\ar[ur]^{m_0}
	\ar[dr]_{t^1}\\
&F1	\ar[r]_{Fm_0}&
F\oc \top
}
$}
\end{mathpar}
\caption{Coherence diagrams for relative Seely functors}
\label{fig:seely-functors}
\end{figure}

In this paper, we only use the following property of relative
Seely functors:
\begin{prop}\label{prop:seely_functor_kleisli}
\changed{ A relative Seely $\sim$-functor $F : \C \to \D$ induces a cartesian closed $\sim$-functor $F_\oc : \C_\oc \to \D_\oc$ defined by 
 $F_\oc(S) = F(S)$ for all
$S \in \C_\oc$, 
\[
F_\oc(f) = (F f) \circ t^\oc_S \in \D_\oc(FS, FT)
\]
for all $f \in \C_\oc(S, T)$, and equipped with the following structural isomorphims:
\[
\begin{array}{rclcl}
k^\top &=& t^\top \circ \der_\top &\in& \D_\oc(\top, F\top)\\
k^\with_{S, T} &=& t^\with_{S, T} \circ \der_{FS\with FT} &\in&
\D_\oc(FS \with FT, F(S\with T))\\
k^\tto_{S, T} &=& t^\lin_{FS, T} \circ ((t^\oc_S)^{-1} \lin FT) \circ
\der_{FS\tto FT} &\in& \D_\oc(FS \tto FT, F(S\tto T)).
\end{array}
\]
}
\end{prop}
\begin{proof}
A lengthy but direct diagram chase.
\end{proof}

This sets most of the proof obligations for proving soundness of the
collapse from $\Strat$ to $\Rel{\N}$: we must show that $\coll(-)$
yields a \changed{relative Seely} $\sim$-functor from $\Strat$ to $\Rel{\N}$.

\subsection{A symmetric monoidal $\sim$-functor} \changed{We define the functor $\coll(-)$ and equip it with relative Seely structure, starting with symmetric monoidal structure.} 

\subsubsection{A $\sim$-functor} As expected, on
arenas we set $\coll(A) = \wconf{A}$. From Corollary \ref{cor:functor},
we already have an operation preserving composition
\[
\coll(-) : \Strat \to \Rel{\N}.
\]
To get a $\sim$-functor, it remains to check
that $\coll(-)$ preserves identities and $\sim$.

\begin{prop}\label{prop:cc_collapse}
Consider $A$ an arena, and $\x_A, \y_A \in \wconf{A}$. Then,
$\coll(\cc_A)_{\x_A, \y_A} = \delta_{\x_A, \y_A}$.
\end{prop}
\begin{proof}
First, assume $\x_A \neq \y_A$ and, seeking a contradiction,
consider
$z_A \parallel z_A \in \wit^+_{\cc_A}(\x_A, \y_A)$,
relying on Proposition \ref{prop:cc_pcov} for the shape of $+$-covered
configurations. So there are
\[
\theta_A^- : z_A \sym_A^- \rep{\x}_A\,,
\qquad
\qquad
\theta_A^+ : z_A \sym_A^+ \rep{\y}_A\,,
\]
and so $\rep{\x}_A \sym_A \rep{\y}_A$ by composition, contradiction.

Next we must show that for all $\x_A \in \wconf{A}$,
$\wit_{\cc_A}^+(\x_A, \x_A)$ has exactly one element. First, it is
immediate by Proposition \ref{prop:cc_pcov} that $\rep{\x}_A \parallel
\rep{\x}_A \in \wit_{\cc_A}^+(\x_A, \x_A)$. For uniqueness, consider
$x_A \parallel x_A \in \wit_{\cc_A}^+(\x_A, \x_A)$. By definition of
$\sim^+$-witnesses, there are symmetries
\[
\theta_A^- : x_A \sym_A^- \rep{\x}_A\,,
\qquad
\qquad
\theta_A^+ : x_A \sym_A^+ \rep{\x}_A\,,
\]
but because $\rep{\x}_A$ is canonical this entails that $x_A =
\rep{\x}_A$.
\end{proof}

Next we prove that $\coll(-)$ preserves $\sim$. As in the target
category the equivalence relation $\sim$ is the identity, this amounts
to $\coll(-)$ being invariant under $\simstrat$.

\begin{prop}\label{prop:coll_pres_sim}
Consider $A, B$ arenas, and $\sigma, \tau \in \Strat(A, B)$ such that
$\sigma \simstrat \tau$.

Then, for all $\x_A \in \wconf{A}$ and $\x_B \in \wconf{B}$,
$\coll(\sigma)_{\x_A, \x_B} = \coll(\tau)_{\x_A, \x_B}$.
\end{prop}
\begin{proof}
By Definition \ref{def:pos_iso}, there is a positive isomorphism
$\varphi : \sigma \simstrat \tau$. Recall that this means
\[
\{(\pr_\sigma(s), \pr_\tau\circ \varphi(s)) \mid s \in x\} \in \ptilde{A\vdash B}
\]
for all $x \in \conf{\sigma}$, with $\varphi$ an isomorphism of ess --
we write $\psi_A^x \parallel \psi_B^x$ for this symmetry, satisfying
by construction $\psi_A^x : x^\sigma_A \sym_A^- y^\tau_A$ and
$\psi_B^x : x^\sigma_B \sym_B^+ y^\tau_B$ writing $y^\tau = \varphi
x^\sigma\in \confp{\tau}$.

Now, for $\x_A \in \wconf{A}$ and $\x_B \in \wconf{B}$, we construct a
bijection
\[
\begin{array}{rcrcl}
\varphi &:& \wit^+_\sigma(\x_A, \x_B) &\bij& \wit^+_\tau(\x_A, \x_B)\\
&& x^\sigma &\mapsto& \varphi(x^\sigma)\,.
\end{array}
\]

Indeed, consider $x^\sigma \in \wit^+_\sigma(\x_A, \x_B)$. By
definition, there are $\theta_A^- : x^\sigma_A \sym_A^- \rep{\x}_A$ and
$\theta_B^+ : x^\sigma_B \sym_B^+ \rep{\x}_B$. Now, $\varphi(x^\sigma)
\in \confp{\tau}$ as $\varphi$ is an order-isomorphism preserving
polarities. Furthermore, we have $\theta_A^- \circ (\psi_A^x)^{-1} :
y^\tau_A \sym_A^- \rep{\x}_A$ and $\theta_B^+ \circ (\psi_B^x) :
y^\tau_B \sym_B^+ \rep{\x}_B$; which entails $y^\tau \in
\wit^+_\tau(\x_A, \x_B)$ as required. By the symmetrical reasoning
$\varphi^{-1}$ sends $\wit^+_\tau(\x_A, \x_B)$ to $\wit^+_\sigma(\x_A,
\x_B)$ and they are clearly mutual inverses, which concludes the proof.
\end{proof}

\subsubsection{Preservation of monoidal
structure}\label{subsubsec:pres_mon_str}
Next, $\coll(-)$ is a symmetric monoidal $\sim$-functor.

For $A, B$ any arenas, we provide the components:
\[
\begin{array}{rcrcl}
t^\tensor_{A, B} &:& \coll(A) \times \coll(B)
&\stackrel{\Rel{\N}}{\longto}& \coll(A\tensor B)\\
t^1 &:& 1 & \stackrel{\Rel{\N}}{\longto} &\coll(1)
\end{array}
\]
defined by $(t^\tensor_{A, B})_{(\x_A, \x_B), \y} = \delta_{\y, \x_A
\parallel \x_B}$ for every $\x_A \in \wconf{A}$ and $\x_B \in
\wconf{B}$; and $(t^1)_{\bullet, \emptyset} = 1$.

\begin{prop}\label{prop:smsf}
We have $(\coll(-), t^\tensor, t^1)$ a symmetric monoidal
$\sim$-functor.
\end{prop}
\begin{proof}
The crux is the naturality of $t^\tensor$, corresponding to the fact
that the tensor operation on morphisms for $\Strat$ and $\Rel{\N}$
agree, \emph{i.e.} the following diagram commutes in $\Rel{\N}$
\[
\xymatrix{
\coll(A) \times \coll(B) 
	\ar[r]^{t^\tensor_{A, B}}
	\ar[d]_{\coll(\sigma) \tensor \coll(\tau)}&
\coll(A\tensor B)
	\ar[d]^{\coll(\sigma \tensor \tau)}\\
\coll(A') \times \coll(B')
	\ar[r]_{t^\tensor_{A', B'}}&
\coll(A' \tensor B')
}
\]
for all $A, B, A', B'$ arenas, and $\sigma \in \Strat(A, A'), \tau \in
\Strat(B, B')$. To prove this, we invoke the characterizing property of
the tensor of strategies in Proposition \ref{prop:def_tensor} -- we have
\[  
\begin{array}{rcrcl}
(- \tensor -) &:& \confp{\sigma} \times \confp{\tau} &\simeq&
\confp{\sigma \tensor \tau}\\
\end{array}
\]
such that $\pr_{\sigma\tensor \tau}(x^\sigma \tensor x^\tau) =
(x^\sigma_A \parallel x^\tau_B) \parallel (x^\sigma_{A'} \parallel
x^\tau_{B'})$ for all $x^\sigma \in \confp{\sigma}$ and $x^\tau \in
\confp{\tau}$. 
For all $\x_A \in \wconf{A}, \x_B \in \wconf{B}, \y_{A'}
\in \wconf{A'}$ and $\y_{B'} \in \wconf{B'}$, this immediately restricts to
\[
(- \tensor -) : \wit^+_\sigma(\x_A, \y_{A'}) \times \wit^+_\tau(\x_B,
\y_{B'}) \to \wit^+_{\sigma \tensor \tau}(\x_A \parallel \x_{B}, \y_{A'}
\parallel \y_{B'})\,.
\]

Via this bijection, both paths around the diagram compute to the quantity:
\[
\sharp (\wit^+_\sigma(\x_A, \y_{A'}) \times \wit^+_\tau(\x_B, \y_{B'}))
\]
for all $(\x_A, \x_B) \in \coll(A) \times \coll(B)$ and $\y_{A'}
\parallel \y_{B'} \in \coll(A\tensor B)$; as required.

The further coherence conditions, expressing that the associators,
unitors and symmetries agree in both categories, are all immediate
verifications relying on the characterization of the $+$-covered
configurations of the corresponding strategies.
\end{proof}

\subsection{A relative Seely $\sim$-functor} \changed{Next we study the
preservation of the modality $\oc(-)$, which is the most challenging.
Then we will deal with $\lin$ and $\with$.} 

\subsubsection{Preservation of the action of $\oc$ on morphisms.} 
\changed{
Inspecting the requirements for relative Seely functors
(Definition~\ref{def:relseelyfunctors}) we must first show that
$\coll(-)$ preserves strict objects; this is immediate since every
object is strict in $\Rel{\N}$. We must then exhibit $t^\oc_C : \oc
\coll(C) \to \coll(\oc C)$ for every strict $C$, and show commutation
of 
the diagram
\eqref{eq:preservationpromotion} up to $\sim$. However,
both in $\Strat$ and $\Rel{\R}$, the relative comonad $\oc$ is in fact
a proper comonad; this means that we have a concrete presentation of
promotion: for every $f : \oc C \to D$, $f^\dagger = \oc f \circ
\delta_C$. Thus the diagram \eqref{eq:preservationpromotion} amounts to
the following, for every $f : \oc C \to D$ with $C, D$ strict: 
\begin{equation}
\label{eq:preservationpromotion2}
\scalebox{.8}{$
\xymatrix@C=10pt{
\oc F C
	\ar[rr]^{\delta_C}
	\ar[dr]_{t^\oc_C}&&
\oc \oc FC
	\ar[rr]^{\oc t^\oc_C}&&
\oc F \oc C
	\ar[rr]^{\oc F f}&&
\oc F D
	\ar[dl]^{t^\oc_D}\\
&F\oc C	\ar[rr]_{F\delta_C}&&
F\oc \oc C
	\ar[rr]_{F\oc f}&&
F\oc D
}
$}
\end{equation}
}

\changed{For this, we will need to understand how the functorial action of $\oc$ in $\Strat$ relates to that of $\oc$ in $\Rel{\N}$. We study this now, before giving the definition of the maps $t^\oc_C$.}

The comparison is subtle and it seems a good idea to first recall the definition in $\Rel{\N}$: 
\begin{eqnarray}
(\oc \alpha)_{\mu, [y_1, \dots, y_n]} &=& 
\sum_{\substack{(x_1, \dots, x_n)\,\text{s.t.}\\ \mu = [x_1, \dots,
x_n]}} \prod_{1\leq i \leq n}
\alpha_{x_i, y_i}\label{eq5.1}
\end{eqnarray}
for any weighted relation $\alpha : X \relto Y$.
Something tricky is going on here. It looks like we are
summing over all permutations of $\{1, \dots, n\}$, but no:
permutations that lead to the same tuple are counted only once (and the
rest of the term is invariant under permutations yielding the same
tuple). We must understand how this arises in game
semantics.

We recall the game semantical definition that we must match
against \eqref{eq5.1}. Consider $A, B$ arenas, $\sigma \in \Strat(A,
B)$, $\x_A \in \wconf{\oc A}$ and $\y_B \in \wconf{\oc B}$, respectively with
\[
\rep{\x}_{\oc A} =\,\parallel_{1\leq i \leq p}^{\neq \emptyset}
\rep{\x}_A^i\,,
\qquad
\qquad
\rep{\y}_{\oc B} =\, \parallel_{1\leq i \leq n}^{\neq \emptyset}
\rep{\y}_B^i
\]
where, and from now on, we label these parallel compositions with ``$\neq
\emptyset$'' to emphasize that each component is non-empty. By
definition, $\coll(\oc \sigma)_{\x_{\oc A}, \y_{\oc B}} = \sharp
\wit_{\oc \sigma}^+(\x_{\oc A}, \y_{\oc B})$, where we have
\[
\wit^+_{\oc \sigma}(\x_{\oc A}, \y_{\oc B})
\bij 
\sum_{x \in \Sym_A^-(\rep{\x}_{\oc A})} 
\sum_{y \in \Sym_B^+(\rep{\y}_{\oc B})}
\wit_{\oc \sigma}(x, y)
\]
with $\Sym_A^-(\rep{\x}_{\oc A}) = \{x \in
\conf{\oc A} \mid x \sym_{\oc A}^- \rep{\x}_{\oc A}\}$ and
$\Sym_B^+(\rep{\y}_{\oc B}) = \{y \in \conf{\oc B} \mid y \sym_{\oc B}^+
\rep{\y}_{\oc B}\}$.

Our task is to link this sum to \eqref{eq5.1}, which will require us to
gradually decompose further its elements. First the sum over all $y \in
\Sym_{\oc B}^+(\rep{\y}_{\oc B})$ may be described very simply:

\begin{lem}
There is a bijection:
\[
\begin{array}{rcrcl}
\Sym_{\oc B}^+(\rep{\y}_{\oc B}) &\bij& \prod_{1 \leq i \leq n}
\Sym_B^+(\rep{\y}_B^i)\\
\parallel_{1\leq i \leq n}^{\neq \emptyset} y_B^i &\mapsto& (y_B^1, \dots, y_B^n)
\end{array}
\]
\end{lem}
\begin{proof}
Obvious by definition of positive symmetries of $\oc B$ in Definition
\ref{def:bang}.
\end{proof}

In contrast, the set $\Sym_{\oc A}^-(\rep{\x}_{\oc A})$ is much wilder, as
negative symmetries on $A$ are free to change copy indices at will. Yet,
the data of some $x \sym_{\oc A}^- \rep{\x}_{\oc A}$ may be witnessed by
distinct symmetries; in fact even the action of the symmetry on copy
indices is not uniquely defined.

To help reason on $\Sym_{\oc A}^-(\rep{\x}_{\oc A})$ we need more
structure. For $x \in \Sym_{\oc A}^-(\rep{\x}_{\oc A})$, we write
\[
x =\,\parallel_{k \in K_x}^{\neq \emptyset} x_A^k
\]
where $x_A^k \in \conf{A}$ for $k \in K_x$. We choose, for each
$x \in \Sym_{\oc A}^-(\rep{\x}_{\oc A})$ a bijection $\pi_x : K_x \bij
\{1, \dots, p\}$ such that for all $k \in K_x$, $x_A^k \sym_A^-
\rep{\x}_A^{\pi_x(k)}$. If $\pi$ is a permutation on $\{1, \dots, p\}$,
we say it is an \textbf{isotropy} of $\x_{\oc A}$ if for all $1\leq i
\leq p$ we have $\rep{\x}_A^i \sym_A \rep{\x}_A^{\pi(i)}$, \emph{i.e.}
$\rep{\x}_A^i = \rep{\x}_A^{\pi(i)}$. Isotropies of
$\x_{\oc A}$ form a group $\m(\x_{\oc A})$, the \textbf{isotropy group} of
$\x_{\oc A}$. Now, we prove:

\begin{lem}\label{lem:isotrop1}
We have the following bijection:
\[
\begin{array}{rcl}
\m(\x_{\oc A}) \times \Sym_{\oc A}^-(\rep{\x}_{\oc A})
&\bij&
\sum_{K \subseteq_f \mathbb{N}}
\sum_{\pi : K \bij \{1,\dots,p\}} 
\prod_{k\in K}
\Sym_A^-(\rep{\x}_A^{\pi(k)})\\
(\vartheta, x) &\mapsto& (K_x, \vartheta\circ \pi_x, (x_A^k)_{k \in
K_x})
\end{array}
\]
\end{lem}
\begin{proof}
We first check that this map is well-defined. Consider $\vartheta \in
\m(\x_{\oc A})$ and $x \sym_A^- \rep{\x}_{\oc A}$. We must show that for
all $k \in K_x$, we have $x_A^k \sym_A^- \rep{\x}_A^{\vartheta\circ
\pi_x(k)}$. We know
that $x_A^k \sym_A^- \rep{\x}_A^{\pi_x(k)}$. Moreover, by definition of
$\m(\x_{\oc A})$, we have $\rep{\x}_A^{\pi_x(k)} =
\rep{\x}_A^{\vartheta \circ \pi_x(k)}$; so $x_A^k \sym_A^-
\rep{\x}_A^{\vartheta\circ \pi_x(k)}$.

We define its inverse. To $K \subseteq_f \mathbb{N}$, $\pi : K
\bij \{1, \dots, p\}$, and $(x_A^k)_{k\in K}$, we associate
\[
(\pi \circ \pi_x^{-1}, x) \in \m(\x_{\oc A}) \times \Sym_{\oc
A}^-(\rep{\x}_{\oc A})
\]
where $x =\,\parallel_{k\in K}^{\neq \emptyset} x_A^k \sym_A^-
\rep{\x}_{\oc A}$ as required. It is clear that the two are inverses.
\end{proof}

Relying on this bijection, we may start the following computation:
\begin{eqnarray*}
\m(\x_{\oc A}) \times \wit_{\oc \sigma}^+(\x_{\oc A}, \y_{\oc
B})
&\bij& \sum_{\vartheta \in \m(\x_{\oc A})} 
\sum_{\left(\parallel_{k \in K}^{\neq \emptyset} x_A^k\right) \in \Sym_{\oc
A}^-(\rep{\x}_{\oc A})}
\sum_{y \in \Sym_{\oc B}^+} \wit_{\oc \sigma}(\parallel_{k \in K}^{\neq
\emptyset} x_A^k, y)\\
&\bij&
\sum_{K\subseteq_f \mathbb{N}}
\sum_{\pi : K \bij \{1, \dots, p\}}
\sum_{(x_A^k)_{k\in K}}
\sum_{(y_B^i)_{1\leq i \leq n}}
\wit_{\oc \sigma}\left(\parallel_{k\in K}^{\neq \emptyset} x_A^k, 
\parallel_{1\leq i \leq n}^{\neq \emptyset} y_B^i\right)
\end{eqnarray*}
where $(x_A^k)_{k\in K}$ ranges over $\Pi_{k\in K}
\Sym_A^-(\rep{\x}_A^{\pi(k)})$ and $(y_B^i)_{1\leq
i \leq n}$ over $\Pi_{1\leq i \leq n}
\Sym_B^+(\rep{\y}_B^i)$.

Now, let us recall that Proposition \ref{prop:char_pcov_bang}
gives us an order-iso
\[
\,[-] : \mathsf{Fam}\left(\confpn{\sigma}\right)  \bij
\confp{\oc \sigma}
\]
with $\mathsf{Fam}(X)$ the set of families of elements of $X$ indexed by
finite subsets of $\mathbb{N}$, such that
\[
\pr_{\oc \sigma}\left(\left[(x^i)_{i\in I}\right]\right)
=
(\parallel_{i\in I} x^i_A) \parallel (\parallel_{i \in I} x^i_B)
\]
where for all $i\in I$, $\pr_\sigma(x^i) = x^i_A \parallel x^i_B$. In
particular, this entails that the set above can be non-empty only if $K
\subseteq \{1, \dots, n\}$; so it is in bijection with
\begin{eqnarray}
&\bij&
\sum_{K\subseteq \{1, \dots, n\}}
\sum_{\pi : K \bij \{1, \dots, p\}}
\sum_{(x_A^k)_{k\in K}}
\sum_{(y_B^i)_{1\leq i \leq n}}
\wit_{\oc \sigma}\left(\parallel_{k\in K}^{\neq \emptyset} x_A^k, 
\parallel_{1\leq i \leq n}^{\neq \emptyset} y_B^i\right)
\label{eq5.2}
\end{eqnarray}

To simplify the sum further we shall need the next lemma. It is a
variant of Lemma \ref{lem:isotrop1}, but also dealing with the fact that
we might have fewer non-empty configurations on $A$ than on $B$, and
introducing a sum over sequences of symmetry classes akin to
\eqref{eq5.1}.

\begin{lem}
We have the following bijection:
\[
\sum_{K\subseteq \{1, \dots, n\}}
\sum_{\pi : K \bij \{1, \dots, p\}}
\prod_{k \in K} \Sym_A^-(\rep{\x}_A^{\pi(k)})
\bij
\sum_{\stackrel{(\z_A^1, \dots, \z_A^n)\,\text{s.t.}}{\rep{\x}_{\oc A}
\sym_A\,\parallel_{1\leq i \leq n} \rep{\z}_A^i}} 
\sum_{\vartheta\in \m(\x_{\oc A})}
\prod_{1\leq i \leq n}
\Sym_A^-(\rep{\z}_A^i)
\]
\end{lem}
\begin{proof}
As for \eqref{eq5.1}, the sum on the right hand side ranges over all
tuples. Fix in advance, for all
$\overrightarrow{z} = (\z_A^1, \dots, \z_A^n)$ such that $\rep{\x}_{\oc
A} \sym_A\,\parallel_{1\leq i \leq n} \rep{\z}_A^i$, an injection
$\kappa_{\overrightarrow{z}} : \{1, \dots, p\} \to \{1, \dots, n\}$
s.t. for all $1\leq i \leq p$, $\rep{\x}_A^i \sym_A
\rep{\z}_A^{\kappa_{\overrightarrow{z}}(i)}$. Necessarily, $\z_A^i$ is
empty for all $i \not \in \cod(\kappa_{\overrightarrow{\z}})$.

Given $K \subseteq \{1, \dots, n\}, \pi : K \bij \{1, \dots, p\}$ and 
$(x_A^k)_{k\in K}$, we set $\overrightarrow{z} = (\z_A^1,
\dots, \z_A^n)$ with
\[
\z_A^i = 
\left\{
\begin{array}{ll}
\x_A^{\pi(i)}	&\text{if $i \in K$}\\
\emptyset	&\text{otherwise}
\end{array}\right.\,,
\]
by construction we have $\rep{\x}_{\oc A} \sym_A\,\parallel_{1\leq i \leq n}
\rep{\z}_A^i$. We set $\vartheta = \pi \circ
\kappa_{\overrightarrow{z}} \in \m(\x_{\oc A})$. Finally, we set
\[
x_A^i = 
\left\{
\begin{array}{ll}
x_A^{i}   	&\text{if $i \in K$}\\
\emptyset       &\text{otherwise}
\end{array}\right.\,;
\]
if $i \not \in K$, $x_A^i = \rep{\z}_A^i = \emptyset$;
if $i \in K$, $x_A^i \sym_A^- \rep{\x}_A^{\pi(i)}$ by hypothesis
and $\rep{\z}_A^i = \rep{\x}_A^{\pi(i)}$ by construction.

Reciprocally, consider $\overrightarrow{\z} = (\z_A^1, \dots, \z_A^n)$,
$\vartheta \in \m(\x_{\oc A})$ and $(x_A^i)_{1\leq i \leq n}$. We set
$K$ as the subset of all $k \in \{1, \dots, n\}$ such that $\z_A^k$ is
non-empty. We set the bijection
\[
\begin{array}{rcrcl}
\pi &:& K &\bij & \{1, \dots, p\}\\
&& k &\mapsto& \vartheta \circ \kappa_{\overrightarrow{\z}}^{-1}
\end{array}
\]
which is well-defined as $K$ is exactly the codomain of
$\kappa_{\overrightarrow{z}}$. For every $k \in K$, we set
$(x_A^k)_{k\in K}$ simply as the restriction of the family
$(x_A^i)_{1\leq i \leq n}$ to $K$ -- and we do indeed have
\[
x_K^k \sym_A^- \rep{\z}_A^k \sym_A^-
\rep{\x}_A^{\kappa_{\overrightarrow{\z}}^{-1}(k)}
\sym_A^- \rep{\x}_A^{\vartheta \circ \kappa_{\overrightarrow{\z}}^{-1}(k)} =
\rep{\x}_A^{\pi(k)}\,.
\]

Finally, it is a direct verification that these constructions are
inverses.
\end{proof}

We start again computing from \eqref{eq5.2}. Substituting the
bijection of the lemma above:
\begin{eqnarray*}
&\bij&
\sum_{K\subseteq \{1, \dots, n\}}
\sum_{\pi : K \bij \{1, \dots, p\}}
\sum_{(x_A^k)_{k\in K}}
\sum_{(y_B^i)_{1\leq i \leq n}}
\wit_{\oc \sigma}\left(\parallel_{k\in K}^{\neq \emptyset} x_A^k, 
\parallel_{1\leq i \leq n}^{\neq \emptyset} y_B^i\right)\\
&\bij& 
\sum_{\vartheta \in \m(\x_{\oc A})}
\sum_{\stackrel{(\z_A^1, \dots, \z_A^n)\,\text{s.t.}}{\rep{\x}_{\oc A}
\sym_A\,\parallel_{1\leq i \leq n} \rep{\z}_A^i}}
\sum_{(x_A^i)_{1\leq i \leq n}}
\sum_{(y_B^i)_{1\leq i \leq n}} 
\wit_{\oc \sigma}\left(\parallel_{1\leq i \leq n} x_A^i,
\parallel_{1\leq i \leq n}^{\neq \emptyset} y_B^i\right)
\end{eqnarray*}
where now $(x_A^i)_{1\leq i \leq n}$ ranges over $\prod_{1\leq i \leq n}
\Sym_A^-(\rep{\z}_A^i)$ and $(y_B^i)_{1\leq i \leq n}$ over
$\prod_{1\leq i \leq n} \Sym_B^+(\rep{\y}_B^i)$. 

Some of the $x_A^i$ may now be empty, but both parallel compositions range over the
same indices. Thanks to this we may apply Proposition
\ref{prop:char_pcov_bang}, which directly yields:
\[
\bij \sum_{\vartheta \in \m(\x_{\oc A})}
\sum_{\stackrel{(\z_A^1, \dots, \z_A^n)\,\text{s.t.}}{\rep{\x}_{\oc A}
\sym_A\,\parallel_{1\leq i \leq n} \rep{\z}_A^i}}
\sum_{(x_A^i)_{1\leq i \leq n}}
\sum_{(y_B^i)_{1\leq i \leq n}} 
\prod_{1\leq i \leq n}
\wit_{\sigma}(x_A^i,y_B^i)
\]

We may now complete the computation, with:
\begin{eqnarray*}
&\bij& \sum_{\vartheta \in \m(\x_{\oc A})}
\sum_{\stackrel{(\z_A^1, \dots, \z_A^n)\,\text{s.t.}}{\rep{\x}_{\oc A}
\sym_A\,\parallel_{1\leq i \leq n} \rep{\z}_A^i}}
\prod_{1\leq i \leq n}
\sum_{x_A^i \sym_A^- \rep{\z}_A^i}
\sum_{y_B^i \sym_B^+ \rep{\x}_B^i}
\wit_{\sigma}(x_A^i,y_B^i)\\
&\bij& \sum_{\vartheta \in \m(\x_{\oc A})}
\sum_{\stackrel{(\z_A^1, \dots, \z_A^n)\,\text{s.t.}}{\rep{\x}_{\oc A}
\sym_A\,\parallel_{1\leq i \leq n} \rep{\z}_A^i}}
\prod_{1\leq i \leq n}
\wit^+_\sigma(\z_A^i, \y_B^i)\,,
\end{eqnarray*}
which concludes the construction of the following bijection:

\begin{lem}\label{lem:bij_bang}
For $A, B$ arenas, strategy $\sigma \in \Strat(A, B)$, and symmetry
classes $\x_{\oc A} \in \wconf{\oc A}$, $\y_{\oc B} \in \wconf{\oc B}$
with $\y_{\oc B} = [\y_B^1, \dots, \y_B^n]$ with each $\y_B^i$
non-empty, we have a bijection
\[
U : \sum_{\pi \in \m(\x_{\oc A})} \wit^+_{\oc \sigma}(\x_{\oc A}, \y_{\oc B})
\bij
\sum_{\varpi \in \m(\x_{\oc A})} 
\sum_{\stackrel{(\z_A^1, \dots,
\z_A^n)\,\text{s.t.}}{
\x_{\oc A} = [\z_A^i \mid \z_A^i \neq \emptyset]}}
\prod_{1\leq i \leq n}
\wit^+_\sigma(\z_A^i, \y_B^i)
\]
such that for all
$K(\pi, x^{\oc \sigma}) = (\varpi, ((\z_A^1, \dots, \z_A^n),
(x^i)_{1\leq i \leq n}))$, we have
$x^{\oc \sigma} = [(x^i)_{1\leq i \leq n}]$.
\end{lem}
\begin{proof}
Note $\rep{\x}_{\oc A} \sym_A\,\parallel_{1\leq i
\leq n} \rep{\z}_A^i$ iff $\x_{\oc A} = [\z_A^i \mid \z_A^i \neq
\emptyset]$ by Lemma \ref{lem:r_bang} -- complete
symmetry classes of $\oc A$ match finite multisets of
\emph{non-empty} complete symmetry classes of $A$. 
\end{proof}

We may finally deduce the desired equality: 

\begin{cor}\label{cor:main_pres_bang}
For $A, B$ arenas, strategy $\sigma \in \Strat(A, B)$, and symmetry
classes $\x_{\oc A} \in \wconf{\oc A}$, $\y_{\oc B} \in \wconf{\oc B}$
with $\y_{\oc B} = [\y_B^1, \dots, \y_B^n]$ and each $\y_B^i$
non-empty, we have:
\begin{eqnarray}
\sharp \wit^+_{\oc \sigma}(\x_{\oc A}, \y_{\oc B})
&=&
\sum_{\stackrel{(\z_A^1, \dots, \z_A^n)\,\text{s.t.}}{
\x_{\oc A} = [\z_A^i \mid \z_A^i \neq \emptyset]}}
\prod_{1\leq i \leq n}
\sharp \wit^+_\sigma(\z_A^i, \y_B^i)\label{eq5.3}
\end{eqnarray}
\end{cor}
\begin{proof}
By Lemma \ref{lem:bij_bang}, taking the cardinalities we have the
equality:
\[
\sharp \m(\x_{\oc A}) \times \sharp \wit^+_{\oc \sigma}(\x_{\oc A}, \y_{\oc B})
=
\sharp \m(\x_{\oc A}) \times
\sum_{\stackrel{(\z_A^1, \dots,
\z_A^n)\,\text{s.t.}}
{\x_{\oc A} = [\z_A^i \mid \z_A^i \neq \emptyset]}}
\prod_{1\leq i \leq n}
\sharp \wit^+_\sigma(\z_A^i, \y_B^i)
\]
from which the result follows by dividing by $\sharp \m(\x_{\oc A})$
(which we can do as it is finite).
\end{proof}

As for composition, we must pad the desired
identity with further symmetry groups in order to realize it. 
The equation \eqref{eq5.3} is very much like \eqref{eq5.2}, and
\changed{we will use this result to show that $\coll(-)$ has the
appropriate preservation properties for $\oc$. First we explain why
only the relative comonad structure is preserved, and not the full
comonad structure.} 

\subsubsection{\changed{Non-preservation of the comonad $\oc$.}}
\label{subsubsec:lastglitch}
For $C$ a strict arena, we define 
\[
\begin{array}{rcrcl}
t^\oc_C &:& \oc \coll(C) &\to& \coll(\oc C)
\end{array}
\]
where $(t^\oc_C)_{\mu, \x_{\oc C}} = \delta_{\mu, s^\oc_C(\x_{\oc
C})}$, via the bijection of Lemma \ref{lem:r_bang}.
For strict $C, D$, we deduce from Corollary
\ref{cor:main_pres_bang} that for $\sigma \in \Strat(C, D)$, \changed{the following diagram commutes} 
in $\Rel{\N}$:
\[
\xymatrix@R=10pt{
\oc \coll(C)
        \ar[r]^{t^\oc_C}
        \ar[d]_{\oc \coll(\sigma)}&
\coll(\oc C)
        \ar[d]^{\coll(\oc \sigma)}\\
\oc \coll(D)
        \ar[r]_{t^\oc_D}&
\coll(\oc D)
}       
\]

\changed{
So $t^\oc$ is a natural transformation $\oc \circ \coll \to \coll \circ \oc : \Strat_s \to \Rel{\N}$.}
However, having this for strict $C$ and $D$ is not sufficient,
\changed{because the construction of the Kleisli category relies
crucially on promotion. We must therefore
consider strategies of the form $\sigma : \oc C \to D$, where of course
here $\oc C$ is not strict and the property above does not directly
apply. 
}

The issue is that there
\emph{is} a difference between $\oc \coll(A) = \mathcal{M}_f(\coll(A))$ and $\coll(\oc
A)$ for $A$ non-strict: the latter has only one empty
configuration, whereas the former distinguishes between elements $[\emptyset, \dots,
\emptyset]$ containing $n$ occurrences of $\emptyset$, for every $n$.

Consequently, the naturality square above fails for any reasonable
extension of $t^\oc_A$ to non-strict $A$. For instance, considering
$\sigma  \in \Strat(1, \gbool)$ that immediately answers $\ttrue$, 
\[
(\oc \coll(\sigma))_{\emptyset^n, \ttrue^p} = \delta_{n,p}
\]
where $\emptyset^n, \ttrue^p$ are the obvious multisets. In other words,
the relational model remembers how many times $\sigma$ ``does not call''
its argument. In contrast, we have
$\coll(\oc \sigma)_{\emptyset, \ttrue^p} = 1$
for all $p \in \mathbb{N}$ -- $\wconf{\oc 1}$ is a singleton set.
Fortunately, this mismatch disappears for promotion.

\subsubsection{Preservation of promotion} \changed{We verify the necessary diagram \eqref{eq:preservationpromotion2}.
}
\begin{prop}\label{prop:pres_prom}
Consider $C, D$ strict arenas, and $\sigma \in \Strat(\oc C, D)$.

Then, \emph{promotion is preserved}, \emph{i.e.} the following diagram
commutes in $\Rel{\N}$:
\[
\scalebox{.8}{$
\xymatrix@C=5pt@R=5pt{
\mathcal{M}_f(\wconf{C})
        \ar[rr]^{\delta_C}
        \ar[dr]_{t^\oc_C}&&
\mathcal{M}_f(\mathcal{M}_f(\wconf{C}))
        \ar[rr]^{\oc t^\oc_C}&&
\mathcal{M}_f(\wconf{\oc C})
        \ar[rr]^{\oc \coll(\sigma)}&&
\mathcal{M}_f(\wconf{D})
        \ar[dl]^{t^\oc_D}\\
&\wconf{\oc C} \ar[rr]_{\coll(\delta_C)}&&
\wconf{\oc \oc C}
        \ar[rr]_{\coll(\oc \sigma)}&&
\wconf{\oc D}
}
$}
\]
\end{prop}
\begin{proof}
For $\mu \in \mathcal{M}_f(\wconf{C})$ and $\y_{\oc D} \in \wconf{\oc
D}$, the upper-right path evaluates to:
\begin{eqnarray}
\sum_{\stackrel{(\x_{\oc C}^1, \dots, \x_{\oc C}^n)\,\text{s.t.}}
{\x_{\oc C}^1 + \dots + \x_{\oc C}^n = \mu}}
\prod_{1\leq i \leq n} \coll(\sigma)_{\x_{\oc C}^i, \y_D^i}\label{sum1}
\end{eqnarray}
writing $\y_{\oc D} = [\y_D^1, \dots, \y_D^n]$, inlining $s^\oc_{D}$ and
$s^\oc_C$. This is by \eqref{eq5.1} and direct computation.  

For the other path, first note that for any $\x_{\oc C} \in
\wconf{\oc C}$ and $\y_{\oc \oc C} \in \wconf{\oc \oc C}$ we have
\[
\coll(\delta_C)_{\x_{\oc C}, \y_{\oc \oc C}} = \delta_{\x_{\oc C},
\y_{\oc C}^1 + \dots + \y_{\oc C}^n}
\]
where $\y_{\oc \oc C} = [\y_{\oc C}^1, \dots, \y_{\oc C}^n]$ with each
$\y_{\oc C}^i$ non-empty -- this is proved by a direct elaboration of
Proposition \ref{prop:cc_collapse}. Relying on this and Corollary
\ref{cor:main_pres_bang}, the bottom-left path evaluates to
\[
\sum_{\stackrel{(\x_{\oc C}^1, \dots, \x_{\oc C}^n)\text{s.t.}}
{\sum [\x_{\oc C}^i \mid \x_{\oc C}^i \neq \emptyset] = \mu}}
\prod_{1\leq i \leq n}
\coll(\sigma)_{\x_{\oc C}^i, \y_D^i}
\]
which is almost \eqref{sum1}, except for the side-condition.
But fortunately,
$\sum [\x_{\oc C}^i \mid 1\leq i \leq n] = 
\sum [\x_{\oc C}^i \mid \x_{\oc C}^i \neq \emptyset]$
as the empty symmetry class corresponds to the empty multiset.
\end{proof}

We see that the mismatch causing the failure of naturality disappears
with the promotion, as the junk enumeration of multisets in the
relational model is erased by the sum.

\subsubsection{A relative Seely $\sim$-functor}
\label{subsubsec:rigid_seely}
To wrap up, we introduce the missing components
\[
\begin{array}{rcrcl}
t^\top &:& \emptyset &\to& \coll(\top)\\
t^\with_{C, D} &:& \coll(C) + \coll(D) &\to& \coll(C\with D)\\
t^\lin_{A,C} &:& \coll(A) \times \coll(C) &\to& \coll(A\lin C)
\end{array}
\]
for $A, B, C, D$ arenas with $C, D$ strict, defined by $t^\top$ with 
empty domain, and
\[
(t^\with_{C, D})_{x, \x} = \delta_{x, s_{C, D}^\with(\x)}
\qquad
\qquad
(t^\lin_{A, C})_{x, \x} = \delta_{x, s_{A, C}^\lin(\x)}\,.
\]

The missing \changed{five} coherence diagrams of Figure \ref{fig:seely-functors}
are direct, from an analysis of the symmetry classes reached by the
component strategies. As for copycat in Proposition
\ref{prop:cc_collapse}, this follows from the description of the
$+$-covered configurations of projections in Section
\ref{subsubsec:cart_id}, dereliction and monoidality in Section
\ref{subsubsec:exp_strat}, and evaluation in Section
\ref{subsubsec:semiclosed_strat}. 
Altogether:

\begin{cor}\label{cor:seely_trans}
We have a \changed{relative} Seely $\sim$-functor $\coll(-) : \Strat \to \Rel{\N}$.
\end{cor}

So by Proposition \ref{prop:seely_functor_kleisli} we 
have a cartesian closed $\sim$-functor
$\coll_\oc(-) : \Strat_\oc \to \Rel{\N}_\oc$.

\subsection{Preservation of the Interpretation} The above covers the
simply-typed $\lambda$-calculus;
it remains to address constants and primitives, and recursion.

\subsubsection{Mediating isomorphisms}
By Proposition \ref{prop:seely_functor_kleisli}, we
have isos in $\Rel{\N}_\oc$ for $A, B$ strict
\[
\begin{array}{rclcl}
k^\top &:& \top &\to & \coll(\top)\\
k^\with_{A, B} &:& \coll(A)\with \coll(B) &\to& \coll(A\with B)\\
k^\tto_{A, B} &:& \oc \coll(A) \lin \coll(B) &\to& \coll(A\tto B)
\end{array}
\]
with $A \tto B = \oc A \lin B$. To these, for ground $\tx$ we add $k^\tx
: \rintr{\tx} \to \coll(\intr{\tx})$ defined as $t^\tx \circ
\epsilon_{\rintr{\tx}}$, with $t^\tx \in \Rel{\N}(\rintr{\tx},
\coll(\intr{\tx}))$ set as $(t^\tx)_{x, \x} = \delta_{x, s^\tx(\x)}$
with $s^\tx$ from Lemma \ref{lem:r_basic}.

We generalize these mediating isos to all types, by defining
isomorphisms $k^\Ctx_\Gamma : \rintr{\Gamma} \to \coll(\intr{\Gamma})$
and $k^\Ty_A : \rintr{A} \to \coll(\intr{A})$ in $\Rel{\N}_\oc$
inductively, as follows:
\[
\begin{array}{rclcrcl}
k^{\Ctx}_{[]} &=& k^\top &:& \top &\to& \coll(\top)\\
k^{\Ctx}_{\Gamma, x:A} &=& k^\with_{\intr{\Gamma}, \intr{A}} \odot_\oc
(k^\Ctx_\Gamma \with_\oc k^\Ty_A) &:& \rintr{\Gamma, x:A} &\to&
\coll(\intr{\Gamma, x:A})\\
k^\Ty_{A\to B} &=& k^\tto_{\intr{A}, \intr{B}} \odot_\oc ((k^\Ty_A)^{-1}
\tto k^\Ty_B) &:& \rintr{A\to B} &\to& \coll(\intr{A\to B})
\end{array}
\]
where $\with_\oc$ is the functorial action of the cartesian product in
$\Rel{\N}_\oc$. 

These isos may be described more directly in the linear category
$\Rel{\N}$. First we set:
\[
\begin{array}{rclcrcl}
t^\Ctx_{[]} &=& t^\top &:& \top &\to& \coll(\top)\\
t^\Ctx_{\Gamma, x:A} &=& t^\with_{\Gamma, A} \odot (t^\Ctx_\Gamma \with
t^\Ty_A) &:& \rintr{\Gamma, x : A} &\to& \coll(\intr{\Gamma, x : A})\\
t^{\oc \Ctx}_\Gamma &=& t^\oc_{\intr{\Gamma}} \odot \oc(t^\Ctx_\Gamma) 
&:& \oc \rintr{\Gamma} &\to& \coll(\oc \intr{\Gamma})\\
t^\Ty_{A\to B} &=& t^\lin_{\oc \intr{A}, \intr{B}} \odot ((t^{\oc
\Ty}_A)^{-1} \lin t^\Ty_B) 
&:& \intr{A\to B} &\to& \coll(\intr{A\to B})\\
t^{\oc \Ty}_A &=& t^\oc_{\intr{A}} \odot \oc(t^\Ty_A)
&:& \oc \rintr{A} &\to& \coll(\oc \intr{A})
\end{array}
\]
and then we may prove the following lemma:

\begin{lem}\label{lem:link_k_t}
For any context $\Gamma$ and type $A$, we have
$k^\Ctx_\Gamma = t^\Ctx_\Gamma \circ \epsilon_{\intr{\Gamma}}$
and $k^\Ty_A = t^\Ty_A \circ \epsilon_{\intr{A}}$.
\end{lem}
\begin{proof}
A direct diagram chase.
\end{proof}

Finally, we also give the following concrete characterization of the
linear mediating isos:

\begin{lem}\label{lem:link_t_s}
For $\Gamma$ a context, $A$ a type, $\gamma \in
\M_f(\rintr{\Gamma}), \x_{\oc \Gamma} \in \wconf{\oc \Gamma}, a \in
\rintr{A}, \x_A \in \wconf{A}$,
\[
(t^{\oc \Ctx}_\Gamma)_{\gamma, \x_{\oc \Gamma}} =
\delta_{s^\Ctx_\Gamma(\gamma), \x_{\oc \Gamma}}
\qquad
\qquad
(t^{\Ty}_A)_{a, \x_A} = \delta_{s^\Ty_A(a), \x_A}\,.
\]
\end{lem}
\begin{proof}
A direct computation.
\end{proof}

We must prove that the two interpretations match up to these
mediating isos. For constants and primitives this is 
the following lemma, which holds by immediate inspection: 

\begin{lem}\label{lem:ground_preservation}
The diagrams of Figure \ref{fig:pres_prim} commute in $\Rel{\N}_\oc$,
for any context $\Gamma$, ground type $\tx$, and with $k :
\rintr{\tbool} \with \rintr{\tx} \with \rintr{\tx} \to \coll(\gbool
\with \gx \with \gx)$ the obvious isomorphism.
\end{lem}
\begin{figure}
\begin{mathpar}
\scalebox{1}{$
\xymatrix@R=15pt{
\rintr{\Gamma}
        \ar[r]^{\rintr{v}}
        \ar[d]_{k^\Ctx_\Gamma}&
\rintr{\tx}
        \ar[d]^{k^\Ty_\tx}\\
\coll(\intr{\Gamma})
        \ar[r]_{\coll_\oc(\intr{v})}&
\coll(\intr{\tx})
}$}
\and
\scalebox{1}{$
\xymatrix@R=15pt{
\rintr{\tbool} \with \rintr{\tx} \with \rintr{\tx}
	\ar[r]^{\mathsf{if}}
	\ar[d]_{k}&
\rintr{\tx}
	\ar[d]^{k^\tx}\\
\coll(\gbool \with \gx \with \gx)
	\ar[r]_{\coll_\oc(\mathsf{if})}&
\coll(\intr{\tx})
}$}
\and
\scalebox{1}{$
\xymatrix@R=15pt{
\rintr{\tnat}
	\ar[r]^{\mathsf{succ}}
	\ar[d]_{k^\tnat}&
\rintr{\tnat}
	\ar[d]^{k^\tnat}\\
\coll(\gnat)
	\ar[r]_{\coll_\oc(\mathsf{succ})}&
\coll(\gnat)
}$}
\and
\scalebox{1}{$
\xymatrix@R=15pt{
\rintr{\tnat}
	\ar[r]^{\mathsf{pred}}
	\ar[d]_{k^\tnat}&
\rintr{\tnat}
	\ar[d]^{k^\tnat}\\
\coll(\gnat)
	\ar[r]_{\coll_\oc(\mathsf{pred})}&
\coll(\gnat)
}$}
\and
\scalebox{1}{$
\xymatrix@R=15pt{
\rintr{\tnat}
	\ar[r]^{\mathsf{iszero}}
	\ar[d]_{k^\tnat}&
\rintr{\tbool}
	\ar[d]^{k^\tbool}\\
\coll(\gnat)
	\ar[r]_{\coll_\oc(\mathsf{iszero})}&
\coll(\gbool)
}$}
\and
\scalebox{1}{$
\xymatrix@R=15pt@C=30pt{
\rintr{\Gamma}
        \ar[r]^{\rintr{\coin}}
        \ar[d]_{k^\Ctx_\Gamma}&
\rintr{\tbool}
        \ar[d]^{k^\tbool}\\
\coll(\intr{\Gamma})
        \ar[r]_{\coll_\oc(\intr{\coin})}&
\coll(\gbool)
}$}
\end{mathpar}
\caption{Preservation of basic primitives}
\label{fig:pres_prim}
\end{figure}

\subsubsection{Recursion} Preservation of the recursion combinator boils
down to:

\begin{prop}\label{prop:coll_cont}
Consider $A, B$ arenas. Then, the collapse function is continuous:
\[
\coll(-) : \Strat(A, B) \to \Rel{\N}(\coll(A), \coll(B))
\]
\end{prop}
\begin{proof}
Consider directed $D \subseteq \Strat(A, B)$, and $\x_A \in
\wconf{A}, \x_B \in \wconf{B}$. 
%
We have
\[
\wit^+_{\vee D}(\x_A, \x_B) = \bigcup_{\sigma \in D} \wit^+_\sigma(\x_A,
\x_B)
\]
directly by Proposition \ref{prop:sup_pcov}. But additionally, we have
\begin{eqnarray}
\mathcal{P}_f(\wit^+_{\vee D}(\x_A, \x_B)) = \bigcup_{\sigma \in D}
\mathcal{P}_f(\wit^+_\sigma(\x_A, \x_B))\,.\label{eq:better_cont}
\end{eqnarray}

Indeed, if $X \subseteq_f \wit^+_{\vee D}(\x_A, \x_B)$ then there is a
finite $Y \subseteq_f D$ such that 
\[
X \subseteq \bigcup_{\sigma \in Y} \wit^+_\sigma(\x_A, \x_B)\,,
\]
but as $D$ is directed, there is some $\tau \in D$ such that for all
$\sigma \in Y$, we have $\sigma \cleq \tau$. It immediately follows that
$X \subseteq \wit_\tau^+(\x_A, \x_B)$ as well. Reciprocally, if $X
\subseteq_f \wit_\sigma^+(\x_A, \x_B)$ for some $\sigma \in D$, then
clearly $X \subseteq_f \wit_{\vee D}^+(\x_A, \x_B)$ as well since
$\sigma \cleq \vee D$.

From there, we can calculate:
\[
\coll(\vee D)_{\x_A, \x_B} = \bigvee_{X\subseteq_f \wit_{\vee
D}^+(\x_A, \x_B)} \sharp X = \bigvee_{\sigma \in D} \bigvee_{X
\subseteq_f \wit_\sigma^+(\x_A, \x_B)} \sharp X = \bigvee_{\sigma \in D}
\coll(\sigma)_{\x_A, \x_B}
\]
using the definition of $\coll_\sigma(\x_A, \x_B)$ as a limit of finite
sums.
\end{proof}

From there, it is easy to deduce preservation of the recursion
combinator:

\begin{prop}\label{prop:pres_recursion}
Consider $\Gamma$ a context and $A$ a type. Then, the diagram
\[
\xymatrix@R=10pt@C=70pt{
\rintr{\Gamma}
	\ar[r]^{\Y_{\rintr{\Gamma}, \rintr{A}}}
	\ar[d]_{k^\Ctx_\Gamma}&
\rintr{(A \to A) \to A}
	\ar[d]^{k^\Ty_{(A\to A) \to A}}\\
\coll(\intr{\Gamma})
	\ar[r]_{\coll_\oc(\Y_{\intr{\Gamma}, \intr{A}})}&
\coll(\intr{(A\to A)\to A})
}
\]
commutes in $\Rel{\N}_\oc$.
\end{prop}
\begin{proof}
First, the following diagram commutes in $\Rel{\N}_\oc$ for all
$n \in \mathbb{N}$:
\[
\xymatrix@R=10pt@C=60pt{
\top
	\ar[r]^{\Y_{\rintr{A}}^n}
	\ar[d]_{k^\top}&
\rintr{(A \to A) \to A}
	\ar[d]^{k^\Ty_{(A\to A) \to A}}\\
\coll(\top)
	\ar[r]_{\coll_\oc(\Y_{\intr{A}}^n)}&
\coll(\intr{(A\to A)\to A})
}
\]
as follows directly by induction on $n$, by direct application of the
preservation of the cartesian closed structure. By Proposition
\ref{prop:coll_cont}, we may take the supremum and get that
\[
\xymatrix@R=10pt@C=60pt{
\top
	\ar[r]^{\Y_{\rintr{A}}}
	\ar[d]_{k^\top}&
\rintr{(A \to A) \to A}
	\ar[d]^{k^\Ty_{(A\to A) \to A}}\\
\coll(\top)
	\ar[r]_{\coll_\oc(\Y_{\intr{A}})}&
\coll(\intr{(A\to A)\to A})
}
\]
commutes in $\Rel{\N}_\oc$. The proposition follows by
preservation of the terminal object.
\end{proof}

\subsubsection{Preservation of the interpretation}
Finally, we may conclude our main theorem.

\begin{thm}\label{thm:main}
Consider $\Gamma \vdash M : A$ any term of $\nPCF$. Then,
\[
\xymatrix@R=10pt{
\rintr{\Gamma}
	\ar[r]^{\rintr{M}}
	\ar[d]_{k^\Ctx_\Gamma}&
\rintr{A}
	\ar[d]^{k^\Ty_A}\\
\coll(\intr{\Gamma})
	\ar[r]_{\coll_\oc(\intr{M})}&
\coll(\intr{A})
}
\]
commutes in $\Rel{\N}_\oc$.
\end{thm}
\begin{proof}
By induction on the typing derivation $\Gamma \vdash M : A$. The
combinators of the simply-typed $\lambda$-calculus
follow by  
preservation of the cartesian closed structure. For constants, it
follows from Lemma \ref{lem:ground_preservation}. For $\mathbf{if},
\mathbf{pred}, \mathbf{succ}, \mathbf{iszero}$ and $\mathbf{coin}$, it
follows from Lemma \ref{lem:ground_preservation} and the preservation of
the cartesian structure. Recursion is by Proposition
\ref{prop:pres_recursion}.
\end{proof}

A direct reformulation of this theorem is the following:

\begin{thm}\label{thm:main2}
Consider any term $\Gamma \vdash M : A$ of $\nPCF$, and $\gamma \in
\rintr{\Gamma}, a \in \rintr{A}$. Then,
\[
\rintr{M}_{\gamma, a} = \sharp \wit_{\intr{M}}^+(s^\Ctx_\Gamma(\gamma),
s^\Ty_A(a))\,.
\]
\end{thm}
\begin{proof}
By Theorem \ref{thm:main}, we have
\[
\xymatrix@R=10pt{
\rintr{\Gamma}
        \ar[r]^{\rintr{M}}
        \ar[d]_{k^\Ctx_\Gamma}&
\rintr{A}
        \ar[d]^{k^\Ty_A}\\
\coll(\intr{\Gamma})
        \ar[r]_{\coll_\oc(\intr{M})}&
\coll(\intr{A})
}
\]
in $\Rel{\N}_\oc$. Simplifying Kleisli composition and using Lemma
\ref{lem:link_k_t}, we have
\[
\xymatrix@R=15pt{
\M_f(\rintr{\Gamma})
	\ar[rr]^{\rintr{M}}
	\ar[d]_{\oc (t^\Ctx_\Gamma)}&&
\rintr{A}
	\ar[d]^{t^\Ty_A}\\
\oc \coll(\intr{\Gamma})
	\ar[r]_{t^\oc_{\intr{\Gamma}}}&
\coll(\oc \intr{\Gamma})
	\ar[r]_{\coll \intr{M}}&
\coll(\intr{A})
}
\]
in $\Rel{\N}$; but by definition of $t^{\oc \Ctx}_\Gamma$, this amounts
exactly to
\[
\xymatrix@R=15pt{
\M_f(\rintr{\Gamma})
        \ar[r]^{\rintr{M}}
        \ar[d]_{t^{\oc \Ctx}_\Gamma}&
\rintr{A}
        \ar[d]^{t^\Ty_A}\\
\coll(\oc \intr{\Gamma})
        \ar[r]_{\coll \intr{M}}&
\coll(\intr{A})
}
\]
in $\Rel{\N}$. The theorem follows by Lemma \ref{lem:link_t_s} and
direct computation.
\end{proof}

This answers our original question: at higher-order types, the weighted
relational model counts witnesses in the concurrent game semantics, up
to positive symmetry.

\section{Collapse of $\R$-weighted strategies}
\label{sec:coll_rw}

In this last technical section, we show how all the results above generalize to
the collapse of strategies whose configurations are labelled with
elements of a continuous semiring $\R$.

\subsection{$\R$-strategies} We build a relative Seely $\sim$-category 
$\R$-$\Strat$, for any $\R$. 

\subsubsection{Basic definition} As for
$\Strat$, the objects of $\R$-$\Strat$ are all arenas. We define:

\begin{defi}
Consider $A$ a game. An \textbf{$\R$-strategy} on $A$ is a strategy
$\sigma : A$, with
\[
\v_\sigma : \confp{\sigma} \to \R
\]
a \textbf{valuation}, \emph{invariant under symmetry}: for all $x
\sym_\sigma y$, $\v_\sigma(x) = \v_\sigma(y)$. 
\end{defi}

For instance, using $\R = \overline{\mathbb{R}}_+$, we may adjoin to the
strategy $\coin : \gbool$ a valuation $\v :
\confp{\coin} \to \overline{\mathbb{R}}_+$ with $\v(\{\qu^-, \ttrue^+\})
= \v(\{\qu^-, \tfalse^+\}) = \frac{1}{2}$; representing a fair coin
toss.  

For $A$ and $B$ arenas, the homset $\R$-$\Strat(A, B)$ comprises
\emph{visible, exhaustive} $\R$-strategies on $A \vdash B$ -- visibility and
exhaustivity are undisturbed by the presence of the valuation.

\subsubsection{Basic strategies and operations} For any arena $A$ and any $x_A
\parallel x_A \in \confp{\cc_A}$, we set
\[
\v_{\cc_A}(x_A \parallel x_A) = 1\,,
\]
as by Proposition \ref{prop:cc_pcov} all $+$-covered configurations of
copycat have this form.
All copycat strategies involved in the relative Seely $\sim$-category
structure are made into $\R$-strategies similarly, by setting their
valuation to be $1$ everywhere -- this covers associators, unitors,
projections, evaluation, dereliction, digging, and Seely
isomorphisms. All operations on strategies involved in the relative Seely
category structure are lifted to $\R$-strategies with:
\[
\begin{array}{rcl}
\v_{\tau \odot \sigma}(x^\tau \odot x^\sigma) &=& \v_{\sigma}(x^\sigma)
\cdot \v_\tau(x^\tau)\\
\v_{\sigma\tensor \tau}(x^\sigma \tensor x^\tau) &=& \v_\sigma(x^\sigma)
\cdot \v_\tau(x^\tau)\\
\v_{\tuple{\sigma, \tau}}(\inj_\sigma(x^\sigma)) &=&
\v_\sigma(x^\sigma)
\end{array}
\qquad
\begin{array}{rcl}
\v_{\tuple{\sigma, \tau}}(\inj_\tau(x^\tau)) &=& 
\v_\tau(x^\tau)\\
\v_{\oc(\sigma)}([(x^i)_{i\in I}]) &=& \prod_{i\in I}(\v_\sigma(x^i))
\end{array}
\]
leveraging the characterizations of $+$-covered configurations for these
operations, respectively found in Propositions \ref{prop:char_comp},
\ref{prop:def_tensor}, \ref{prop:def_pairing}, and
\ref{prop:char_pcov_bang}. In the last case, $\prod$ denotes the
iterated product $(\cdot)$ of $\R$. One must ensure that the resulting valuation
is invariant under symmetry, which is immediate (in the last case,
using that the product $\cdot$ is commutative).

Likewise, the partial order $\cleq$ is extended to $\R$-strategies by
setting $\sigma \cleq \tau$ if it holds for the underlying strategies,
and $\v_\sigma(x) = \v_\tau(x)$ for all $x \in \confp{\sigma}$. It is
clear that $\cleq$ retains the same completeness properties, and that
all operations on $\R$-strategies are continuous.

In particular, for any strict arenas $\Gamma$ and $A$, we may define the
recursion combinator
\[
\Y_{\Gamma, A} \in \text{$\R$-$\Strat_\oc(\Gamma, \oc (\oc A \lin A)
\lin A)$}
\]
exactly as in Section \ref{subsubsec:recursion} -- the same strategy,
with valuation again set to $1$ everywhere.
\subsubsection{Positive isomorphisms} Finally, we must adapt the
equivalence relation on strategies.

\begin{defi}
Consider $A$ a game, and $\sigma, \tau : A$ two $\R$-strategies.

A \textbf{positive isomorphism} $\varphi : \sigma \simstrat \tau$ is a
positive isomorphism between the underlying strategies, such that for all
$x \in \confp{\sigma}$, we have $\v_\tau(\varphi(x)) = \v_\sigma(x)$.
\end{defi}

We say that $\sigma, \tau : A$ are \textbf{positively isomorphic},
written $\sigma \simstrat \tau$, if there exists a positive isomorphism
$\varphi : \sigma \simstrat \tau$. We must ensure that this
valuation-aware equivalence relation is still preserved under all
operations on strategies -- it is evident for all, save composition.

For composition, we need more information on how positive isos
are propagated:

\begin{restatable}{prop}{horizontal}
Let $\sigma, \sigma' \in
\Strat(A, B)$; $\tau, \tau' \in \Strat(B, C)$; 
$\varphi : \sigma \simstrat \sigma'$, $\psi : \tau \simstrat \tau'$. 

Then, there exists a positive isomorphism $\psi\odot \varphi : \tau
\odot \sigma \simstrat \tau' \odot \sigma'$ such that for all $x^\tau
\odot x^\sigma \in \confp{\tau \odot \sigma}$, writing $y^{\tau'} \odot
y^{\sigma'}  = (\psi\odot \varphi)(x^\tau \odot x^\sigma) \in
\confp{\tau' \odot \sigma'}$, we have  
\[
\varphi(x^\sigma) \sym_{\sigma'} y^{\sigma'}\,,
\qquad
\qquad
\psi(x^\tau) \sym_{\tau'} y^{\tau'}\,.
\]
\end{restatable}

See Appendix \ref{app:horizontal} for the proof. From that, we may deduce:

\begin{prop}
Let $\sigma, \sigma' \in
\text{$\R$-$\Strat$}(A, B)$; $\tau, \tau' \in \text{$\R$-$\Strat$}(B, C)$;
$\varphi : \sigma \simstrat \sigma'$, $\psi : \tau \simstrat \tau'$.

Then, $\psi \odot \varphi : \tau \odot \sigma \simstrat \tau' \odot
\sigma'$ is a positive iso between $\R$-strategies.
\end{prop}
\begin{proof}
For $x^\tau \odot x^\sigma \in \confp{\tau \odot \sigma}$, writing
$y^{\tau'} \odot y^{\sigma'} = (\psi \odot \varphi)(x^\tau \odot
x^\sigma)$, we calculate
\begin{eqnarray*}
\v_{\tau' \odot \sigma'}((\psi \odot \varphi)(x^\tau \odot
x^\sigma)) &=& \v_{\sigma'}(y^{\sigma'}) \cdot \v_{\tau'}(y^{\tau'})\\
&=& \v_{\sigma'}(\varphi(x^\sigma)) \cdot \v_{\tau'}(\psi(x^\tau))\\
&=& \v_\sigma(x^\sigma) \cdot \v_\tau(y^\tau)\\
&=& \v_{\tau \odot \sigma}(x^\tau \odot x^\sigma)\,,
\end{eqnarray*}
where $\v_{\sigma'}(y^{\sigma'}) = \v_{\sigma'}(\varphi(x^\sigma))$ as
$y^{\sigma'} \sym_{\sigma'} \varphi(x^\sigma)$ and likewise for $\tau$.
\end{proof}

\begin{cor}
There is a relative Seely $\sim$-category $\R$-$\Strat$.
\end{cor}
\begin{proof}
It remains to establish the required positive isomorphisms, \emph{i.e.}
to show that the corresponding positive isomorphisms for $\Strat$
preserve valuations. As an illustration, recall from Proposition
\ref{prop:sim_cat} that associativity is realized with the positive
isomorphism
\[
\alpha_{\sigma, \tau, \delta} : (\delta \odot \tau) \odot \sigma \simstrat \delta \odot
(\tau \odot \sigma)
\]
such that $\alpha_{\sigma, \tau, \delta}((x^\delta
\odot x^\tau) \odot x^\sigma) = x^\delta \odot (x^\tau \odot x^\sigma)$
for all $(x^\delta \odot x^\tau) \odot x^\sigma \in \confp{\delta \odot
(\tau \odot \sigma)}$. Clearly, this preserves valuations by
associativity of $\cdot$. Other cases are similar.
\end{proof}

\subsubsection{Interpretation of $\R$-$\PCF$} 
All basic strategies for $\nPCF$ primitives have valuation set to
$1$ everywhere, completing the interpretation of $\nPCF$. But valuations
remain trivial:

\begin{prop}\label{prop:npcf_one}
Consider $\Gamma \vdash M : A$ a term of $\nPCF$.

Then, for all $x \in \confp{\intr{M}}$, we have $\v_{\intr{M}}(x) = 1$.
\end{prop}

This is obvious: all basic strategies have all valuations $1$,
and the operations on strategies only involve the product $\cdot$ of
$\R$, never the sum. To explain that, recall that in $\Rel{\R}$ the sum
serves to aggregate weights for all executions made distinct by
non-deterministic choices. But $\R$-$\Strat$ maintains explicit
branching information, and each witness represents only one individual
execution -- so it makes sense that coefficients should remain $1$.

So as to better illustrate the model of $\R$-strategies, we add a new primitive:
\[
\inferrule
	{\Gamma \vdash M : A}
	{\Gamma \vdash r \cdot M : A}
\]
for all $r \in \R$ -- we refer to $\R$-$\PCF$ for the enriched language.
There is a matching operation:

\begin{defi}
Consider $A$ a game, and $\sigma : A$ a $\R$-strategy. 

We set $r \cdot \sigma : A$ with strategy $\sigma$ and valuation 
$\v_{r\cdot \sigma}(x) = r\cdot \v_{\sigma}(x)$
for all $x \in \confp{\sigma}$.
\end{defi}

Altogether, this yields an interpretation of $\R$-$\PCF$ into
$\R$-$\Strat_\oc$, sending a term $\Gamma \vdash M : A$ to $\intr{M}
\in \text{$\R$-$\Strat_\oc(\intr{\Gamma}, \intr{A})$}$.
We must also set the interpretation of $\R$-$\PCF$ in $\Rel{\R}$, set
with the exact same clauses as for the interpretation in $\Rel{\N}$,
except for: 
\[
\rintr{\Gamma \vdash r \cdot M : A}_{\mu, a} = r\cdot \rintr{\Gamma
\vdash M : A}_{\mu, a}\,.
\]
%

This completes the interpretation of any $\Gamma \vdash M : A$ as
$\rintr{M} \in \Rel{\R}_\oc(\rintr{\Gamma}, \rintr{A})$, which we must
now compare with $\intr{M} \in \text{$\R$-$\Strat_\oc(\intr{\Gamma},
\intr{A})$}$.

\subsection{A relative Seely $\sim$-functor}
Next, we show how Corollary \ref{cor:seely_trans}
extends in the presence of quantitative valuations. With
the earlier developments of this paper this is mostly a formality:
as all earlier compatibility results are realized by explicit
bijections between sets of witnesses, we must only exploit that these
bijections preserve valuations.

First, we define the quantitative collapse as follows. For any $\sigma
\in \text{$\R$-$\Strat(A, B)$}$, we set:
\begin{eqnarray}
\coll(\sigma)_{\x_A, \y_B} &=& \sum_{x \in \wit^+_\sigma(\x_A,
\x_B)} \v_\sigma(x) \label{eq:qcoll}
\end{eqnarray}
for all $\x_A \in \wconf{A}$ and $\x_B \in \wconf{B}$. It is a clear
generalization of \eqref{eq:coll}, with all witnesses weighted according
to their valuation. It is a conservative extension of \eqref{eq:coll}:
for strategies arising from terms without $r \cdot -$, by Proposition
\ref{prop:npcf_one} it is equivalent to \eqref{eq:coll}. 

\subsubsection{Composition} First, we show that \eqref{eq:qcoll} is
compatible with composition.
Fortunately, it suffices to exploit the bijections introduced in Section
\ref{subsec:pres_comp_main}, along with \emph{integer division}:

\begin{prop}
For $\sigma \in \text{$\R$-$\Strat(A, B)$}$, $\tau \in
\text{$\R$-$\Strat(B, C)$}$, $\x_A \in \wconf{A}, \x_C \in
\wconf{C}$:
\[
\coll(\tau \odot \sigma)_{\x_A, \x_C} = \sum_{\x_B \in \wconf{B}}
\coll(\sigma)_{\x_A, \x_B} \cdot \coll(\tau)_{\x_B, \x_C}\,.
\]
\end{prop}
\begin{proof}
Let us first fix some $\x_B \in \wconf{B}$. We then perform the
computation in $\R$:
\begin{eqnarray*}
&& (\sharp \ntilde{\x_A}) * (\sharp \tilde{\x_B}) * (\sharp
\ptilde{\x_C}) * \sum_{x^\tau \odot x^\sigma \in \wit^+_{\sigma,
\tau}(\x_A, \x_B, \x_C)} \v_{\tau \odot \sigma}(x^\tau \odot x^\sigma)\\
&=& \sum_{\theta_A^- \in \ntilde{\x_A}}
\sum_{\theta_B \in \tilde{\x_B}}
\sum_{\theta_C^+ \in \ptilde{\x_C}}
\sum_{x^\tau \odot x^\sigma \in \wit^+_{\sigma, \tau}(\x_A, \x_B, \x_C)}
\v_\sigma(x^\sigma) \cdot \v_\tau(x^\tau)\\
&=& \sum_{\varphi_A^- \in \ntilde{\x_A}}
\sum_{\varphi_B \in \tilde{\x_B}}
\sum_{\varphi_C^+ \in \ptilde{\x_C}}
\sum_{y^\sigma \in \wit^+_\sigma(\x_A, \x_B)}
\sum_{y^\tau \in \wit^+_\tau(\x_B, \x_C)}
\v_\sigma(y^\sigma) \cdot \v_\tau(y^\tau)\\
&=& (\sharp \ntilde{\x_A}) * (\sharp \tilde{\x_B}) * (\sharp
\ptilde{\x_C}) *
\left[
\left(\sum_{y^\sigma \in \wit^+_\sigma(\x_A, \x_B)}
\v_\sigma(y^\sigma)\right) \cdot
\left(\sum_{y^\tau \in \wit^+_\tau(\x_B, \x_C)}
\v_\tau(y^\tau)\right)
\right]\\
&=& (\sharp \ntilde{\x_A}) * (\sharp
\tilde{\x_B}) * (\sharp \ptilde{\x_C}) * (\coll(\sigma)_{\x_A, \x_B} \cdot
\coll(\tau)_{\x_B, \x_C})
\end{eqnarray*}
using the definition of integer multiplication in $\R$ and of the
valuation of $\tau \odot \sigma$; then substituting by
$\Phi$ of Lemma \ref{lem:mainbij} and using that the valuation is
invariant under symmetry; and using distributivity of $\cdot$
over $+$ in $\R$ and the definition of integer multiplication. 

Now, since $\R$ satisfies \emph{integer division}, we may deduce the
equality in $\R$:
\[
\sum_{x^\tau \odot x^\sigma \in \wit^+_{\sigma, \tau}(\x_A, \x_B, \x_C)}
\v_{\tau \odot \sigma}(x^\tau \odot x^\sigma)
= \coll(\sigma)_{\x_A, \x_B} \cdot \coll(\tau)_{\x_B, \x_C}\,,
\]
by dividing each side by $(\sharp \ntilde{\x_A})$, $(\sharp
\tilde{\x_B})$, $(\sharp \ptilde{\x_C})$.
By \eqref{eq7}, summing both sides over all $\x_B \in \wconf{B}$
concludes the proof of the desired equation.
\end{proof}

We do not know if integer division is really needed. One could avoid it
by extracting from the bijection $\Phi$ in Lemma
\ref{lem:mainbij} a direct bijection preserving symmetry classes
\[
\wit^+_{\sigma, \tau}(\x_A, \x_B, \x_C) 
\bij \wit^+_\sigma(\x_A, \x_B) \times \wit^+_\tau(\x_B, \x_C)\,,
\]
but it is not immediately clear how to do that. Assuming integer
division does not remove any interesting example of continuous semiring;
so we did not push this. We obtain: 

\begin{prop}
The operation $\coll(-) : \text{$\R$-$\Strat$} \to \Rel{\R}$ is a
$\sim$-functor.
\end{prop}
\begin{proof}
It remains to prove that $\coll(-)$ preserves the identities and the
equivalence relation. For identities the proof of Proposition
\ref{prop:cc_collapse} applies, using that valuations for $\cc_A$ are
$1$.

For the equivalence relation, given $\sigma, \tau \in
\text{$\R$-$\Strat(A, B)$}$ and $\x_A \in \wconf{A}, \x_B \in
\wconf{B}$, given $\varphi : \sigma \simstrat \tau$ we proved in 
Proposition \ref{prop:coll_pres_sim} that $\varphi$ specializes to a
bijection
\[
\varphi : \wit^+_\sigma(\x_A, \x_B) \bij \wit^+_\tau(\x_A, \x_B)\,,
\]
but as $\varphi$ is now required to preserve valuations, we have
$\coll(\sigma)_{\x_A, \x_B} = \coll(\tau)_{\x_A, \x_B}$ as needed.
\end{proof}

\subsubsection{Preservation of symmetric monoidal structure} For the
tensor, it is straightforward:

\begin{prop}
The operation $\coll(-)$ is a symmetric monoidal $\sim$-functor.
\end{prop}
\begin{proof}
The structural isomorphisms involved are the $\R$-weighted relations
defined with the same formulas as in Section
\ref{subsubsec:pres_mon_str}. All coherence laws follow. For naturality
of $t^\tensor_{A, B}$, we build on the proof of Proposition
\ref{prop:smsf} by noting that the bijection 
\[
(- \tensor -) : \wit^+_\sigma(\x_A, \y_{A'}) \times \wit^+_\tau(\x_B,
\y_{B'}) \to \wit^+_{\sigma \tensor \tau}(\x_A \parallel \x_{B}, \y_{A'}
\parallel \y_{B'})
\]
is such that $\v_{\sigma\tensor \tau}(x^\sigma \tensor x^\tau) =
\v_\sigma(x^\sigma) \cdot \v_\tau(x^\tau)$ by definition of the
valuation for tensor. From this, the naturality of $t^\tensor_{A, B}$
follows by an immediate calculation.
\end{proof}

\subsubsection{Preservation of promotion} As expected, for $C$ a strict
arena we set $t^\oc_C$ as the $\R$-weighted relation defined with
the same formula as in Section \ref{subsubsec:lastglitch}.
We prove:

\begin{prop}
For $A, B$ arenas, strategy $\sigma \in \Strat(A, B)$, and symmetry
classes $\x_{\oc A} \in \wconf{\oc A}$, $\y_{\oc B} \in \wconf{\oc B}$
with $\y_{\oc B} = [\y_B^1, \dots, \y_B^n]$ with each $\y_B^i$
non-empty, we have:
\[
\coll(\oc \sigma)_{\x_{\oc A}, \y_{\oc B}}
= 
\sum_{\stackrel{(\z_A^1, \dots,
\z_A^n)\,\text{s.t.}}
{\x_{\oc A} = [\z_A^i \mid \z_A^i \neq \emptyset]}}
\prod_{1\leq i \leq n}
\coll(\sigma)_{\z_A^i, \y_B^i}
\]
\end{prop}
\begin{proof}
We perform the computation:
\begin{eqnarray*}
\sharp \m(\x_{\oc A}) * \coll(\oc \sigma)_{\x_{\oc A}, \y_{\oc B}}
&=& \sum_{\pi \in \m(\x_{\oc A})}
\sum_{x^{\oc \sigma} \in \wit^+_{\oc \sigma}(\x_{\oc A},
\y_{\oc B})}
\v_{\oc \sigma}(x^{\oc \sigma})\\
&=& \sum_{\varpi \in \m(\x_{\oc A})}
\sum_{\stackrel{(\z_A^1, \dots,
\z_A^n)\,\text{s.t.}}{
\x_{\oc A} = [\z_A^i \mid \z_A^i \neq \emptyset]}}
\sum_{(x^i) \in \Pi_{1\leq i \leq n} \wit^+_\sigma(\z_A^i, \y_B^i)}
\v_{\oc \sigma}([(x^i)_{1\leq i \leq n}])\\
&=& \sum_{\varpi \in \m(\x_{\oc A})}
\sum_{\stackrel{(\z_A^1, \dots,
\z_A^n)\,\text{s.t.}}{
\x_{\oc A} = [\z_A^i \mid \z_A^i \neq \emptyset]}}
\sum_{(x^i) \in \Pi_{1\leq i \leq n} \wit^+_\sigma(\z_A^i, \y_B^i)}
\prod_{1\leq i \leq n}
\v_\sigma(x^i)\\
&=& \sum_{\varpi \in \m(\x_{\oc A})}
\sum_{\stackrel{(\z_A^1, \dots,
\z_A^n)\,\text{s.t.}}{
\x_{\oc A} = [\z_A^i \mid \z_A^i \neq \emptyset]}}
\prod_{1\leq i \leq n}
\sum_{x^i \in \wit^+_\sigma(\z_A^i, \y_B^i)}
\v_\sigma(x^i)\\
&=& \sum_{\varpi \in \m(\x_{\oc A})}
\sum_{\stackrel{(\z_A^1, \dots,
\z_A^n)\,\text{s.t.}}{
\x_{\oc A} = [\z_A^i \mid \z_A^i \neq \emptyset]}}
\prod_{1\leq i \leq n}
\coll(\sigma)_{\z_A^i, \y_B^i}\\
&=& \sharp \m(\x_{\oc A}) * \sum_{\stackrel{(\z_A^1, \dots,
\z_A^n)\,\text{s.t.}}{
\x_{\oc A} = [\z_A^i \mid \z_A^i \neq \emptyset]}}
\prod_{1\leq i \leq n}
\coll(\sigma)_{\z_A^i, \y_B^i}
\end{eqnarray*}
using the definition of integer multiplication; the bijection $U$ in
Lemma \ref{lem:bij_bang}; the definition of the valuation for $\oc
\sigma$; distributivity of $\cdot$ over sum; definition of
$\coll(\sigma)$; and again definition of integer multiplication.
Finally, the desired equality follows by integer division.
\end{proof}

From this, it follows -- with the same proof -- that promotion is
preserved as in Proposition \ref{prop:pres_prom}.
As for composition, it is not clear whether one can avoid
integer division here.  

From this point, we can conclude the preservation of the relative Seely
structure.

\begin{cor}\label{cor:main_seely_R}
We have a relative Seely $\sim$-functor $\coll(-) : \text{$\R$-$\Strat$}
\to \Rel{\R}$.
\end{cor}
\begin{proof}
It remains to define $t^\top, t^\with_{C, D}$ and $t^\lin_{A, C}$ for
$A, C, D$ arenas with $C, D$ strict -- those are defined with the same
formulas as in Section \ref{subsubsec:rigid_seely}. The coherence laws
follow likewise.
\end{proof}

\subsection{Preservation of the Interpretation} While Corollary
\ref{cor:main_seely_R} does the heavy lifting, there remain a few things
to check. First, preservation of recursion boils down to:

\begin{prop}\label{prop:coll_cont_R}
Consider $A, B$ arenas. Then, the collapse function is continuous:
\[
\coll(-) : \text{$\R$-$\Strat(A, B)$} \to \Rel{\R}(\coll(A), \coll(B))
\]  
\end{prop}
\begin{proof}
Consider directed $D \subseteq \text{$\R$-$\Strat(A, B)$}$, and $\x_A \in
\wconf{A}, \x_B \in \wconf{B}$. We compute:
\begin{eqnarray*}
\coll(\vee D)_{\x_A, \x_B} &=& \sum_{x^\sigma \in \wit^+_{\vee D}(\x_A,
\x_B)} \v_\sigma(x^\sigma)\\
&=& \bigvee_{X \subseteq_f \wit^+_{\vee D}(\x_A, \x_B)}
\sum_{x^\sigma \in X} \v_\sigma(x^\sigma)\\
&=& \bigvee_{\sigma \in D} \bigvee_{Y \subseteq_f \wit^+_{\sigma}(\x_A, \x_B)}
\sum_{x^\sigma \in Y} \v_\sigma(x^\sigma)
\end{eqnarray*}
which is $\vee_{\sigma \in D} \coll(\sigma)_{\x_A, \x_B}$ by definition
-- we used the definition of infinite sums, and \eqref{eq:better_cont}.
\end{proof}

It follows that the recursion combinator is preserved, with the same
proof as Proposition \ref{prop:pres_recursion}. Likewise, all the
diagrams in Figure \ref{fig:pres_prim} immediately hold. Finally,
the interpretations of $r\cdot -$ trivially agree with each other as
well. To conclude, we have:

\begin{thm}\label{thm:main3}
Consider $\Gamma \vdash M : A$ any term of $\R$-$\PCF$. Then,
\[
\xymatrix@R=15pt{
\rintr{\Gamma}
	\ar[r]^{\rintr{M}}
	\ar[d]_{k^\Ctx_\Gamma}&
\rintr{A}
	\ar[d]^{k^\Ty_A}\\
\coll(\intr{\Gamma})
	\ar[r]_{\coll_\oc(\intr{M})}&
\coll(\intr{A})
}
\]
commutes in $\Rel{\R}_\oc$.
\end{thm}
\begin{proof}
As for Theorem \ref{thm:main} with the ingredients introduced in this
section.
\end{proof}

As in the earlier case, we also provide a more concrete statement:

\begin{thm}\label{thm:main4}
Consider any term $\Gamma \vdash M : A$ of $\R$-$\PCF$, and $\gamma \in
\rintr{\Gamma}, a \in \rintr{A}$. Then,
\[
\rintr{M}_{\gamma, a} = \sum_{x \in \wit_{\intr{M}}^+(s^\Ctx_\Gamma(\gamma),
s^\Ty_A(a))} \v_{\intr{M}}(x)\,.
\]
\end{thm}
\begin{proof}
Direct from Theorem \ref{thm:main3}, with the same proof as for Theorem
\ref{thm:main2}.
\end{proof}

From this, one obtains game semantics for various continuous semirings,
inheriting adequacy properties from \cite{DBLP:conf/lics/LairdMMP13}.
Details are out of scope of the paper.

The weighted relational model is inherently infinite, because the sum
\eqref{eq:relsum} involved in the composition of weighted relations has
no reason to be finite. This infinitary nature is sometimes criticized; 
for instance probabilistic coherence spaces
\cite{DBLP:conf/lics/EhrhardPT11} consist in
enriching the weighted relational model with a biorthogonality
construction ensuring (among other things) that all coefficients remain
finite. So it is noteworthy that no infinity arises in $\R$-$\Strat$:
the construction unfolds just fine with only a plain semiring -- or in
fact, only a \emph{monoid} $(\ev{\R}, \cdot, 1)$! Indeed, as it stands,
the sum only arises when collapsing to $\Rel{\R}$.

\section{Conclusion}

As a rough approximation, there are essentially two families of
denotational models in the legacy of linear logic: on the one hand the
\emph{web-based} semantics such as relational models, coherence spaces
and their weighted counterparts, arising from Girard's quantitative
semantics \cite{DBLP:journals/apal/Girard88}; and on the other hand the
\emph{interactive semantics} drawing inspiration, among others, from
Girard's geometry of interaction \cite{girard1989geometry}. The two
families are great for different things: the former family has had
impressive achievements in modeling quantitative aspects of programming,
with notably the recent full abstraction result for probabilistic $\PCF$
due to Ehrhard, Pagani and Tasson \cite{DBLP:journals/jacm/EhrhardPT18};
while the latter has proved particularly powerful in capturing
effectful programming languages \cite{DBLP:journals/siglog/MurawskiT16}.
It is certainly puzzling that these families, though sharing such a
close genesis, have remained almost separated! 

We believe the results presented here are an important step towards
bringing these two families together, aiming towards a unified landscape
of quantitative denotational models of programming languages.
We proved this for $\PCF$, but
there is no doubt that this extends to other languages or evaluation
strategies -- in fact, the first author and de Visme proved a similar
collapse theorem for the (call-by-value) \emph{quantum
$\lambda$-calculus} \cite{DBLP:journals/pacmpl/ClairambaultV20} (this
relies on some of the constructions of this paper, first appearing in an
unpublished technical report by the first author
\cite{DBLP:journals/corr/abs-2006-05080}). Of course, much remains to be done:
notably, we would like to understand better the links between thin
concurrent games and generalized species of structure
\cite{fiore2008cartesian}. Much of the \changed{present} development is also
reminiscent of issues related to rigid resource terms and the Taylor
development of $\lambda$-terms \cite{DBLP:journals/corr/abs-2008-02665}. 

\paragraph{Acknowledgments} We would like to thank Marc de Visme for
lively discussions during the development of this work, and 
the counter-example to representability (Appendix \ref{app:nonrep}).

This work was supported by the ANR project DyVerSe
(ANR-19-CE48-0010-01); and by the Labex MiLyon (ANR-10-LABX-0070) of
Universit\'e de Lyon, within the program ``Investissements d'Avenir''
(ANR-11-IDEX-0007), operated by the French National Research Agency
(ANR).

\bibliographystyle{alpha}
\bibliography{main}

\appendix

\section{Postponed Proofs and Constructions}


\subsection{Theory of relative adjunctions and comonads.}
\label{app:relative}

We first recall the basic theory of relative adjunctions, and relative
comonads. These are defined with respect to any functor $J : \C \to
\D$, but we only give the special case where $\C$ is a full subcategory
of $\D$, and $J : \C \hookrightarrow \D$ is the inclusion functor. The
general definitions can be found in e.g. \cite{relativemonads}. 

\newcommand{\B}{\mathcal{B}}
\paragraph{Relative adjunctions.} If $F : \C \to \B$ and $G : \B \to
\D$, we say that $F$ is a \textbf{$J$-relative left adjoint} to $G$ if
for every $C \in \C$ and $B \in \B$ there is a natural bijection 
\[
\B(F(C), B) \cong \D(C, G(B)).
\]

We say that $F$ is a \textbf{$J$-relative right adjoint} to $G$ if for
all $C \in \C$ and $B \in \B$ there is 
\[
\B(B, F(C)) \cong \D(G(B), C).
\]
a natural bijection 
These two situations are respectively pictured as the diagrams below:  
\begin{equation}
\label{eq:adj}
\begin{tikzcd}
& \B \arrow[phantom]{d}{\dashv}  \arrow{dr}{G} & \\ 
\C \arrow{ur}{F}   \arrow[hook]{rr}{} & {} & \D  
\end{tikzcd}
\qquad 
\begin{tikzcd}
& \B \arrow[phantom]{d}{\vdash}  \arrow{dr}{G} & \\ 
\C \arrow{ur}{F} \arrow[hook]{rr}{} & {} & \D  
\end{tikzcd}
\end{equation}

Note that this definition is asymmetric: if $F$ is a right adjoint to $G$ relative to $J$, then it does \emph{not} make sense to say that $G$ is a $J$-left adjoint to $F$. 

\paragraph{Relative comonads}
A \textbf{$J$-relative comonad} consists of: \emph{(1)} 
for every $C \in \C$, an object $\oc C \in \D$; \emph{(2)} 
for every $C \in \C$, a morphism $\der_C : \oc C \to C$; and \emph{(3)} 
for every $B, C \in \C$ and $f : \oc B \to C$, a morphism $f^\dagger : \oc B \to \oc C$,
such that for $A, B, C \in \C$, 
\[
\begin{array}{rl}
\text{\emph{(1)}}& 
\text{if $f : \oc B \to C$, then $f = \der_{C} \circ f^\dagger$,}\\
\text{\emph{(2)}}&
\text{$\der_C^\dagger = \id_{\oc C}$,}\\
\text{\emph{(3)}}&
\text{if $f : \oc A \to B$ and $g : \oc B \to C$, then 
$(g \circ f^\dagger)^\dagger = g^\dagger \circ f^\dagger$.}
\end{array}
\]

The axioms ensure that $\oc$ can be extended to a functor, sending $f : B \to C$ to $(f \circ \der_{B})^\dagger$. 

\paragraph{The Kleisli category of a $J$-relative comonad.} The
relationship between adjunctions and comonads extends to the
$J$-relative setting. For $F$ and $G$ as in the
right-hand diagram in \eqref{eq:adj}, their composite $GF$ is
a relative comonad. Conversely, any relative comonad has an associated
\textbf{Kleisli category} $\C_\oc$ which can be used to construct a
relative adjunction: it has objects those of $\C$, and
homsets given by $\C_\oc(B, C) = \D(\oc B, C)$. We have a situation
\[
\begin{tikzcd}
& \C_\oc \arrow[phantom]{d}{\vdash}  \arrow{dr}{} & \\ 
\C \arrow{ur}{} \arrow[hook]{rr}{} & {} & \D  
\end{tikzcd}
\]
where the right adjoint is identity-on-objects and maps $f : B \to C$ to $f \circ \der_{B}$, and the left adjoint maps $C \in \C_\oc$ to $\oc C \in \D$ and $f \in \C_\oc(B, C)$ to $f^\oc \in \D(\oc B, \oc C)$. 

\subsection{The Kleisli category of a relative Seely category.}
\label{app:kleisli}

For a relative Seely category $\C$ as in Definition~\ref{def:relativeseely}, the Kleisli category for the relative comonad $\oc$ is cartesian closed. We give some details of the proof. Recall that the situation is the following:
\[
\begin{tikzcd}
& \C_\oc \arrow[phantom]{d}{\vdash}  \arrow{dr}{G} & \\ 
\C_s \arrow{ur}{F} \arrow[hook]{rr}{} & {} & \C  
\end{tikzcd}
\]
where $F$ and $G$ are defined in the last paragraph of the previous section. 

\paragraph{Products.}
It is easy to show that $J$-relative right adjoints preserve the limits
that $J$ preserves, see e.g. \cite{ulmer1968properties}. Since $J$
preserves products, $\C_\oc$ has all finite products constructed as in
$\C_s$, with projections $\pi_A \circ \der_{A\with B} \in \C_\oc(A\with
B,A)$ and $\pi_B \circ \der_{A\with B} \in \C_\oc(A\with B,B)$. 

\paragraph{Cartesian closure.} For $A, B \in \C_s$, we define the function space $A \tto B = \oc A \lin B$; we know this is an object of $\C_s$ since by definition the  functor $\oc A \lin - $ has type $\C_s \to \C_s$. We use the (relative) closed structure of $\C$ to derive the required bijection. For $A, B, C \in \C_s$,
\begin{align*}
\C_\oc (A \with B, C) &= \C (\oc (A \with B), C) \\ 
&\cong \C(\oc A \tensor \oc B, C)  \qquad \qquad \text{(using $m_{A, B}$)} \\ 
&\cong \C(\oc A, \oc B \lin C) = \C_\oc (A, B \tto C)
\end{align*} 
and this is natural. The evaluation map $\evm_{A,
B} \in \C_\oc((A \tto B) \with A, B)$ is given by 
\[
\oc((A\tto B) \with A) \xrightarrow{m_{A,B}^{-1}} \oc (A \tto B) \tensor \oc A \xrightarrow{\der_{A \tto B} \tensor \oc A} \oc A \lin B \tensor \oc A \xrightarrow{\evm_{\oc A, B}} B.
\]

\subsection{Synchronization up to Symmetry}
\label{app:syncsym}
First we include the proof of: 

\syncsym*
\begin{proof}
Let us write $A_-$ for $A$ with symmetry replaced with $\tilde{A_-} =
\ntilde{A}$, $\ntilde{A_-} = \ntilde{A}$, and $\ptilde{A_-}$ reduced to
identity bijections. Likewise, $C_+$ is $C$ with isomorphism families
restricted to positive symmetries. Then, then we consider the maps of
essp
\[
(\pr_\sigma \parallel \id_C) : \sigma \parallel C_+ \to A^\perp \parallel B
\parallel C^\perp
\qquad
(\id_A \parallel \pr_\tau) : A_- \parallel \tau \to A \parallel B^\perp
\parallel C
\]
are dual \emph{pre-$\sim$-strategies} in the sense of \cite{cg2}.
\emph{Existence} then follows directly from an application of Lemma 3.23 of
\cite{cg2} to these two.
For \emph{uniqueness}, consider $z^\sigma \in \confp{\sigma}$ and $z^\tau
\in \confp{\tau}$ causally compatible together with $\psi^\sigma :
z^\sigma \sym_\sigma x^\sigma$ and $\psi^\tau : z^\tau \sym_\tau x^\tau$
with $\psi^\sigma_A \in \ntilde{A}$, $\psi^\tau_C \in \ptilde{C}$ and
$\psi^\tau_B \circ \theta = \psi^\sigma_B$. Then,
\[
\pr_\sigma (\psi^\sigma \circ (\varphi^\sigma)^{-1}) = \theta_A^-
\parallel (\psi^\sigma_B \circ (\varphi^\sigma_B)^{-1})
\qquad
\pr_\tau (\psi^\tau \circ (\varphi^\tau)^{-1}) = (\psi^\tau_B \circ
(\varphi^\tau_B)^{-1}) \parallel \theta_C^+
\]
where $\psi^\sigma_B \circ (\varphi^\sigma_B)^{-1} = \psi^\tau_B \circ
\theta \circ \theta^{-1} \circ (\varphi^\tau_B)^{-1} = \psi^\tau_B \circ
(\varphi^\tau_B)^{-1}$, so by Proposition \ref{prop:char_comp},
\[
\omega = (\psi^\sigma \circ (\varphi^\sigma)^{-1}) \odot (\psi^\tau \circ
(\varphi^\tau)^{-1}) \in \tildep{\tau\odot \sigma}
\]
but its image by $\pr_{\tau \odot \sigma}$ is a positive symmetry, so
$\omega$ is an identity symmetry by Lemma 3.28 of \cite{cg2}. It follows
easily from Proposition \ref{prop:char_comp} that $\psi^\sigma =
(\varphi^\sigma)^{-1}$ and $\psi^\tau = (\varphi^\tau)^{-1}$. 
\end{proof}

We can also prove the same property on \emph{symmetries} rather than
configurations. For this, we use \emph{higher symmetries} on ess: if
$\theta, \theta' \in \tilde{E}$, we write
$\Theta : \theta \sym_E \theta'$
for a bijection between their graphs, such that writing $\dom(\Theta) =
\{(a_1, b_1) \mid ((a_1, a_2), (b_1, b_2)) \in
\Theta\}$, 
\[
\dom(\Theta) : \dom(\theta) \sym_E \dom(\theta')
\]
and likewise for $\cod(\Theta)$.

\begin{prop}\label{prop:sync_sym2}
Consider $\sigma : A \vdash B$ and $\tau : B \vdash C$ two strategies.

For $\theta^\sigma \in \tildep{\sigma}, \theta^\tau \in \tildep{\tau}$
and $\Theta : \theta^\sigma_B \sym_B \theta^\tau_B$ s.t. the composite
bijection is \emph{secured}:
\[
\theta^\sigma \parallel \theta^\tau_C 
\quad
\stackrel{\pr_{\sigma}\parallel C}{\simeq}
\quad
\theta^\sigma_A \parallel \theta^\sigma_B \parallel \theta^\tau_C 
\quad
\stackrel{A \parallel \Theta \parallel C}{\sym}
\quad
\theta^\sigma_A \parallel \theta^\tau_B \parallel \theta^\tau_C
\quad
\stackrel{A \parallel \pr_{\tau}^{-1}}{\simeq}
\quad
\theta^\sigma_A \parallel \theta^\tau\,,
\]
then there are (necessarily unique) $\vartheta^\sigma \in
\tildep{\sigma}$ and $\vartheta^\tau \in \tildep{\tau}$ causally
compatible, and 
\[
\Phi^\sigma : \theta^\sigma \sym_\sigma \vartheta^\sigma\,,
\qquad
\qquad
\Phi^\tau : \theta^\tau \sym_\tau \vartheta^\tau\,,
\]
s.t.
$\Phi^\sigma_A$ is negative (\emph{i.e.} $\dom(\Phi^\sigma_A)$
and $\cod(\Phi^\sigma_A)$ negative), $\Phi^\tau_C$ is positive, and
$\Phi^\tau_B \circ \Theta = \Phi^\sigma_B$.
\end{prop}
\begin{proof}
Let us write $\theta^\sigma : x^\sigma \sym_\sigma y^\sigma$ and
$\theta^\tau : x^\tau \sym_\tau y^\tau$. Applying Proposition
\ref{prop:sync_sym}, we get
\[
\varphi^\sigma : x^\sigma \sym_\sigma u^\sigma
\qquad
\varphi^\tau : x^\tau \sym_\tau u^\tau
\qquad
\psi^\sigma : y^\sigma \sym_\sigma v^\sigma
\qquad
\psi^\tau : y^\tau \sym_\tau v^\tau
\]
where $u^\sigma, u^\tau$ causally compatible, $v^\sigma, v^\tau$
causally compatible, and satisfying additional properties not listed
here. We may then define $\vartheta^\sigma$ and $\vartheta^\tau$ as the
missing sides of:
\[
\xymatrix{
x^\sigma\ar[r]^{\varphi^\sigma}
	\ar[d]_{\theta^\sigma}&
u^\sigma\\
y^\sigma\ar[r]_{\psi^\sigma}&
v^\sigma
}
\qquad
\qquad
\xymatrix{
x^\tau	\ar[r]^{\varphi^\tau}
	\ar[d]_{\theta^\tau}&
u^\tau\\
y^\tau	\ar[r]_{\psi^\tau}&
v^\tau
}
\]
and $\Phi^\sigma : \theta^\sigma \sym_\sigma \vartheta^\sigma$ and
$\Phi^\tau : \theta^\tau \sym_\tau \vartheta^\tau$ induced by those
commuting diagrams. It is then a simple diagram chasing that the
additional properties are satisfied.
\end{proof}

\subsection{Horizontal Composition of Positive Isomorphisms} 
\label{app:horizontal}
Next we detail the proof of:

\horizontal*
\begin{proof}
We define $\psi \odot \varphi$ on $+$-covered configurations. Take
$x^\tau \odot x^\sigma \in \confp{\tau \odot \sigma}$.
There is, of course, no reason why $\varphi(x^\sigma)$ and
$\psi(x^\tau)$ would be compatible.
However, since $\varphi$ and $\psi$ are positive
isomorphisms, there are (unique) symmetries $\theta_A^-, \theta_B^+, \theta_B^-,
\theta_C^+$ such that
\[
\xymatrix{
x^\sigma 
	\ar[d]_{\varphi}
	\ar[r]^{\pr_\sigma}&
x^\sigma_A \parallel x^\sigma_B
	\ar[d]^{\theta_A^- \parallel \theta_B^+}\\
\varphi(x^\sigma)
	\ar[r]_{\pr_{\sigma'}}&
\varphi(x^\sigma)_A \parallel \varphi(x^\sigma)_B
}
\qquad
\qquad
\xymatrix{
x^\tau
        \ar[d]_{\psi}
        \ar[r]^{\pr_\tau}&
x^\tau_B \parallel x^\tau_C
        \ar[d]^{\theta_B^- \parallel \theta_C^+}\\
\psi(x^\tau)
        \ar[r]_{\pr_{\tau'}}&
\psi(x^\tau)_B \parallel \psi(x^\tau)_C
}
\]
commute. We show that there are unique $y^{\tau'} \odot y^{\sigma'} \in
\confp{\tau'\odot \sigma'}$ and symmetries
\[
\omega : \varphi(x^\sigma) \sym_{\sigma'} y^{\sigma'}
\qquad
\qquad
\nu : \psi(x^\tau) \sym_{\tau'} y^{\tau'}
\]
such that $\omega_A \in \ntilde{A}$, $\nu_C \in \ptilde{C}$ and
$\omega_B \circ \theta_B^+ = \nu_B \circ \theta_B^-$.

\emph{Existence.} We get $\theta_B = \theta_B^- \circ
(\theta_B^+)^{-1} : \varphi(x^\sigma)_B \sym_B \psi(x^\tau)_B$ a
mediating symmetry between $\varphi(x^\sigma)$ and $\psi(x^\tau)$ and
from the two diagrams above we easily deduce that
\[
\xymatrix@R=8pt@C=8pt{
&x^\sigma \parallel x^{\tau}_C
	\ar[dl]_{\varphi \parallel \theta_C^+}
	\ar[rr]^{\pr_\sigma \parallel x^\tau_C}&&
x^\sigma_A \parallel x_B \parallel x^\tau_C
	\ar[rr]^{x^\sigma_A \parallel \pr_\tau^{-1}}&&
x^\sigma_A \parallel x^\tau
	\ar[dr]^{\theta_A^- \parallel \psi}\\
z^{\sigma'} \parallel z^{\tau'}_C
	\ar[rr]_{\pr_{\sigma'}\parallel z^{\tau'}_C}&&
z^{\sigma'}_A \parallel z^{\sigma'}_B \parallel z^{\tau'}_C
	\ar[rr]_{z^{\sigma'}_A \parallel \theta_B \parallel z^{\tau'}_C}&&
z^{\sigma'}_A \parallel z^{\tau'}_B \parallel z^{\tau'}_C
	\ar[rr]_{z^{\sigma'}_A \parallel \pr_\tau^{-1}}&&
z^{\sigma'}_A \parallel z^{\tau'}
}
\]
commutes, writing $z^{\sigma'} = \varphi(x^\sigma)$ and $z^{\tau'} =
\psi(x^\tau)$. As the (bijection induced by) the top row is secured and
$\varphi, \psi$ are order-isomorphisms, it follows that the (bijection
induced by) the bottom row is also secured. Therefore, applying
Proposition \ref{prop:sync_sym}, it follows that there are
\[
y^{\tau'} \odot y^{\sigma'} \in \confp{\tau'\odot \sigma'}
\qquad
\omega : z^{\sigma'} \sym_{\sigma'} y^{\sigma'}
\qquad
\nu : z^{\tau'} \sym_{\tau'} y^{\tau'}
\]
such that $\omega_A \in \ntilde{A}$, $\nu_C \in \ptilde{C}$ and $\nu_B
\circ \theta_B = \omega_B$ as required. 
For \emph{uniqueness}, if
\[
u^{\tau'} \odot u^{\sigma'} \in \confp{\tau'\odot \sigma'}
\qquad
\mu : z^{\sigma'} \sym_{\sigma'} u^{\sigma'}
\qquad
\gamma : z^{\tau'} \sym_{\tau'} u^{\tau'}
\]
then $(\mu \circ \omega^{-1}) \odot (\gamma \circ \nu^{-1}) \in
\tilde{\tau'\odot \sigma'}$ displays to a positive symmetry of $A\vdash
C$, so is an identity by Lemma 3.28 of \cite{cg2}. By Proposition
\ref{prop:char_comp}, $\mu = \omega$ and $\gamma = \nu$.

Now, we may set $(\psi \odot \varphi)(x^\tau \odot x^\sigma) = y^{\tau'}
\odot x^{\sigma'}$. To prove preservation of symmetry, we perform the
exact same construction on symmetries, using Proposition
\ref{prop:sync_sym2}, which commutes with domain and codomain.
The inverse $(\psi \odot \varphi)^{-1}$ is constructed similarly. The
fact that these are inverses and their monotonicity are direct
consequences of the uniqueness of the construction above. Finally, any
order-isomorphism preserving symmetry between ess is generated by a
unique isomorphism of ess, see \emph{e.g.} Lemma D.4 from \cite{cg3}.
\end{proof}

\subsection{Invariance of $\sim^+$-witnesses}
\label{app:invwit}
We show that the cardinality of $\sim^+$-witnesses do not depend on the
choice of representative. 

\cardinv*
\begin{proof}
We show the following. Consider $A$ a game,
$\sigma : A$ a strategy. For any $x_A \in \conf{A}$, set
$\pswit_\sigma(x_A) = \{(x^\sigma, \psi^+) \mid \psi^+ : x^\sigma_A
\sym_A x_A\}$. 
Then, for any $x_A \sym_A y_A$, we have a bijection $\pswit_\sigma(x_A)
\bij \pswit_\sigma(y_A)$. Indeed, fix $\theta : x_A \sym_A y_A$, which
can be factored uniquely as $\theta^+ \circ \theta^-$ and as
$\vartheta^- \circ \vartheta^+$ by Lemma \ref{lem:factor}.
Now, given $(x^\sigma, \psi^+) \in \pswit_\sigma(x_A)$, then there are
unique $y^\sigma \in \conf{\sigma}$, $\varphi^\sigma : x^\sigma
\sym_\sigma y^\sigma$, $\omega^+ : y^\sigma_A \sym_A^+ z_A$ s.t.
the following diagram commutes.
\[
\xymatrix{
x^\sigma_A
        \ar[r]^{\psi^+}
        \ar[d]_{\varphi^\sigma_A}&
x_A
        \ar[d]^{\theta^-}
	\ar[r]^{\vartheta^+}&
z'_A	\ar[d]^{\vartheta^-}\\
y^\sigma_A
        \ar[r]_{\omega^+}&
z_A	\ar[r]_{\theta^+} &y_A
}
\]

For \emph{existence}, refactor $\theta^- \circ \psi^+ = \Theta^+ \circ
\Theta^-$ by Lemma \ref{lem:factor}, say $\Theta^- : x^\sigma_A \sym_A^-
u_A$ and $\Theta^+ : u_A \sym_A^+ z_A$. By Lemma \ref{lem:b4}, there are
unique $\varphi^\sigma : x^\sigma \sym_\sigma y^\sigma$ and $\Omega^+ :
y^\sigma_A \sym_A^+ u_A$ such that $\Omega^+ \circ \varphi^\sigma_A =
\Theta^-$. Setting $\omega^+ = \Theta^+ \circ \Omega^+$ satisfies our
constraints.

\emph{Uniqueness.} If we also have ${y^\sigma}', {\varphi^\sigma}'$, and
${\omega^+}'$ satisfying those constraints, then $\varphi^\sigma \circ
({\varphi^\sigma}')^{-1} \in \tilde{\sigma}$ maps to the identity which
is a positive symmetry, so must be an identity by Lemma
\ref{lem:pos_symm}. It follows that $y^\sigma = {y^\sigma}'$,
$\varphi^\sigma = {\varphi^\sigma}'$, and by necessity $\omega^+ =
{\omega^+}'$ as well.

This yields a construction from $\pswit_\sigma(x_A)$ to
$\pswit_\sigma(y_A)$. Note that the construction is symmetric and may be
applied from $\pswit_\sigma(y_A)$ to $\pswit_\sigma(x_A)$ as well via
$\theta^{-1}$. That the two constructions are inverse follows
immediately from the uniqueness property.
\end{proof}

\section{Not Every Game is Representable}
\label{app:nonrep}

The following counter-example is due to Marc de Visme.

\begin{exa}
Consider the tcg $A$, with events, polarities, and causality and
follows:
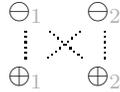
\[
\xymatrix@C=10pt@R=10pt{
\ominus_{\grey{1}}
        \ar@{.}[d]
        \ar@{.}[dr]&
\ominus_{\grey{2}}
        \ar@{.}[dl]
        \ar@{.}[d]\\
\oplus_{\grey{1}}&
\oplus_{\grey{2}}
}
\]

Its symmetry comprises all order-isomorphisms between configurations.
The negative
symmetry has all order-isomorphisms included in one of the two maximal
bijections
\[
\raisebox{15pt}{$
\xymatrix@R=10pt@C=10pt{
\ominus_{\grey{1}}
        \ar@{.}[d]
        \ar@{.}[dr]&
\ominus_{\grey{2}}
        \ar@{.}[dl]
        \ar@{.}[d]\\
\oplus_{\grey{1}}&
\oplus_{\grey{2}}
}$}
\sym^-_A
\raisebox{15pt}{$
\xymatrix@R=10pt@C=10pt{
\ominus_{\grey{1}}
        \ar@{.}[d]
        \ar@{.}[dr]&
\ominus_{\grey{2}}
        \ar@{.}[dl]
        \ar@{.}[d]\\
\oplus_{\grey{1}}&
\oplus_{\grey{2}}
}$}
\qquad
\qquad
\raisebox{15pt}{$
\xymatrix@R=10pt@C=10pt{
\ominus_{\grey{1}}
        \ar@{.}[d]
        \ar@{.}[dr]&
\ominus_{\grey{2}}
        \ar@{.}[dl]
        \ar@{.}[d]\\
\oplus_{\grey{1}}&
\oplus_{\grey{2}}
}$}
\sym^-_A
\raisebox{15pt}{$
\xymatrix@R=10pt@C=10pt{
\ominus_{\grey{2}}
        \ar@{.}[d]
        \ar@{.}[dr]&
\ominus_{\grey{1}}
        \ar@{.}[dl]
        \ar@{.}[d]\\
\oplus_{\grey{2}}&
\oplus_{\grey{1}}
}$}
\]
where again, the bijection matches those events in the corresponding
position of the
diagram. Likewise, the positive symmetry has all order-isomorphisms
included in one
of:

\[
\raisebox{15pt}{$
\xymatrix@R=10pt@C=10pt{
\ominus_{\grey{1}}
        \ar@{.}[d]
        \ar@{.}[dr]&
\ominus_{\grey{2}}
        \ar@{.}[dl]
        \ar@{.}[d]\\
\oplus_{\grey{1}}&
\oplus_{\grey{2}}
}$}
\sym^+_A
\raisebox{15pt}{$
\xymatrix@R=10pt@C=10pt{
\ominus_{\grey{1}}
        \ar@{.}[d]
        \ar@{.}[dr]&
\ominus_{\grey{2}}
        \ar@{.}[dl]
        \ar@{.}[d]\\
\oplus_{\grey{1}}&
\oplus_{\grey{2}}
}$}
\qquad
\qquad
\raisebox{15pt}{$
\xymatrix@R=10pt@C=10pt{
\ominus_{\grey{1}}
        \ar@{.}[d]
        \ar@{.}[dr]&
\ominus_{\grey{2}}
        \ar@{.}[dl]
        \ar@{.}[d]\\
\oplus_{\grey{1}}&
\oplus_{\grey{2}}
}$}
\sym^+_A
\raisebox{15pt}{$
\xymatrix@R=10pt@C=10pt{
\ominus_{\grey{1}}
        \ar@{.}[d]
        \ar@{.}[dr]&
\ominus_{\grey{2}}
        \ar@{.}[dl]
        \ar@{.}[d]\\
\oplus_{\grey{2}}&
\oplus_{\grey{1}}
}$}
\]
forming, altogether, a tcg. Then, the endosymmetry
\[
\raisebox{15pt}{$
\xymatrix@R=10pt@C=0pt{
\ominus_{\grey{1}}
        \ar@{.}[dr]&&
\ominus_{\grey{2}}
        \ar@{.}[dl]\\
&\oplus_{\grey{1}}
}$}
\qquad
\sym_A
\qquad
\raisebox{15pt}{$
\xymatrix@R=10pt@C=0pt{
\ominus_{\grey{2}}
        \ar@{.}[dr]&&
\ominus_{\grey{1}}
        \ar@{.}[dl]\\
&\oplus_{\grey{1}}
}
$}
\]
which is neither positive nor negative, uniquely factors as
\[
\raisebox{15pt}{$
\xymatrix@R=10pt@C=0pt{
\ominus_{\grey{1}}
        \ar@{.}[dr]&&
\ominus_{\grey{2}}
        \ar@{.}[dl]\\
&\oplus_{\grey{1}}
}$}
\qquad
\sym_A^-
\qquad
\raisebox{15pt}{$
\xymatrix@R=10pt@C=0pt{
\ominus_{\grey{2}}
        \ar@{.}[dr]&&
\ominus_{\grey{1}}
        \ar@{.}[dl]\\
&\oplus_{\grey{2}}
}
$}
\qquad
\sym_A^+
\qquad
\raisebox{15pt}{$
\xymatrix@R=10pt@C=0pt{
\ominus_{\grey{2}}
        \ar@{.}[dr]&&
\ominus_{\grey{1}}
        \ar@{.}[dl]\\
&\oplus_{\grey{1}}
}
$}
\]
which is not formed of endosymmetries. So this configuration is not
canonical, but
its only symmetric $\{\ominus_{\grey{1}}, \ominus_{\grey{2}},
\oplus_{\grey{2}}\}$ is
not canonical either, for  the same reason.
\end{exa}

\section{Further Content on Groupoids of Strategies}

\subsection{On Weights of Symmetry Classes}
\label{app:weights}
How should one correct the sum, if one is to count symmetry classes
instead of positive witnesses?
Let us fix $A$ a game, $\sigma : A$ any strategy,
and $\x_A \in \wconf{A}$. We show how negative symmetries act on
$\sim^+$-witnesses of $\x_A$.

\begin{prop}\label{prop:act1}
For any $(x^\sigma, \theta^+) \in \pswit_\sigma(\x_A)$ and $\varphi^-
\in \ntilde{\x_A}$ there are unique $(y^\sigma, \psi^+) \in
\pswit_\sigma(\x_A)$ and $\phi : x^\sigma \sym_\sigma y^\sigma$ such
that the following diagram commutes: 
\[
\xymatrix@R=10pt@C=10pt{
x^\sigma_A   \ar[rr]^{\theta^+}
        \ar[d]_{\phi^\sigma_A}&&
\rep{\x}_A
        \ar[d]^{\varphi^-}\\
y^\sigma_A   \ar[rr]_{\psi^+}&&
\rep{\x}_A
}
\]
%
\end{prop}
\begin{proof}
Consider $(x^\sigma, \theta^+) \in \pswit(\x_A)$ and $\varphi^- \in
\ntilde{\x_A}$. We show
that there is unique $\phi^\sigma : x^\sigma \sym_\sigma y^\sigma$ and
$\psi^+ : y^\sigma_A
\sym_A^+
\rep{\x}_A$ making the following diagram commute:
\[
\xymatrix@R=10pt@C=10pt{
x^\sigma_A   \ar[rr]^{\theta^+}
        \ar[d]_{\phi^\sigma_A}&&
\rep{\x}_A
        \ar[d]^{\varphi^-}\\
y^\sigma_A   \ar[rr]_{\psi^+}&&
\rep{\x}_A
}
\]

For existence, by Lemma \ref{lem:factor}, $\varphi^- \circ \theta^+ :
x^\sigma_A \sym_A
\rep{\x}_A$ factors uniquely as
$\Xi^+ \circ \Xi^- : x^\sigma_A \sym_A \rep{\x}_A$.
Next, by Lemma \ref{lem:b4}, there is $\phi^\sigma : x^\sigma
\sym_\sigma y^\sigma$ such
that we have
\[
\phi^\sigma_A = \Omega^+ \circ \Xi^- : x^\sigma_A \sym_A y^\sigma_A
\]
for some $\Omega^+ : y_A \sym_A^+ y^\sigma_A$. We then form $\psi^+ =
\Xi^+
\circ
(\Omega^+)^{-1}$ to conclude.

For uniqueness, if we have $\varphi_1 : x^\sigma \sym_\sigma y^\sigma$ and $\varphi_2 :
x^\sigma \sym_\sigma
z^\sigma$ satisfying the requirements,
\[
\xymatrix@C=40pt@R=15pt{
y^\sigma_A   \ar[r]^{(\sigma \varphi_1)^{-1}}
        \ar[d]_+&
x^\sigma_A   \ar[d]_+
        \ar[r]^{\sigma\varphi_2}&
z^\sigma_A   \ar[d]^+\\
\rep{\x}_A
        \ar[r]_{(\varphi^-)^{-1}}&
\rep{\x}_A
        \ar[r]_{\varphi^-}&
\rep{\x}_A
}
\]
commutes, so $(\sigma \varphi_2) \circ (\sigma \varphi_1)^{-1} = \sigma
(\varphi_2
\circ \varphi_1^{-1})$ is positive, so by Lemma 3.28 of \cite{cg2} we
have $\varphi_2
\circ \varphi_1^{-1} = \id$, so $\varphi_1 = \varphi_2$.
\end{proof}

It follows easily that there is a group action
\[
(\_ \acts \_) : \ntilde{\x_A} \times \pswit(\x_A) \to \pswit(\x_A)\,,
\]
though we shall not use this specifically.

Next, we show that representatives of symmetry classes of configurations
in $\sigma$ can always be chosen to be positively symmetric to the
chosen representative in the game.

\begin{lem}
Consider $\x^\sigma \in \swit_\sigma(\x_A)$. 
Then, there exists $x^\sigma \in \x^\sigma$ such that $\pr_\sigma
x^\sigma \sym_A^+ \rep{\x}_A$.
\end{lem}
\begin{proof}
By hypothesis, $\pr_\sigma x^\sigma = x^\sigma_A \in \x_A$, so there
exists 
\[
x^{\sigma}_A \stackrel{\theta}{\sym_A} \rep{\x}_A
\]
which factors uniquely as $x^\sigma_A
\stackrel{\theta^-}{\sym_A^-} y_A \stackrel{\theta^+}{\sym_A^+}
\rep{\x}_A$ by Lemma \ref{lem:factor}. But then, by Lemma \ref{lem:b4},
there is $\varphi : x^\sigma \sym_\sigma y^\sigma$ and $\psi^+ : y_A
\sym_A^+ y^\sigma_A$ such that  
\[
\pr_\sigma \varphi = \psi^+ \circ \theta^- : x^\sigma_A \sym_A
y^\sigma_A\,,
\]
so that in particular $y^\sigma \in \x^\sigma$ and $y^\sigma_A
\stackrel{(\psi^+)^{-1}}{\sym_A^+} y_A \stackrel{\theta^+}{\sym_A^+}
\rep{\x}_A$.
\end{proof}

So, for each $\x^\sigma \in \swit_\sigma(\x_A)$ we fix a representative
$\rep{\x}^\sigma$ such that $\rep{\x}^\sigma_A \sym_A^+ \rep{\x}_A$.
We also choose a reference $\theta_{\x^\sigma}^+ : \rep{\x}^\sigma_A
\sym_A^+ \rep{\x}_A$. Finally, for every $x^\sigma \in \x^\sigma$ we
choose $\kappa_{x^\sigma} : \rep{\x}^\sigma \sym_\sigma x^\sigma$.

Our aim is, for every symmetry class $\x^\sigma \in \swit_\sigma(\x_A)$,
count the number of concrete witnesses in $\x^\sigma$. We introduce some
notations for this set -- let us write
\begin{eqnarray*}
\wit^+_\sigma[\x^\sigma] &=& \{x^\sigma \in \wit^+_\sigma(\x_A) \mid
x^\sigma \in \x^\sigma\}\\
\pswit_\sigma[\x^\sigma] &=& 
\{x^\sigma \in \pswit_\sigma(\x_A) \mid x^\sigma \in \x^\sigma\}
\end{eqnarray*}
for the concrete witnesses (resp. $\sim^+$-witnesses) within a
symmetry class $\x^\sigma \in \swit_\sigma(\x_A)$. 

Then, we prove the following bijection:

\begin{prop}\label{prop:bijwitc}
There is a bijection
$\pswit_\sigma[\x^\sigma] \times \tilde{\x^\sigma} \bij \tilde{\x_A}$.
\end{prop}
\begin{proof}
First we show that for every $(x^\sigma, \theta^+) \in
\pswit_\sigma(\x_A)$ and $\varphi \in \tilde{\x^\sigma}$, there is a
unique $\psi \in \tilde{\x_A}$ such that the following diagram commutes:
\[
\xymatrix@R=15pt@C=15pt{
\rep{\x}^\sigma_A
	\ar[r]^{\theta_{\x^\sigma}^+}
	\ar[d]_{(\kappa_{x^\sigma} \circ \varphi)_A}&
\rep{\x}_A 
	\ar[d]^{\psi}\\
x^\sigma_A
	\ar[r]_{\theta^+}&
\rep{\x}_A
}
\]
but this is obvious, as $\psi$ is determined by composition from the
other components.

Reciprocally, we show that for all $\psi \in \tilde{\x_A}$, there are
unique $(x^\sigma, \theta^+) \in \pswit_\sigma(\x_A)$ and $\varphi \in
\tilde{\x^\sigma}$ such that the same diagram above commutes. First, by
canonicity of $\rep{\x}_A$, $\psi = \psi^+ \circ \psi^-$ for $\psi^- \in
\ntilde{\x_A}$ and $\psi^+ \in \ptilde{\x_A}$. By Proposition
\ref{prop:act1}, there are unique $(y^\sigma, \omega^+) \in
\pswit_\sigma(\x_A)$ and $\phi^\sigma : \rep{\x}^\sigma \sym_\sigma
y^\sigma$ such that the following diagram commutes:
\[
\xymatrix@R=15pt@C=15pt{
\rep{\x}^\sigma_A
	\ar[r]^{\theta^+_{\x^\sigma}}
	\ar[d]_{\phi^\sigma_A}&
\rep{\x}_A
	\ar[d]^{\psi^-}\\
y^\sigma_A
	\ar[r]_{\omega^+}&
\rep{\x}_A
}
\]

We may then define
$x^\sigma := y^\sigma$,
$\theta^+ := \psi^+ \circ \omega^+$, and
$\varphi := (\kappa_{y^\sigma})^{-1} \circ \phi^\sigma$
and the diagram is obviously satisfied. It remains to prove uniqueness,
so assume we have $(z^\sigma, \vartheta^+) \in \pswit_\sigma(\x_A)$ and
$\nu \in \tilde{\x^\sigma}$ such that the following diagram commutes:
\[
\xymatrix@R=15pt@C=15pt{
\rep{\x}^\sigma_A
        \ar[r]^{\theta^+_{\x^\sigma}}
        \ar[d]_{(\kappa_{z^\sigma} \circ \nu)_A}&
\rep{\x}_A
        \ar[d]^{\psi}\\
z^\sigma_A
        \ar[r]_{\vartheta^+}&
\rep{\x}_A
}
\]

But then $(\kappa_{z^\sigma} \circ \nu)\circ (\kappa_{x^\sigma} \circ
\varphi^{-1}) \in \tilde{\sigma}$ displays to a positive symmetry, so
must be an identity by Lemma \ref{lem:pos_symm}. Thus $x^\sigma =
y^\sigma$, so $\nu = \varphi$, and so $\theta^+ = \vartheta^+$
as it is uniquely determined from the other components by the diagram.
This gives constructions in both directions, and that they are inverses
follows directly from the uniqueness properties.
\end{proof}

From that bijection, we may conclude the following result:

\begin{thm}
Consider $A$ a game, $\sigma : A$ and $\x_A \in \wconf{A}$, and
$\x^\sigma \in \swit_\sigma(\x_A)$. Then,
\[
\sharp \wit_\sigma^+[\x^\sigma] = \frac{\sharp \ntilde{\x_A}}{\sharp
\tilde{\x^\sigma}}\,.
\]
\end{thm}
\begin{proof}
By Proposition \ref{prop:bijwitc}, we have
$\sharp \pswit_\sigma[\x^\sigma] \times \sharp \tilde{\x^\sigma} = \sharp
\tilde{\x_A}$, so we have
\[
\sharp \ptilde{\x_A} \times \sharp \wit^+_\sigma[\x^\sigma] \times
\sharp \tilde{\x^\sigma} = \sharp
\ptilde{\x_A} \times \sharp \ntilde{\x_A}
\]
using Lemma \ref{lem:pswit_wit} and canonicity of $\rep{\x}_A$. The
identity follows.
\end{proof}

This finally lets us state the collapse formula for symmetry classes --
below we use that by invariance under symmetry, any the valuation of any
$\R$-strategy lifts canonically to symmetry classes.
So we can finally reformulate \eqref{eq:qcoll} as:

\begin{thm}
Consider $\sigma \in \text{$\R$-$\Strat(A, B)$}$ and $\x_A \in
\wconf{A}$, $\x_B \in \wconf{B}$. Then,
\[
\coll(\sigma)_{\x_A, \x_B} = \sum_{\x^\sigma \in \swit_\sigma(\x_A,
\x_B) }
\frac{\sharp \ntilde{\x_A}}{\sharp \tilde{\x^\sigma}} *
\v_\sigma(\x^\sigma)
\]
\end{thm}
\begin{proof}
We calculate:
\begin{eqnarray*}
\coll(\sigma)_{\x_A, \x_B} &=&
   \sum_{x^\sigma \in \wit^+_\sigma(\x_A, \x_B)} \v_\sigma(x^\sigma)\\
&=&\sum_{\x^\sigma \in \swit_\sigma(\x_A, \x_B)}
\sum_{x^\sigma \in \wit^+_\sigma[\x^\sigma]}
\v_\sigma(x^\sigma)\\
&=& \sum_{\x^\sigma \in \swit_\sigma(\x_A, \x_B)}
\sharp \wit^+_\sigma[\x^\sigma] * \v_\sigma(\x^\sigma)\\
&=& 
\sum_{\x^\sigma \in \swit_\sigma(\x_A, \x_B)} 
\frac{\sharp \ntilde{\x_A}}{\sharp \tilde{\x^\sigma}} *
\v_\sigma(\x^\sigma)
\end{eqnarray*}
\end{proof}

\subsection{Absorption of Symmetries} As final contribution, we include
a property which, though not used for the main results of this paper,
was required for the quantum collapse of
\cite{DBLP:journals/pacmpl/ClairambaultV20}. As such, we believe it fits
with the present development.

For $A$ a game, $\sigma : A$ a strategy, and $\x_A \in \wconf{A}$,
a variant of the $\sim^+$-witnesses is
\[
\text{$\sim$-$\wit^+_\sigma(\x_A)$} = \{(x^\sigma, \theta) \mid x^\sigma
\in \wit^+_\sigma(\x_A),~\theta : x^\sigma_A \sym_A \rep{\x}_A\}\,,
\]
so we still consider witnesses $x^\sigma \in \confp{\sigma}$ such that
$x^\sigma_A \sym_A^+ \rep{\x}_A$ still, but associated with all possible
\emph{symmetries}, not only positive symmetries as in
$\pswit_\sigma(\x_A)$. Our last contribution consists in counting
$\text{$\sim$-$\wit^+_\sigma(\x_A)$}$, compared to $\pswit_\sigma(\x_A)$.
First we need:

\begin{lem}\label{lem:dec_canonical}
Consider $A$ a tcg, $\x_A \in \sconf{A}$, and $x \in \conf{A}$ s.t.
$x \sym_A^+ \rep{\x}_A$.

Then, any $\theta : x \sym_A \rep{\x}_A$ factors uniquely as $\theta^-
\circ \theta^+$, where $\theta^+ : x \sym_A^+ \rep{\x}_A$, $\theta^-
\in \ntilde{\x_A}$.
\end{lem}
\begin{proof}
Fix some $\varphi : x \sym_A^+ \rep{\x}_A$. Now, take $\theta : x \sym_A
\rep{\x}_A$.
By Lemma \ref{lem:factor}, $\theta$ factors uniquely as $\theta^-
\circ \theta^+$,
where $\theta^+ : x \sym_A^+ z$ and $\theta^- : z \sym_A^- \rep{\x}_A$
for some $z
\in \conf{A}$. But then,
\[
\varphi \circ \theta^{-1} : \rep{\x}_A \sym_A \rep{\x}_A
\]
factors via $(\varphi \circ (\theta^+)^{-1}) : z \sym_A^+ \rep{\x}_A$ and
$(\theta^-)^{-1}  : \rep{\x}_A \sym_A^- z$, so $\rep{\x}_A = z$ since
$\rep{\x}_A$ is canonical.
\end{proof}

\begin{cor}
There is a bijection
$\text{$\sim$-$\wit^+_\sigma(\x_A)$} \bij \ntilde{\x_A} \times
\pswit_\sigma(\x_A)$.
\end{cor}
\begin{proof}
First, we show that for all $(x^\sigma, \theta) \in
\text{$\sim$-$\wit^+_\sigma(\x_A)$}$ there are unique $y^\sigma \in
\confp{\sigma}, \varphi^\sigma \in \tilde{\sigma}, \theta^- \in
\ntilde{A}, \theta^+ \in \ptilde{A}$ and $\psi^+ \in \ptilde{A}$, such
that the diagram commutes:
\[
\xymatrix{
x^\sigma_A
	\ar[r]^{\theta^+}
	\ar[d]_{\varphi^\sigma_A}
	\ar[dr]^{\theta}&
\rep{\x}_A
	\ar[d]^{\theta^-}\\
y^\sigma_A
	\ar[r]_{\psi^+}&
\rep{\x}_A
}
\]

By Lemma \ref{lem:dec_canonical}, $\theta$ factors uniquely as claimed.
But then, the other components and their uniqueness follows from
Proposition \ref{prop:act1}. 
Reciprocally, we show that for all $\theta^- \in \ntilde{\x_A}$ and
$(y^\sigma, \psi^+) \in \pswit_\sigma(\x_A)$, there are unique
$x^\sigma \in \conf{\sigma}, \theta \in \tilde{A}, \theta^+ \in
\ptilde{A}$ and
$\varphi^\sigma \in \tilde{\sigma}$ such that the diagram above
commutes -- but this is again Proposition \ref{prop:act1}. 

These two constructions immediately provide the two sides of the
bijection, and that they are inverses  immediately follows from the
uniqueness.
\end{proof}

\end{document}

%% file: macros.tex
\newcommand{\ev}[1]{|#1|}
\newcommand{\R}{\mathcal{R}}
\newcommand{\N}{\mathcal{N}}
\def\profto{\!\!\!\xymatrix@C-.75pc{\ar[r]|-{\! +\!} &}\!\!\!}
\newcommand{\relto}{\profto}
\newcommand{\id}{\mathrm{id}}
\newcommand{\tensor}{\otimes}
\newcommand{\lin}{\multimap}
\newcommand{\bij}{\simeq}
\newcommand{\evm}{\mathsf{ev}}
\newcommand{\tuple}[1]{\langle #1 \rangle}
\newcommand{\M}{\mathcal{M}}
\newcommand{\der}{\epsilon}
\newcommand{\dig}{\delta}
\newcommand{\mon}{m}
\newcommand{\eval}{\Downarrow}
\newcommand{\intr}[1]{\llbracket #1 \rrbracket}
\newcommand{\conflict}{\mathrel{\#}} 
\newcommand{\imc}{\rightarrowtriangle}
\newcommand{\conf}[1]{\mathscr{C}(#1)}
\newcommand{\cov}{{{\,\mathrel-\joinrel\subset\,}}}
\renewcommand{\tilde}[1]{\mathscr{S}(#1)}
\newcommand{\ptilde}[1]{\mathscr{S}_+(#1)}
\newcommand{\ntilde}[1]{\mathscr{S}_-(#1)}
\newcommand{\tildep}[1]{\mathscr{S}^+(#1)}
\newcommand{\sym}{\cong}
\newcommand{\dom}{\mathsf{dom}}
\newcommand{\cod}{\mathsf{cod}}
\newcommand{\pol}{\mathrm{pol}}
\newcommand{\coin}{\mathbf{coin}}
\renewcommand{\qu}{\mathbf{q}}
\newcommand{\rintr}[1]{\llparenthesis #1 \rrparenthesis}
\newcommand{\gbool}{\mathbf{B}}
\newcommand{\gnat}{\mathbf{N}}
\newcommand{\gx}{\mathbf{X}}
\newcommand{\zconf}[1]{\mathscr{C}^0(#1)}
\newcommand{\x}{\mathsf{x}}
\newcommand{\y}{\mathsf{y}}
\newcommand{\z}{\mathsf{z}}
\newcommand{\sconf}[1]{\mathscr{C}_{\sym}(#1)}
\newcommand{\wconf}[1]{\mathscr{C}_{\sym}^0(#1)}
\newcommand{\nconf}[1]{\mathscr{C}^0(#1)}
\newcommand{\iso}{\cong}
\newcommand{\Ty}{\mathsf{ty}}
\newcommand{\Ctx}{\mathsf{ctx}}
\newcommand{\pr}{\partial}
\newcommand{\gcc}{\mathsf{gcc}}
\newcommand{\confp}[1]{\mathscr{C}^+(#1)}
\newcommand{\confpn}[1]{\mathscr{C}^{+,\neq \emptyset}(#1)}
\newcommand{\confn}[1]{\mathscr{C}^{\neq \emptyset}(#1)}

\newcommand{\wconfn}[1]{\mathscr{C}^{0, \neq \emptyset}_{\sym}(#1)}

\renewcommand{\v}{\mathscr{V}}
\renewcommand{\cc}{\mathsf{c}\!\mathsf{c}}
\newcommand{\simstrat}{\approx}
\newcommand{\C}{\mathcal{C}}
\newcommand{\term}{e}
\newcommand{\inj}{\mathsf{inj}}
\newcommand{\cleq}{\trianglelefteq}
\newcommand{\D}{\mathcal{D}}
\newcommand{\coll}{{\smallint\!}}
\newcommand{\wit}{\mathsf{wit}}
\newcommand{\swit}{\mathsf{wit}^{\sym}}
\newcommand{\spconf}[1]{\mathscr{C}_{\sym}^+(#1)}
\newcommand{\rep}[1]{\underline{#1}}
\newcommand{\pswit}{\text{$\sim^{\!\!+}\!\!$-$\wit$}}
\newcommand{\tto}{\Rightarrow}
\newcommand{\longto}{\longrightarrow}

\newcommand{\Sym}{\mathsf{Sym}}
\newcommand{\m}{\mathfrak{m}}

\def\acts{\curvearrowright}

\newdir{C}{%
!/4.5pt/@{(}*:(1,-.2)@{}*:(1,+.2)@_{}}

\newdir{|>}{%
!/4.5pt/@{|}*:(1,-.2)@{>}*:(1,+.2)@_{>}}

\definecolor{grey}{rgb}{.7,.7,.7}
\newcommand{\grey}[1]{{\color{grey}#1}}

\newcommand{\Rel}[1]{\text{$#1$-$\mathsf{Rel}$}}
\newcommand{\Strat}{\mathsf{Strat}}

\newcommand{\PCF}{\mathbf{PCF}}
\newcommand{\nPCF}{\mathbf{nPCF}}

\newcommand{\tbool}{\mathbb{B}}
\newcommand{\tnat}{\mathbb{N}}
\newcommand{\tx}{\mathbb{X}}
\newcommand{\ty}{\mathbb{Y}}
\newcommand{\ttrue}{\mathbf{t\!t}}
\newcommand{\tfalse}{\mathbf{f\!f}}
\newcommand{\ite}[3]{\mathbf{if}\,#1\,#2\,#3}
\newcommand{\tsucc}{\mathbf{succ}}
\newcommand{\tpred}{\mathbf{pred}}
\newcommand{\iszero}{\mathbf{iszero}}
\newcommand{\Y}{\mathcal{Y}}